\pdfoutput=1
\documentclass{article}

\usepackage[T1]{fontenc}
\usepackage{amsmath, amsthm}
\usepackage{newtxtext}
\usepackage[varvw]{newtxmath}
\usepackage[margin=1in]{geometry}

\newtheorem{theorem}{Theorem}
\newtheorem{lemma}{Lemma}
\newtheorem{corollary}{Corollary}
\newtheorem{proposition}{Proposition}
\theoremstyle{definition}

\theoremstyle{remark}

\usepackage{graphicx}
\usepackage{subcaption}

\usepackage{booktabs}

\usepackage{algorithm}
\usepackage{algpseudocode}
\algtext*{EndIf}
\algtext*{EndFor}

\usepackage{xcolor}

\usepackage{hyperref}

\usepackage{todonotes}

\newcommand{\Prob}{\Pr\nolimits}

\newcommand{\keywords}[1]{\vspace{0.5em}\noindent\textbf{Keywords:} \textit{#1}}

\newcommand{\remove}[1]{}

\title{Single-Range-Cover Search on DAGs: Average-Case Analysis, Efficient Implementations, Privacy Guarantees}

\author{Ladan Kian \and Dariusz R. Kowalski}

\begin{document}

\maketitle

\begin{center}
Augusta University, Augusta, GA, USA\\
lkian@augusta.edu \quad dkowalski@augusta.edu
\end{center}

\begin{abstract}
Extracting exact and complete information using range queries is not only time-inefficient, on average, but also susceptible to privacy breaches.
As a remedy, Demertzis et al.~\cite{demertzis2016practical} (SIGMOD'16) introduced, in the context of searchable symmetric encryption, the principle of covering a query with a single canonical range using a DAG structure (TDAG). It was later formalized by Falzon et al.~\cite{falzon2022range} (VLDB'22) as the single range cover (SRC) search algorithm -- a faster and potentially privacy-preserving search principle.
SRC-search can, however, produce false positives. In this work we study how a family of specific Directed Acyclic Graphs (DAGs) could substantially decrease the false-positive ratio, without asymptotic increase of SRC-search time and storage, compared to the classical 1D-Tree.
We perform a competitive analysis of a $c$-DAG with respect to the 1D-Tree and derive only an {\bf\em additive constant Returned Level Overhead} while reaching a substantial {\bf\em multiplicative logarithmic improvement in the Accuracy Ratio}, on average. We also provide a generic framework extending our results to non-uniform datasets whose level overhead distribution is close to that of a uniform dataset.
We define and analyze several single-path local selection rules that can substantially speed up the original SRC-search algorithm, which was based on exhaustive search, and their privacy implications. Experiments on the Gowalla and NYC Yellow Taxi datasets show that the level overhead stays at or below the proven constant, up to sampling error, and that the level overhead distribution of real data is close to a uniform baseline in most configurations.
\end{abstract}


\section{Introduction}
\label{sec:intro}

\noindent
\textbf{Motivation and Applications.}
Range queries are a widely used primitive in structured data analysis, information retrieval, and data management systems, arising in applications such as real-time analytics, large-scale data processing, multimedia data
systems~\cite{candan2010data,yamuna2001efficient}, and secure data access~\cite{guo2023privacy}. Answering them on a large dataset by inspecting every data point is prohibitive, so systems instead build a hierarchical \emph{index} over the data and resolve each query by descending it. Increasingly, moreover, the dataset and its index are outsourced to a third party, so that the index must satisfy two requirements simultaneously: it must answer each query in sub-linear time despite the scale of the data, and it must not reveal the queried data to the party performing the search.

The first requirement is met by partitioning the data hierarchically rather than scanning it, as in Bentley's KD-Tree, which recursively splits the elements at median values~\cite{bentley1975multidimensional,BentleyFriedman1979}. A range query is then resolved by descending $O(\log N)$ levels of the structure, where $N$ is the number of indexed elements. We adopt the one-dimensional KD-Tree (the \emph{1D-Tree}, formalized in Sec.~\ref{sec:DS}) as our baseline. The second requirement is met by privacy-preserving indexes, notably those of searchable symmetric encryption (SSE)~\cite{curtmola2006searchable,cash2013highly,cash2014dynamic,hore2004privacy,fugkeaw2024secure,faber2015rich}, where range queries grant access to the relevant subset of an encrypted dataset without revealing the underlying data and queries, and where the choice of index structure determines the scheme's efficiency, accuracy, and privacy~\cite{heidaripour2024organizing,chase2010structured}. Exact range search, however, is computationally expensive~\cite{BerendsohnK22} and may leak information about the data or the query~\cite{cash2015leakage,oya2021exploiting,demertzis2022attacks}. \emph{Approximate range search}, admitting a bounded number of false positives in the result, offers an attractive alternative, improving both efficiency and privacy \cite{demertzis2016practical,demertzis2018practical,falzon2022range}.

\smallskip
\noindent
\textbf{Model, Problem, and Measures.}
Among approximate range search primitives, \emph{Single Range Cover} (SRC) search covers a range query with a \emph{single} canonical index node whose range contains the query, leaving the false positives to be filtered client-side. We study SRC-search on a family of range-supporting Directed Acyclic Graphs (DAGs) over a dataset of $N$ elements drawn from a domain $\{0, 1, \ldots, M-1\}$, and compare it against SRC-search on the 1D-Tree. Our model is a client issuing a stream of range queries answered in parallel by two servers, one indexing with the 1D-Tree and one with the $c$-DAG (Fig.~\ref{fig:model}). We measure the $c$-DAG against this baseline by two per-query quantities, averaged over the stream: the additive \emph{Returned Level Overhead}, the extra depth the $c$-DAG must descend, and the
multiplicative \emph{Accuracy Ratio}, the factor by which the 1D-Tree's returned result exceeds the $c$-DAG's in extraneous elements.

\begin{figure}[h]
    \centering    \includegraphics[width=\linewidth, height= 7cm]{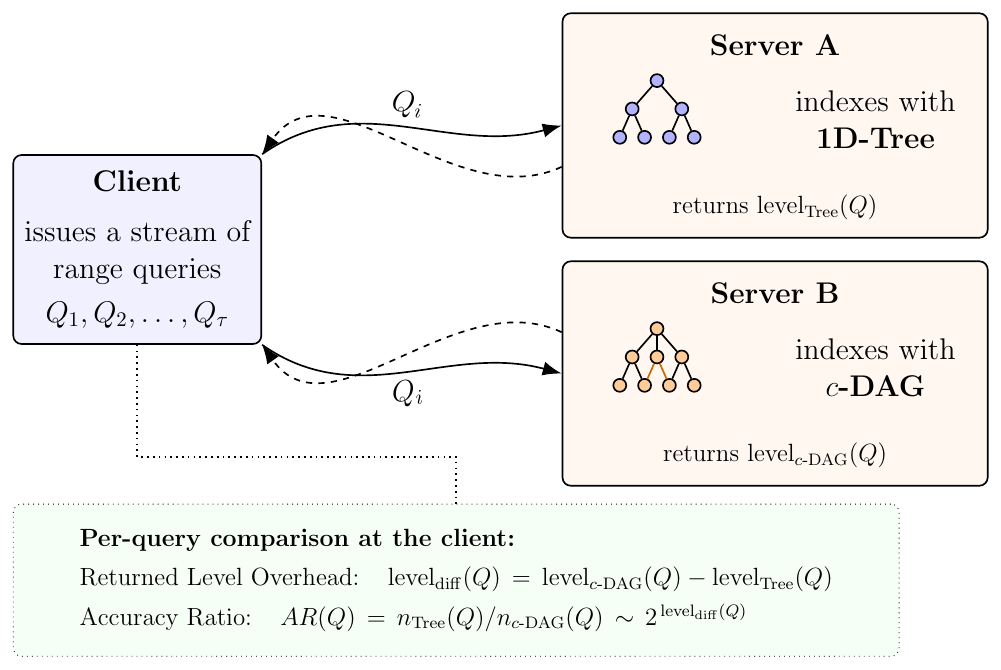}
    \caption{Our comparison model. A client issues a stream of range queries $Q_1,Q_2,\ldots,Q_\tau$, sent in parallel to two servers indexing the same dataset, one with the 1D-Tree and one with the $c$-DAG. Each server runs SRC-search and returns the level of the node it found; the client compares the two returned levels to compute the per-query Returned Level Overhead and Accuracy Ratio.}
    \label{fig:model}
\end{figure}

\smallskip
\noindent
\textbf{Prior Work and Our Positioning.}
Demertzis et al.~\cite{demertzis2016practical,demertzis2018practical}
introduced SRC-search in the context of SSE for range queries, together with the \emph{Tree-based DAG} (TDAG), which augments a binary tree with one overlapping intermediate node between every two neighboring nodes at each level, provably reducing the worst-case size of the returned canonical range from $O(M)$ (a naive single-node cover on a binary tree) to $O(s)$ for a query of length $s$. Falzon et al.~\cite{falzon2022range} subsequently formalized SRC-search as a general algorithm for range-supporting data structures and extended the construction to multiple dimensions, with analogous worst-case bounds. We continue this line of work in three directions: (1) whereas these structures partition the domain into subdivisions fixed in advance, we study \emph{data-dependent} structures, whose canonical ranges are computed from the realized elements of the dataset; (2) the TDAG is a single member of a broader family in which exactly one overlapping node is inserted between neighboring canonical ranges at every level, and we generalize it to the \emph{$c$-DAG} family, parameterized by a tunable branching factor $c=2^\alpha+1$ ($\alpha\ge1$),
of which the TDAG is the case $c=3$; and (3) whereas prior work bounds false positives in the \emph{worst case}, we carry out an
\emph{average-case, competitive} analysis, comparing executions of
SRC-search on the $c$-DAG and on the 1D-Tree over streams of randomly drawn range queries, which yields guarantees worst-case bounds cannot express.

\smallskip
\noindent
\textbf{Our Contribution.}
We introduce the data-dependent \textbf{$c$-DAG family} (Sec.~\ref{sec:DS}), built directly on the realized dataset with a tunable branching parameter $c=2^\alpha+1$ ($\alpha\ge1$), and analyze the performance of SRC-search over the $c$-DAG family relative to the 1D-Tree baseline. 
We use two complexity measures, collected over an incoming stream of range queries: the {\bf Average Returned Level Overhead} and the {\bf Average Accuracy Ratio}. 
Both $c$-DAG and 1D-Tree have the same height $\lceil \log N \rceil$ and asymptotically equal number of nodes and links $\Theta(N)$, but the $c$-DAG offers a richer poset of overlapping intervals, enabling finer canonical ranges. 
Our {\bf main complexity results}~in Thm.~\ref{thm:stream-uniform}~are:

\smallskip
\noindent
(i) a {\em small} upper bound on expected \underline{additive constant returned level overhead} of at most $2 \cdot \frac{c-2}{c-1}$ for the $c$-DAG versus the 1D-Tree, independent of $N$ and the query length $s$, and 

\smallskip
\noindent
(ii) a {\em significant} lower bound on expected \underline{multiplicative 
improvement in accuracy ratios} by a factor of $  \Omega(\log(N/s))  $ compared to the 1D-Tree.

\smallskip

We extend both these results to streams of queries with mixed lengths (Cor.~\ref{cor:stream-nonuniform}).

Our {\bf core technical tool} to prove Thm.~\ref{thm:stream-uniform} is the stochastic \textbf{Level Overhead Distribution (LOD)} between nodes returned by SRC-search on the $c$-DAG and on the 1D-Tree, for elements uniformly distributed on the domain, see Thm.~\ref{thm:lod}. It is motivated by the following observation:
analyzing each structure's returned level distribution in isolation often does not result in concise formulas, as both exhibit high variance under data skew (cf., Fig.~\ref{fig:level-dist} in Sec.~\ref{sec:experiments}).
However, structural cancellations between the two distributions stabilize their competitive difference, making the LOD analytically tractable where neither individual distribution could be (c.f.,~Thm.~\ref{thm:lod}~and~Fig.~\ref{fig:lod-gowalla}~in~Sec.~\ref{sec:experiments}).
We also apply a {\em competitive framework} that treats the 1D-Tree as a baseline and analyzes the $c$-DAG family relative to it. This competitive lens not only yields clean competitive bounds on returned level overhead and accuracy ratio, but also provides a natural \underline{machine learning} \underline{pathway} for extending the theory for uniform datasets to \underline{arbitrary} \underline{(empirical) distributions}: \ 
one may derive the actual LOD (e.g., by sampling, as many decision making algorithms do in their preliminary stage, cf.,~\cite{BabaioffIK07}) and then find best-fitting baseline distribution from the uniform case (specified in Thm.~\ref{thm:lod}).
We extend both bounds from  
uniform to arbitrary datasets, conditionally on the empirical LOD's fit to the closest uniform baseline: if the empirical LOD is $\varepsilon$-close in $L_2$ distance to a baseline LOD, the two expected measures degrade by an explicit, $\varepsilon$-dependent~term~(Cor.~\ref{cor:arbitrary-datasets},~Sec.~\ref{sec:skewed}).

We propose a faster implementation of SRC-search: \textbf{single-path SRC-search} (Algorithm~\ref{alg:src-pi}) replaces the exhaustive descent of~\cite{falzon2022range} by a descent into one child chosen by a \emph{local selection rule}, reducing the node-evaluation count by a factor of 
$\Theta\left(\left(\frac{c}{c-1}\right)^{\log N}\right)$.
We define four such rules and show that the three non-Tree rules match the level of the exhaustive $c$-DAG search exactly in the uniform model (Thm.~\ref{thm:pi-bounds}).
We also discuss the {\bf privacy implications} of SRC-search on the $c$-DAG vs. the 1D-Tree: under the Random selection rule, 
there exist queries for which the $c$-DAG induces strictly positive access-pattern entropy over the returned node -- a structural distinction from the 1D-Tree, with direct 
implications for leakage in encrypted-search deployments (Sec.~\ref{sec:privacy}).

\smallskip
\noindent
\textbf{Paper Overview.}
Sec.~\ref{sec:related_work} reviews related work. Sec.~\ref{sec:preliminaries} formalizes the range search model, the data-dependent structures, and performance measures. Sec.~\ref{sec:theo} derives the Level Overhead Distribution (LOD) and the resulting average-case bounds on the Returned Level Overhead and Accuracy Ratio. Sec.~\ref{sec:skewed} extends these bounds to arbitrary datasets, and Sec.~\ref{sec:src_rules} studies single-path SRC-search under local selection rules and Sec.~\ref{sec:privacy} analyzes their privacy, and Sec.~\ref{sec:experiments} reports our experimental evaluation on real-world datasets. Sec.~\ref{sec:conclusion} concludes.

\section{Related Work}
\label{sec:related_work}

\noindent
\textbf{Searchable encryption and range queries.}
Searchable Symmetric Encryption (SSE) enables a client to outsource an encrypted database and later retrieve, via encrypted search tokens, the records matching a query, without revealing the query or the underlying data beyond controlled leakage~\cite{song2000practical,curtmola2006searchable,chase2010structured}.
Supporting range queries efficiently is harder than single-keyword search: naively treating each domain value as its own keyword requires one token per point covered by the query, giving a query cost linear in the range length, and in the worst case in the domain size~\cite{demertzis2016practical,falzon2022range}. A standard way to achieve sub-linear range search is to partition the data hierarchically, as in Bentley's KD-Tree~\cite{bentley1975multidimensional,BentleyFriedman1979}; our baseline, the 1D-Tree, is the one-dimensional specialization of this structure. Demertzis et al.~\cite{demertzis2016practical,demertzis2018practical} were among the first to reduce a range query to a small number of token lookups over such a tree-shaped, range-covering data structure, the \emph{tree-like DAG} (TDAG), which augments the domain tree with one overlapping intermediate node between every two neighboring nodes at each level, provably reducing the worst-case size of the returned canonical range from $O(M)$ (a naive single-node cover on the plain tree) to $O(s)$ for a query
of length $s$. Falzon et al.~\cite{falzon2022range} subsequently formalized this Single Range Cover (SRC) search as a general algorithm for range-supporting data structures and extended it to multiple dimensions.
In our one-dimensional setting, under the uniform-occupancy model of Sec.~\ref{sec:theo}, their TDAG is precisely the $c=3$ member of the $c$-DAG family we study, inserting exactly one shifted, overlapping interval between neighboring canonical ranges at every level. Building on this construction, we generalize it to the tunable family of $c$-DAGs for $c=2^\alpha+1$, and replace their worst-case false-positive analysis with an average-case, competitive comparison against the 1D-Tree baseline over random query streams.

\smallskip
\noindent
\textbf{Related range-covering and approximate range-search constructions.} Best Range Cover (BRC), the exact covering rule that SRC-search relaxes, uses $O(\log s)$ nodes to cover a query of range size $s$ without false positives, but queries of the same size may receive covers of different cardinalities, so the number of returned nodes can itself reveal information about the query's position; Uniform Range Cover (URC) removes this leakage by making the number of nodes depend only on $s$~\cite{demertzis2016practical,demertzis2018practical,falzon2022range}. Faber et al.~\cite{faber2015rich} independently proposed closely related tree-covering constructions for range, substring, and wildcard queries, and Kermanshahi et al.~\cite{KermanshahiEtAl2024} give a further multi-dimensional instantiation of encrypted range search. Beyond SRC, bucketization-based approaches retrieve whole buckets of records, trading false positives for privacy~\cite{hore2004privacy,heidaripour2024organizing}, and differential-privacy-based approaches offer an alternative accuracy/privacy trade-off outside the SSE setting: adaptive hierarchical decomposition trees answer range queries under local differential privacy by dynamically restructuring the domain tree~\cite{du2021ahead}, and approximate range counting under (centralized) differential privacy achieves logarithmic error by relaxing exact boundary containment~\cite{huang2021approximate}.

\smallskip
\noindent
\textbf{Encrypted-search system model and leakage.} An end-to-end range-searchable-encryption system operates in two sequential phases: in the index-retrieval phase, the client submits a search token and receives an encrypted set of candidate record identifiers from the encrypted index; in the document-retrieval phase, the client uses these identifiers to fetch and decrypt the corresponding encrypted documents. Leakage arises across both phases, and their composition determines system-wide security~\cite{Gui2024SWiSSSE}.
This work, like the range-search literature it builds on~\cite{demertzis2016practical,demertzis2018practical,falzon2022range}, focuses exclusively on the index-retrieval phase: how the index is searched and how tightly it identifies the relevant records, leaving
document-retrieval orthogonal to our analysis. The index structures underlying SSE, including the range-covering constructions of~\cite{demertzis2016practical,demertzis2018practical,falzon2022range}, are variants of an encrypted inverted index: an encrypted map from search tokens to the identifiers of the records containing them, augmented with a pseudorandom function or hash so lookups reveal only the identifiers for the queried token~\cite{curtmola2006searchable,song2000practical}. Leakage-abuse attacks on such range-covering schemes have been studied: structure-pattern leakage alone enables database reconstruction even against response-hiding schemes, including schemes whose responses contain false positives~\cite{Markatou2023Attacks}; Kamara et al. survey and systematize this line of leakage-attack evaluation~\cite{Kamara2022LEAKER}.

\vspace*{-0.1ex}
\section{Preliminaries and Query Stream Model}
\label{sec:preliminaries}

\vspace*{-0.1ex}
In this section, we formalize the range search model underlying SRC-search on range-supporting data structures. Our goal is to compare SRC-search on the
1D-Tree and on richer overlapping DAG structures, defined in Sec.~\ref{sec:DS}; the comparison metrics -- Returned Level Overhead and Accuracy Ratio -- are defined in Sec.~\ref{sec:metrics}.

We consider a multi-dimensional database $\mathcal{DB}$, each of whose records carries one or more attributes, and we restrict our analysis to range queries over a single attribute of $\mathcal{DB}$, which we call the \emph{range search dimension}; the remaining attributes are carried by every record but are not searched.
Let $D$ denote the projection of $\mathcal{DB}$ onto the range search dimension.
We write $[M] = \{0, 1, \ldots, M-1\}$ for the domain of the range search dimension, i.e., the $M$ distinct coordinate positions it may take, and $N = |D|$ for the number of elements of $\mathcal{DB}$ realized in $[M]$ without repetition, so that $N \le M$. Sec.~\ref{sec:DS} defines data-dependent range-supporting structures such as the 1D-Tree and $c$-DAG. In Sec.~\ref{sec:theo}, we assume $N=M$, i.e., every domain position holds exactly one element of $D$ (uniform occupancy), to obtain closed-form results.
Sec.~\ref{sec:skewed} drops this assumption and extends these results to arbitrary, possibly heterogeneously distributed datasets over the domain, with $N \le M$.

Starting positions and lengths are real-valued: we embed the domain in the continuous interval $[0,M)$. A \emph{range query} is a half-open interval
$Q=[x,x+s)\subseteq[0,M)$, where $0<s \le M$ is its length and $x\in[0,M-s]$ is its starting position (for $s=M$, $\mathcal I_s$ below is the point mass at $x=0$); the elements of $D$ matching $Q$ are those in $D\cap Q$.
For a fixed $s$, let $\mathcal I_s$ be the uniform distribution on $[0,M-s]$.
To capture a server handling live traffic, a \emph{stream} is a sequence
$Q_1,\dots,Q_\tau$, $\tau\ge1$, of queries drawn independently from $\mathcal I_s$,
unless stated otherwise. In a \emph{mixed stream} there are $\sigma\ge1$ lengths
$s_1,\dots,s_\sigma$ with frequencies $\omega_j>0$, $\sum_j\omega_j=1$: each query
independently selects $s_j$ with probability $\omega_j$ and then draws its starting
position from $\mathcal I_{s_j}$.

\subsection{Range-Supporting Structures and Covering Algorithms}
\label{sec:DS}
Following~\cite{falzon2022range}, a range-supporting data structure over $[M]$ is a pair $(\mathcal{G},\mathsf{RC})$: $\mathcal{G}=(V,E)$ is a
Directed Acyclic Graph with a single source node (the root), where each node $v$ stores a canonical range $\mathrm{range}_v \subseteq [0,M)$ and is associated with the elements $D_v = D \cap \mathrm{range}_v$, $n_v=|D_v|$ (the root has $\mathrm{range}_{\mathrm{root}} = [0,M)$ and $D_{\mathrm{root}}=D$), with $\mathrm{range}_v = \bigcup_{(v,w)\in E}\mathrm{range}_w$ for every non-leaf $v$; and $\mathsf{RC}$ is a range covering rule that maps a query $Q$ to a subset $W \subseteq V$ of nodes whose canonical ranges cover $Q$, $Q\subseteq \bigcup_{w\in W}\mathrm{range}_w$.

Best Range Cover (BRC) and Uniform Range Cover (URC) are exact covering rules: for every query that is a union of canonical ranges of leaves, both return a cover $W$ with $Q = \bigcup_{w\in W}\mathrm{range}_w$, so the returned canonical ranges contain no false positives, i.e., points of the database outside of range $Q$, which must be filtered out to obtain the exact response~\cite{demertzis2016practical,demertzis2018practical,falzon2022range}.
BRC selects the minimum number of nodes achieving this exact cover; URC instead fixes the cover's cardinality $|W|$ to depend only on the query length $s$ and not on the query's position, at the cost of potentially more nodes than BRC~\cite{demertzis2016practical,demertzis2018practical}.
We instead study the single range cover (SRC) algorithm~\cite{falzon2022range}, which relaxes exactness for a single node $v$ with $Q \subseteq \mathrm{range}_v$. Since false positives are elements of $D$, SRC-search on $\mathcal{G}$ returns a \emph{deepest} node $v_{\mathcal{G}}(Q)$
among all nodes whose canonical range contains $Q$, breaking ties in favor of a node with minimum $n_v$; this set is nonempty, since the root contains every query, and all its nodes are reached by the exhaustive SRC-search of~\cite{falzon2022range} (restated as Alg.~\ref{alg:SRC} below), since every ancestor of a node containing $Q$ also contains $Q$. If several nodes remain after tie-breaking, any one of them may be returned. For the structures below with $N$ a power of two, every node at level $\ell$ is associated with exactly $N/2^{\ell}$ elements, so $v_{\mathcal{G}}(Q)$ minimizes $n_v$; under the uniform occupancy of Sec.~\ref{sec:theo} it also minimizes the length of $\mathrm{range}_v$, the objective used in~\cite{falzon2022range}.
Returning one node, instead of the exact cover returned by BRC or URC, makes the result compact, but the returned range may strictly contain $Q$: the elements of $D_v \setminus Q$ are then false positives, filtered
client-side. SRC-search is formalized in~\cite{falzon2022range}; single-path SRC-search variants with local selection rules are introduced in Sec.~\ref{sec:src_rules}.

\begin{algorithm}[t!]
\caption{SRC-search on DAGs~\cite{falzon2022range}}
\label{alg:SRC}
\begin{algorithmic}[1]
\Require A range-supporting DAG $\mathcal{G}=(V,E)$, a query $Q\subseteq[0,M)$, and a
current node $v\in V$
\Ensure The inclusion-minimal node $\mathrm{cand}\in V$ such that
$Q\subseteq\mathrm{range}_{\mathrm{cand}}$, or $\mathsf{null}$ if no such node is
reachable from $v$.
\State // Invoked with $\text{SRC-search}(\mathcal{G},Q,r)$, where $r$ is the root of $\mathcal{G}$
\State Mark $v$ as explored
\State $\mathrm{cand}\gets\mathsf{null}$
\If{$Q\subseteq\mathrm{range}_v$}
  \State $\mathrm{cand}\gets v$
  \For{each $(v,w)\in E$ such that $w$ is not yet explored}
    \State $t\gets\text{SRC-search}(\mathcal{G},Q,w)$
    \If{$t\neq\mathsf{null}$ \textbf{and} $|\mathrm{range}_t|<|\mathrm{range}_{\mathrm{cand}}|$}
      \State $\mathrm{cand}\gets t$
    \EndIf
  \EndFor
\EndIf
\State \Return $\mathrm{cand}$
\end{algorithmic}
\end{algorithm}

Alg.~\ref{alg:SRC} compares range lengths, as in~\cite{falzon2022range}. Under the uniform occupancy and dyadic assumptions of Sec.~\ref{sec:theo}, $|\mathrm{range}_v|=n_v=N/2^{\ell}$
for a node at level $\ell$, so it returns a deepest node containing $Q$, i.e., the node
$v_{\mathcal G}(Q)$ defined above; for arbitrary datasets we use the definition above,
which is also the one implemented in Sec.~\ref{sec:experiments}.

The 1D-Tree and $c$-DAG structures we study are \emph{data-dependent}: their canonical ranges are not fixed subdivisions of $[M]$ chosen in advance, but are computed from the realized elements of $D$, using the element counts $n_v$ defined above. Both structures are leveled: the root is at level $0$, every link leads from a node at level $\ell$ to a node at level $\ell+1$, and nodes at the same level that are associated with the same elements are identified, so that a node may have several parents.
Every child $w$ of a node $v$ is associated with a set $D_w$ of consecutive elements of $D_v$, and $\mathrm{range}_w$ is the half-open interval from the first element of $D_w$ to the element of $D_v$ that follows the last element of $D_w$, except that the left (resp.\ right) endpoint is that of $\mathrm{range}_v$ if $D_w$ contains the first
(resp.\ last) element of $D_v$; hence $D \cap \mathrm{range}_w = D_w$.
In both structures, recursion at a node $v$ stops, and $v$ is a leaf, once $n_v \le 1$.

\smallskip
\noindent

\textbf{The 1D-Tree.} The baseline structure is the 1D-Tree over $D$, the one-dimensional specialization of Bentley's KD-Tree~\cite{bentley1975multidimensional,BentleyFriedman1979}.
Each internal node $v$ splits $D_v$ at its median element into two children of size $\lceil n_v/2\rceil$ and $\lfloor n_v/2\rfloor$, partitioning
$\mathrm{range}_v$ with no overlap. The 1D-Tree therefore has $\Theta(N)$ nodes and links and height $\lceil \log N\rceil$. For a query $Q$, let
$\mathrm{level}_{\mathrm{Tree}}(Q)$ denote the level of the node returned by SRC-search on the 1D-Tree.

\smallskip
\noindent
\textbf{The $c$-DAG Family.} We compare the 1D-Tree with a family of overlapping, data-dependent DAGs, called $c$-DAGs, built on top of it. The
parameter $c=2^\alpha+1$ ($\alpha \ge 1$) controls the density of canonical ranges: every non-leaf node $v$
has $c$ children $w_1,\ldots,w_c$, each associated with an almost equal number of consecutive elements of $D_v$, namely $\lceil n_v/2 \rceil$ or
$\lfloor n_v/2 \rfloor$. The first and last children, $w_1$ and $w_c$, are the two Tree children of $v$: $w_1$ is associated with the first $\lceil n_v/2\rceil$ elements and $w_c$ with the remaining $\lfloor n_v/2\rfloor$ elements, so that their canonical ranges partition $\mathrm{range}_v$ with
no overlap, exactly as in the 1D-Tree. The remaining $c-2$ augmented children $w_2,\ldots,w_{c-1}$ are each associated with $\lceil n_v/2 \rceil$ consecutive elements of $D_v$ as well, but starting at rank $j\cdot\left\lfloor \frac{n_v}{2(c-1)} \right\rfloor + 1,~ j = 1, \ldots, c-2$, relative to the first element of $D_v$. These $c$ children are pairwise distinct if $n_v \ge 2(c-1)$; otherwise $\lfloor n_v/(2(c-1)) \rfloor = 0$, every augmented child coincides with $w_1$, and $v$ has only its two Tree children. The canonical ranges of $v$'s $c$ children overlap with their neighbors, and their union is exactly $\mathrm{range}_v$. Links connect each node $v$ to its $c$ children. For $c=3$, this places a single augmented child $w_2$ starting at rank $\lfloor n_v/4 \rfloor + 1$, i.e., centered on the boundary between $w_1$ and $w_3$, straddling it and overlapping both Tree children. For $c=5$, three augmented children $w_2, w_3, w_4$ are placed at starting ranks
$\lfloor n_v/8\rfloor+1$, $2\lfloor n_v/8\rfloor+1$, and $3\lfloor n_v/8\rfloor+1$, refining the coverage between $w_1$ and $w_5$ with three additional overlapping ranges.

Both structures share height $\lceil \log N\rceil$. If $N$ is a power of two, the $c$-DAG has at most $(\log_2(c-1)+2)\,N = O(cN)$ nodes and at most $c$ times as many links, $O(c^2N)$ asymptotically matching the 1D-Tree for constant $c$. Fig.~\ref{fig:1d-dag-uni} illustrates the case $c=3$ over a homogeneous dataset with $N=M=16$; the analogous heterogeneous construction is in Fig.~\ref{fig:1d-dag-hetero}. For a query $Q$, let $\mathrm{level}_{c\text{-DAG}}(Q)$ denote the level of the node returned by SRC-search on the $c$-DAG.

\begin{figure}[t!]
    \hspace*{-1em}\includegraphics[width=1.05\textwidth, height=3.2cm]{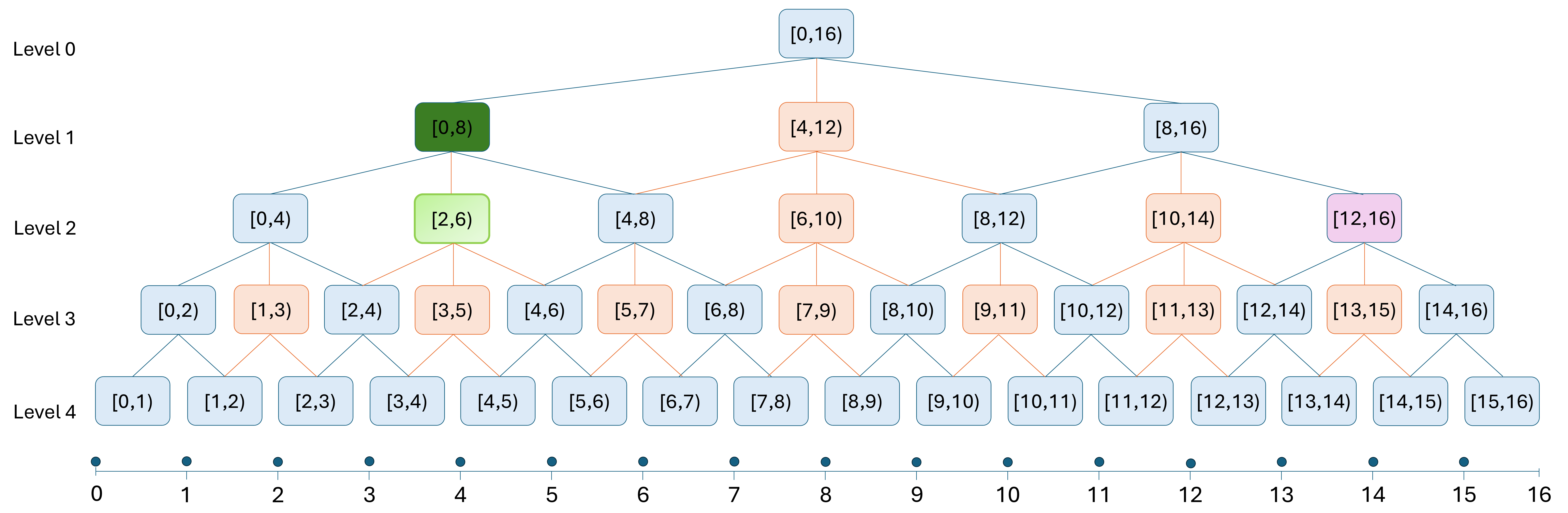}
    \caption{A 3-DAG over a homogeneous dataset with $N=M=16$. Orange intervals are the added overlapping nodes on top of the 1D-Tree intervals in blue. For $Q_1=[2,5)$, the 1D-Tree returns $[0,8)$ while the 3-DAG returns $[2,6)$. For $Q_2=[12,15)$, both return $[12,16)$.}
    \label{fig:1d-dag-uni}
    \vspace*{-3ex}
\end{figure}

 \begin{figure}[h!]
    \hspace{-1em}
\includegraphics[width=1.05\textwidth, height=3cm]{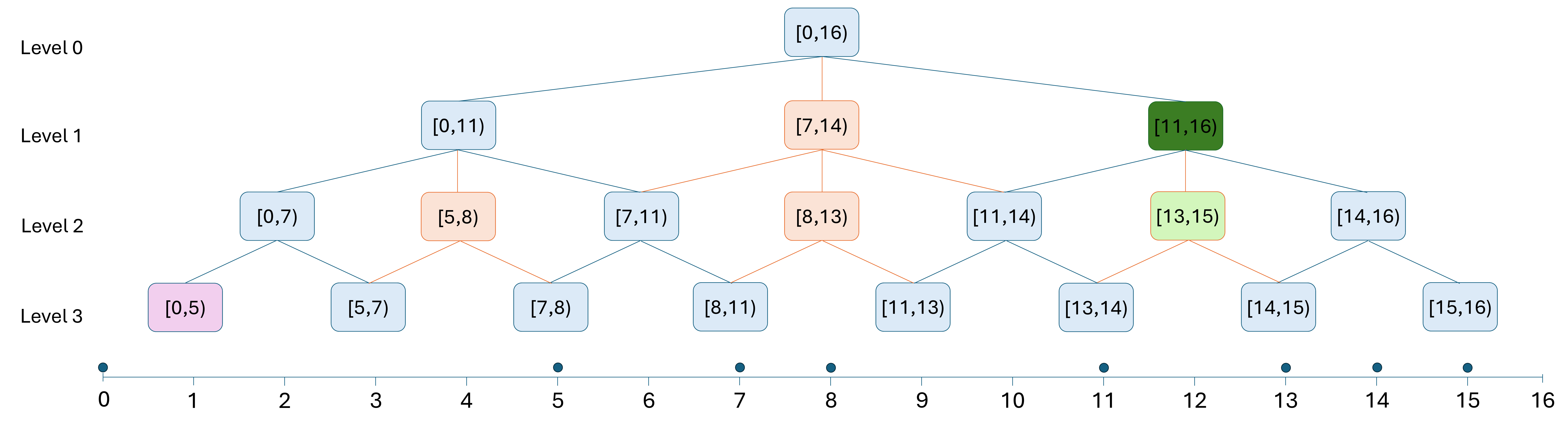}
    \caption{A 3-DAG over a heterogeneous dataset with $M=16$ and $N=8$. Orange intervals are the added overlapping nodes on top of the 1D-Tree intervals in blue. For $Q_1=[2,5)$, both return $[0,5)$. For $Q_2=[12,15)$,
    the 1D-Tree returns $[11,16)$ while the 3-DAG returns $[13,15)$.}
    \label{fig:1d-dag-hetero}
\end{figure}

\subsection{Performance Measures}
\label{sec:metrics}
Consider two servers indexing the same dataset in parallel, one using the 1D-Tree, the other the $c$-DAG, and each query $Q_i$ in the stream of queries is answered by both. To evaluate SRC-search across the $c$-DAG family (parameterized by $c\ge 3$), with the 1D-Tree as baseline, we focus on two performance measures: an additive \emph{Returned Level Overhead} and a multiplicative \emph{Accuracy Ratio}. Our competitive analysis measures the $c$-DAG's performance relative to the 1D-Tree, where a higher $c$ provides a denser family of overlapping canonical ranges, which may improve accuracy at the cost of denser data structures and deeper search paths. We first define these two measures for a single query, then generalize both to averages over a stream of queries.

\smallskip
\noindent
\textbf{Returned Level (Overhead).} For a range-supporting data structure $\mathcal{G}$ (1D-Tree or $c$-DAG) and a query $Q=[x,x+s)$, let $v_{\mathcal{G}}(Q)$ be the node returned by SRC-search on $\mathcal{G}$, and let $\mathrm{level}_{\mathcal{G}}(Q)$ denote its level in $\mathcal{G}$. We use the returned level, rather than wall-clock search time, as our cost measure because running time is implementation-dependent. On the 1D-Tree, children partition the parent's canonical range, so at most one can contain $Q$ and SRC-search follows a single, forced path. On the $c$-DAG, several of the $c$ children may validly contain $Q$, so the SRC-search of~\cite{falzon2022range} recurses into every valid child to find a deepest node containing $Q$, at a higher per-level cost than the 1D-Tree. The single-path SRC-search rule of Alg.~\ref{alg:src-pi} (Sec.~\ref{sec:src_rules}) avoids this by collecting the valid children and recursing into just one, matching the 1D-Tree's per-level cost but risking a shallower, less accurate result. The returned level abstracts away these implementation costs and isolates the structural depth we analyze theoretically in the rest of the paper. For a single query $Q$, the \emph{Returned Level Overhead} of the $c$-DAG relative to the~1D-Tree~is

\vspace*{-1.5ex}
\[
\mathrm{level}_\mathrm{diff}(Q) \;=\; \mathrm{level}_{c\text{-}\mathrm{DAG}}(Q) - \mathrm{level}_{\mathrm{Tree}}(Q) \;\ge\; 0
\ ,
\]

\vspace*{-1ex}
\noindent
reflecting how much deeper the $c$-DAG descends to return a more refined
containing canonical range; it is nonnegative because every node of the
1D-Tree is a node of the $c$-DAG at the same level, its Tree children
being $w_1$ and $w_c$.

\smallskip
\noindent
\textbf{Accuracy (Ratio).} For a data structure $\mathcal{G}$ and query $Q$, let $n_{\mathcal{G}}(Q) = |D_{v_{\mathcal{G}}(Q)}|$ be the number of elements associated with the returned node, and let $|D\cap Q|$ be the number of elements actually matching the query. For $D \cap Q \neq \emptyset$, the \emph{accuracy} of the returned node is
$A_{\mathcal{G}}(Q) = \frac{n_{\mathcal{G}}(Q)}{|D\cap Q|}$.
Since $Q\subseteq\mathrm{range}_{v_{\mathcal{G}}(Q)}$, we always have
$A_{\mathcal{G}}(Q)\ge 1$, with $A_{\mathcal{G}}(Q)=1$ exactly when the
returned node contains no false positives; larger values indicate more
extraneous elements in the result, i.e., lower accuracy. For a single
query $Q$, we define the \emph{Accuracy Ratio} of the 1D-Tree relative to
the $c$-DAG as
\[
AR(Q) \;=\; \frac{A_{\mathrm{Tree}}(Q)}{A_{c\text{-DAG}}(Q)}
\;=\; \frac{n_{\mathrm{Tree}}(Q)}{n_{c\text{-}\mathrm{DAG}}(Q)} \;\ge\; 1 ,
\]
where the second expression also defines $AR(Q)$ when $D\cap Q=\emptyset$.
Every node at level $\ell$ of either structure is associated with
$\lfloor N/2^\ell\rfloor$ or $\lceil N/2^\ell\rceil$ elements. Hence
$AR(Q)\ge1$: if $\mathrm{level}_\mathrm{diff}(Q)=0$, the node returned by
the 1D-Tree is also a deepest candidate in the $c$-DAG, so the
tie-breaking rule gives $n_{c\text{-DAG}}(Q)\le n_{\mathrm{Tree}}(Q)$;
otherwise, with $\ell=\mathrm{level}_\mathrm{Tree}(Q)$, we have
$N/2^{\ell}\ge 1$ and
$n_{c\text{-DAG}}(Q)\le\lceil N/2^{\ell+1}\rceil\le\lfloor N/2^{\ell}\rfloor
\le n_{\mathrm{Tree}}(Q)$.
If $N$ is a power of two, every split halves an even number of elements,
so every node at level $\ell$ is associated with exactly $N/2^\ell$
elements and $AR(Q) 
\;=\; 2^{\,\mathrm{level}_\mathrm{diff}(Q)}$;
for general $N$ this identity holds only up to rounding. The Accuracy Ratio is the factor by which the number of elements returned by the
1D-Tree exceeds the number returned by the $c$-DAG for the same query; a larger ratio
indicates the $c$-DAG achieves tighter query containment.

\smallskip
\noindent
\textbf{Averaging over a Stream of Queries.} Given a stream of queries $Q_1,\dots,Q_\tau$, each of length $s_i$, we average the two per-query comparative measures above over the stream. The \textbf{average Returned Level Overhead} is

\vspace*{-2ex}
\begin{equation}
\label{eq:avg-level-overhead}
\overline{\mathrm{level}_\mathrm{diff}} \;=\; \frac{1}{\tau}\sum_{i=1}^{\tau} \Big[\mathrm{level}_\mathrm{diff}(Q_i)\Big]
\ ,
\end{equation}

\vspace*{-1ex}
\noindent
and the \textbf{average Accuracy Ratio} is

\vspace*{-2ex}
\begin{equation}
\label{eq:avg-accuracy-ratio}
\overline{AR} \;=\; \frac{1}{\tau}\sum_{i=1}^{\tau} AR(Q_i) 
\ ,
\end{equation}

\vspace*{-1ex}
\noindent
Both are averages over the given query sequence; no distributional assumption on $Q_1, \ldots, Q_\tau$ is required to define them. Under the
assumptions of Section~\ref{sec:theo},
$\overline{AR} = \frac{1}{\tau} \sum_{i=1}^{\tau}
2^{\mathrm{level}_{\mathrm{diff}}(Q_i)}$.

\section{Theoretical Analysis}
\label{sec:theo}
\noindent
\textbf{Assumptions.}
Throughout this section we work under the \emph{uniform occupancy} assumption, $N = M$, i.e., every domain position holds exactly one element of $D$. Two simplifications follow, and we use them freely below. First, the canonical range of a node is in bijection with the elements it is associated with, so $n_v = |\mathrm{range}_v|$ for every node $v$, and a query $Q=[x,x+s)$ matches $|D\cap Q|\in\{\lfloor s\rfloor,\lceil s\rceil\}$ elements
(exactly $s$ if $x$ and $s$ are integers); consequently the Accuracy Ratio of Sec.~\ref{sec:metrics} may equivalently be read in terms of range lengths. Second, the median split of the 1D-Tree and the augmented children of the $c$-DAG fall at fixed domain positions, independent of the data, so both structures reduce to the data-independent constructions over $[M]$ that we analyze here. We further assume $N = 2^{n+1}$ for some integer $n \geq 0$, so that all canonical ranges are dyadic, i.e., have lengths that are powers of $2$, and both structures have height exactly $\log N$. These assumptions are relaxed in Sec.~\ref{sec:skewed}, which extends the results to arbitrary, possibly heterogeneously distributed datasets with $N \le M$.
For a fixed query length $s$, let $\kappa \;=\; \bigl\lfloor \log(N/s) \bigr\rfloor$, equivalently $2^{n-\kappa} < s \leq 2^{n-\kappa+1}$,
this is the deepest level whose canonical ranges, of length $2^{n-\kappa+1} \geq s$, can fully contain any query of length $s$.

Sec.~\ref{sec:metrics} introduced two per-query comparative measures, the Returned Level Overhead $\mathrm{level}_{\mathrm{diff}}(Q)$ and the Accuracy Ratio $AR(Q)$, together with their stream averages $\overline{\mathrm{level}_{\mathrm{diff}}}$ and $\overline{AR}$. This section derives these four quantities under the assumptions stated above. Sec.~\ref{sec:lod} first characterizes the level returned by SRC-search on the 1D-Tree and on the $c$-DAG separately (Lems.~\ref{lem:tree-level-dist} and~\ref{lem:cdag-level-dist}), and then combines the two into our central technical tool, the \emph{Level Overhead Distribution} (LOD), the closed-form probability mass function of $\mathrm{level}_{\mathrm{diff}}(Q)$ for a query drawn from $\mathcal{I}_s$ (Thm.~\ref{thm:lod}). The LOD is the enabling step of our analysis: the individual returned level distributions of the two structures are of limited use on their own, whereas structural cancellations between them make their difference analytically tractable. Sec.~\ref{sec:stream-results} then lifts this single-query result to a stream of queries, bounding $\overline{\mathrm{level}_{\mathrm{diff}}}$ and $\overline{AR}$ for streams of uniform query length (Thm.~\ref{thm:stream-uniform}) and, more generally, of non-uniform query lengths (Cor.~\ref{cor:stream-nonuniform}).

\subsection{1D-Tree vs. $c$-DAG: Level Overhead Distribution (LOD)}
\label{sec:lod}
This section derives the Level Overhead Distribution (LOD) in two steps. Lem.~\ref{lem:tree-level-dist} gives the complete level distribution of SRC-search on the 1D-Tree, and Lem.~\ref{lem:cdag-level-dist} gives the analogous distribution for the $c$-DAG, distinguishing a small-$s$ regime in which the $c$-DAG always returns level $\kappa$ from a large-$s$ regime in which it may return either $\kappa$ or $\kappa-1$. Assume $1<s<N$ (for $s=N$ only the root contains $Q$) and the \emph{non-degeneracy
condition} $c - 1 \leq 2^{\,n-\kappa+1}$, equivalently $2^{\alpha}\le 2^{n-\kappa+1}$, so that
the augmented children of every node visited by SRC-search are pairwise distinct.
Combining the two distributions, for a random query $Q = [x, x+s)$ of length $s > 1$, with $x$ drawn uniformly from $\mathcal{I}_s$, Thm.~\ref{thm:lod} obtains the LOD in closed form, i.e., the distribution of the random variable $\mathrm{level}_{\mathrm{diff}}(Q) = \mathrm{level}_{c\text{-DAG}}(Q) - \mathrm{level}_{\mathrm{Tree}}(Q)$.

\begin{lemma}[Level Distribution of SRC-search on the 1D-Tree]
\label{lem:tree-level-dist}
Consider a random query $Q = [x, x+s)$, with $x$ drawn
uniformly from $\mathcal{I}_s$.   
Then $\mathrm{level}_{\mathrm{Tree}}(Q) \in \{0, 1, \ldots, \kappa\}$,

\vspace*{-2.5ex}
\[  \Prob_{\mathcal{I}_s}\!\bigl[\mathrm{level}_{\mathrm{Tree}}(Q) = \ell\bigr]
  \;=\;
  \begin{cases}
    \dfrac{2^{\ell} \cdot s}{N-s}  & \ell \in \{0,\,1,\,\ldots,\,\kappa-1\} \ , \\[8pt]
    \dfrac{N - 2^{\kappa} s}{N-s}  & \ell = \kappa \ .
  \end{cases}
\]
\end{lemma}

\begin{proof}
The proof proceeds in five parts.

\smallskip

\noindent
\textbf{Part 1: Range of possible levels.}
At level $\ell$, every canonical range has length $2^{n-\ell+1}$.
Since $ 2^{n-\kappa} < s \leq 2^{n-\kappa+1}$, every canonical range at level $\kappa$ has length $\geq s$ and can contain $Q$.
At any level $\ell > \kappa$, the canonical ranges have length
$2^{n-\ell+1} \leq 2^{n-\kappa} < s$, so no such interval can contain $Q$.
Hence $\mathrm{level}_{\mathrm{Tree}}(Q) \in \{0,1,\ldots,\kappa\}$.

\smallskip

\noindent
\textbf{Part 2: Structural characterization for $\ell \in \{0,\ldots,\kappa-1\}$.}
At level $\ell$, the $2^\ell$ canonical ranges are
$[j \cdot 2^{n-\ell+1},\,(j+1)\cdot 2^{n-\ell+1})$ for $j = 0,\ldots,2^\ell-1$.
The two level $(\ell+1)$ children of the $j$-th level $\ell$ node
are $[j\cdot 2^{n-\ell+1},\;(2j+1)\cdot 2^{n-\ell})$ and
$[(2j+1)\cdot 2^{n-\ell},\;(j+1)\cdot 2^{n-\ell+1})$,
with split point $(2j+1)\cdot 2^{n-\ell}$ separating them. We call an odd multiple of $2^{n-\ell}$ a \emph{primitive boundary of level $\ell+1$}: it is the split point
of a level $\ell$ node and separates two nodes of level $\ell+1$ (so a primitive boundary of level $j$ is an
odd multiple of $2^{n-j+1}$). In particular $(2j+1)\cdot 2^{n-\ell}$ is a primitive boundary of level $\ell+1$.

Suppose SRC-search returns level $\ell$. Since SRC-search returns the deepest node containing $Q$, $Q$ is fully contained in some level $\ell$ interval $[j\cdot 2^{n-\ell+1},\,(j+1)\cdot 2^{n-\ell+1})$, but is not contained in either of its two level $(\ell+1)$ children.
Hence $Q$ straddles the split point $(2j+1)\cdot 2^{n-\ell}$,
i.e., $x < (2j+1)\cdot 2^{n-\ell} < x+s$.
Setting $m = 2j+1$ (odd), this gives
$x \in (m\cdot 2^{n-\ell}-s,\;m\cdot 2^{n-\ell})$.
Since $Q$ lies inside the level $\ell$ interval, it does not straddle any primitive boundary at level $\leq \ell$.

Conversely, suppose $x \in (m\cdot 2^{n-\ell}-s,\;m\cdot 2^{n-\ell})$ for
some odd $m \in \{1,3,\ldots,2^{\ell+1}-1\}$; write $m = 2j+1$.
Since $\ell \leq \kappa-1$, we have $s \leq 2^{n-\kappa+1} \leq 2^{n-\ell}$.
Therefore:
\begin{align*}
  x   &> (2j+1)\cdot 2^{n-\ell} - s \;\geq\; (2j+1)\cdot 2^{n-\ell} - 2^{n-\ell}
       \;=\; j\cdot 2^{n-\ell+1}, \\
  x+s &< (2j+1)\cdot 2^{n-\ell} + s \;\leq\; (2j+1)\cdot 2^{n-\ell} + 2^{n-\ell}
       \;=\; (j+1)\cdot 2^{n-\ell+1}.
\end{align*}
Hence $Q \subseteq [j\cdot 2^{n-\ell+1},\,(j+1)\cdot 2^{n-\ell+1})$,
so $Q$ is fully contained in the $j$-th level $\ell$ canonical range.
Since $x < (2j+1)\cdot 2^{n-\ell} < x+s$, the query $Q$ straddles the
split point between the two level $(\ell+1)$ children and is not
contained in either.
By Part~1, no node at level $\geq \ell+1$ can contain $Q$
(any such node would be a descendant of one of the two children).
Hence the level $\ell$ node is the deepest node whose canonical range
contains $Q$, and SRC-search returns it.

\smallskip

\noindent
\textbf{Part 3: Structural characterization for $\ell = \kappa$.}
If $Q$ is fully contained in some level $\kappa$ canonical range
$[m\cdot 2^{n-\kappa+1},\,(m+1)\cdot 2^{n-\kappa+1})$
for $m \in \{0,\ldots,2^\kappa-1\}$, then by Part~1 no deeper level can
contain $Q$, so SRC-search returns level $\kappa$.
The containment condition $m\cdot 2^{n-\kappa+1} \leq x$ and
$x+s < (m+1)\cdot 2^{n-\kappa+1}$ is equivalent to
\[
  x \;\in\; [m\cdot 2^{n-\kappa+1},\;(m+1)\cdot 2^{n-\kappa+1} - s).
\]

\smallskip

\noindent
\textbf{Part 4: Probability computation.}
Since $x$ is drawn uniformly from $[0, N-s)$, the probability of each level
equals the total length of the corresponding set of valid starting positions
divided by $N-s$.

\emph{Case $\ell \in \{0,\ldots,\kappa-1\}$.}
By Part~2, the valid query starting positions form the union
$\bigcup_{\substack{m\;\mathrm{odd}\\1\leq m\leq 2^{\ell+1}-1}}
(m\cdot 2^{n-\ell}-s,\;m\cdot 2^{n-\ell})$.
There are $2^\ell$ such intervals (one per odd $m$), each of length $s$.
They are pairwise disjoint: consecutive odd values $m$ and $m+2$ produce
a gap of length
\[
  2^{n-\ell+1} - s \;\geq\; 2^{n-\kappa+2} - 2^{n-\kappa+1} \;=\; 2^{n-\kappa+1} \;>\; 0,
\]
where we used $\ell \leq \kappa-1$ (giving $2^{n-\ell+1} \geq 2^{n-\kappa+2}$)
and $s \leq 2^{n-\kappa+1}$.
The total length is $2^\ell \cdot s$, giving
\[
  \Prob_{\mathcal{I}_s}[\mathrm{level}_{\mathrm{Tree}}(Q) = \ell]
  \;=\; \frac{2^\ell \cdot s}{N-s}.
\]

\emph{Case $\ell = \kappa$.}
By Part~3, the valid query starting positions form the union
$\bigcup_{m=0}^{2^\kappa-1}[m\cdot 2^{n-\kappa+1},\;(m+1)\cdot 2^{n-\kappa+1}-s)$.
There are $2^\kappa$ such intervals, each of length $2^{n-\kappa+1}-s \geq 0$.
They are pairwise disjoint: consecutive intervals are separated by gaps
of exactly $s > 0$.
The total length is $2^\kappa(2^{n-\kappa+1}-s) = N - 2^\kappa s$, giving
\[
  \Prob_{\mathcal{I}_s}[\mathrm{level}_{\mathrm{Tree}}(Q) = \kappa]
  \;=\; \frac{N - 2^\kappa s}{N-s}
  \ .
\]

\smallskip

\noindent
\textbf{Part 5: Completeness.}
The probabilities sum to $1$:
\[
  \sum_{\ell=0}^{\kappa-1}\frac{2^\ell s}{N-s}
  \;+\; \frac{N-2^\kappa s}{N-s}
  \;=\; \frac{s(2^\kappa-1) + N - 2^\kappa s}{N-s}
  \;=\; \frac{N-s}{N-s} \;=\; 1.
\]
Hence the probability is zero for all $\ell \notin \{0,\ldots,\kappa\}$.
\end{proof}

\begin{figure}[h]
\hspace*{-1em}
\includegraphics[width=1.05\linewidth]{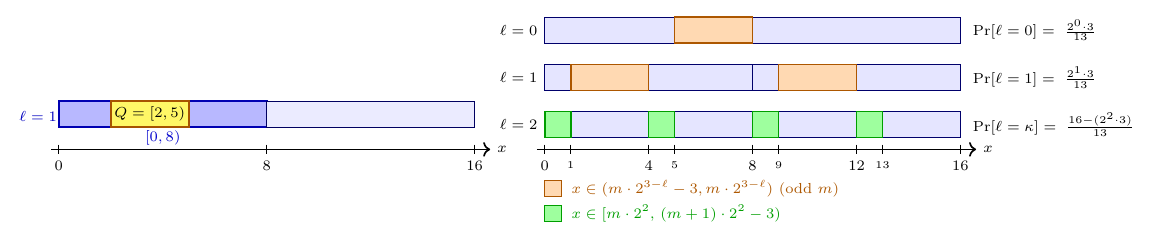}
  \caption{Illustration of Lem.~\ref{lem:tree-level-dist} for the example
$N=16$, $Q=[2,5)$, $s=3$, $\kappa=2$.
\textbf{Left:}~The query $Q=[2,5)$ straddles the primitive boundary of level 2, $1\cdot 2^{3-1}=4$, so it cannot fit in $[0,4)$ or $[4,8)$;
the deepest canonical range containing $Q$ is $[0,8)$, SRC-search returned level~$1$.
\textbf{Right:}~For each $\ell\in\{0,1,2\}$, the shaded sub-intervals are the valid starting positions $x$ for which
$\mathrm{level}_{\mathrm{Tree}}(Q)=\ell$: orange for the straddle case
($\ell<\kappa$) and green for the containment case ($\ell=\kappa$).}
\label{fig:lm1-tree}
\end{figure}

\begin{lemma}[Level Distribution of SRC-search on the $c$-DAG]
\label{lem:cdag-level-dist}
Consider a $c$-DAG with $c = 2^\alpha+1$ for  $\alpha \geq 1$, and assume the non-degeneracy condition $c-1\le 2^{n-\kappa+1}$.
Let $Q = [x, x+s)$ be a random query with $x$ drawn
uniformly from $\mathcal{I}_s$.
Then
$\mathrm{level}_{c\text{-DAG}}(Q) \in \{\kappa-1,\,\kappa\}$, and:

\smallskip
\noindent\textbf{Case 1} \; $2^{n-\kappa} < s \leq \tfrac{c-2}{c-1}\cdot 2^{n-\kappa+1}$: \ 
$  \Prob_{\mathcal{I}_s}\!\bigl[\mathrm{level}_{c\text{-DAG}}(Q) = \kappa\bigr] = 1
$.

\noindent\textbf{Case~2}\;$\tfrac{c-2}{c-1}\cdot 2^{n-\kappa+1} < s \leq 2^{n-\kappa+1}$:

\vspace*{-5ex}
\begin{align*}
 &\Prob_{\mathcal{I}_s}\!\bigl[\mathrm{level}_{c\text{-DAG}}(Q) = \ell\bigr] \quad  =
  \begin{cases}
    \dfrac{\bigl[(c-1)2^\kappa-(c-2)\bigr]\cdot\bigl(\frac{N}{2^\kappa}-s\bigr)}{N-s}
      & \ell = \kappa, \\[12pt]
    \dfrac{(c-1)(2^\kappa-1)\cdot\bigl(s-\frac{(c-2)N}{(c-1)2^\kappa}\bigr)}{N-s}
      & \ell = \kappa-1.
  \end{cases}
\end{align*}
For $c = 3$, Case~1 is vacuous and only Case~2 applies.
\end{lemma}

\begin{proof}
The proof proceeds in six parts.

\smallskip

\noindent\textbf{Part 0: Level structure.} For $0\le\ell\le n+1-\alpha$ let $L=2^{n+1-\ell}$ and
$\delta=L/(c-1)$, an integer since $c-1=2^{\alpha}$. We claim that the nodes of level $\ell$ are exactly
the intervals $[j\delta,\,j\delta+L)$, $j=0,1,\dots,(c-1)(2^{\ell}-1)$. The proof is by induction on $\ell$; for $\ell=0$ the only node is $[0,N)$. Let $\ell+1\le n+1-\alpha$, and
let $v=[a,a+L)$, $a=j\delta$, be a node of level $\ell$. Then $n_v=L\ge2^{\alpha+1}=2(c-1)$, so
$\lfloor n_v/(2(c-1))\rfloor=\delta/2$, and by Sec.~\ref{sec:DS} the children of $v$ are
$[a+i\delta/2,\,a+i\delta/2+L/2)$, $i=0,\dots,c-1$. As $a=2j\cdot\delta/2$, they start at $t\cdot\delta/2$ with
$t\in\{2j,\dots,2j+c-1\}$. Since $c\ge2$, these index sets overlap or touch for consecutive $j$, so over
$j=0,\dots,(c-1)(2^\ell-1)$ the index $t$ takes exactly the values $0,\dots,(c-1)(2^{\ell+1}-1)$, which are
the level $(\ell+1)$ nodes (length $L/2$, spacing $\delta/2$). \hfill$\square$

\noindent
\textbf{Part 1: Range of possible levels.}
At level $\ell$ in the $c$-DAG, every canonical range has length
$2^{n-\ell+1}$.
Since $2^{n-\kappa} < s \leq 2^{n-\kappa+1}$, canonical ranges at level $\kappa$ have
length $\geq s$.
At any level $\ell > \kappa$, the canonical ranges have length
$2^{n-\ell+1} \leq 2^{n-\kappa} < s$, so no such interval can contain $Q$. At level $\kappa-1$, the canonical ranges have length $2^{n-\kappa+2} > s$. For $\kappa=0$ the root contains $Q$. For $\kappa\ge1$, by Part~0 (applicable since $\kappa-1\le\kappa\le n+1-\alpha$) the starts
of the level $(\kappa-1)$ nodes are all multiples of $2^{n-\kappa+2}/(c-1)$ in $[0,N-2^{n-\kappa+2}]$. Let $a$ be the largest such start with $a\le x$. If $x<a+2^{n-\kappa+2}/(c-1)$, then, as $c\ge3$ and $s\le2^{n-\kappa+1}$, $x+s<a+2^{n-\kappa+1}+2^{n-\kappa+1}=a+2^{n-\kappa+2}$; otherwise $a$ is the last start, $a=N-2^{n-\kappa+2}$, and $x+s\le N=a+2^{n-\kappa+2}$. Hence $Q$ always fits in some level $(\kappa-1)$ node.

\smallskip
\noindent\textbf{Part 2: Structural characterization.}
At level $\kappa$, the $c$-DAG has $(c-1)2^\kappa-(c-2)$ containment intervals, indexed by $m = 0, 1, \ldots, (c-1)(2^\kappa-1)$ (by Part~0 with $\ell=\kappa$, which applies by the non-degeneracy condition).
The $m$-th containment interval is $\bigl[m\cdot\tfrac{2^{n-\kappa+1}}{c-1},\; (m+c-1)\cdot\tfrac{2^{n-\kappa+1}}{c-1}\bigr)$,
of length $2^{n-\kappa+1}$.
SRC-search returns level $\kappa$ when $Q$ is fully contained in the
$m$-th containment interval, i.e., when $m\cdot\tfrac{2^{n-\kappa+1}}{c-1} \leq x$ and $x+s \leq (m+c-1)\cdot\tfrac{2^{n-\kappa+1}}{c-1}$. Write $\delta=\frac{2^{n-\kappa+1}}{c-1}$, $L_m=m\delta$ and $R_m=(m+c-1)\delta-s$. The set of starting positions $x$ satisfying this is
\[
  x \;\in\; [L_m,\,R_m],
\]
a range of length $2^{n-\kappa+1} - s \geq 0$.
The gap between the valid starting ranges of consecutive containment
intervals $m$ and $m+1$ has length
\[
  (m+1)\cdot\frac{2^{n-\kappa+1}}{c-1}
  - \left[(m+c-1)\cdot\frac{2^{n-\kappa+1}}{c-1} - s\right]
  \;=\; s - \frac{c-2}{c-1}\cdot 2^{n-\kappa+1}.
\]
When $x$ falls in such a gap, $Q$ is not fully contained in any level
$\kappa$ containment interval; by Part~1, SRC-search then returns level
$\kappa-1$.
There are $(c-1)(2^\kappa-1)$ such gaps.

\smallskip

\noindent\textbf{Part 3: Case~1.}
When $2^{n-\kappa} < s \leq \tfrac{c-2}{c-1}\cdot 2^{n-\kappa+1}$, the gap
length $s - \tfrac{c-2}{c-1}\cdot 2^{n-\kappa+1} \leq 0$, so consecutive
valid starting ranges overlap or touch. We have $L_{m+1} \leq R_m$ for all $m$, so the ranges $[L_m, R_m]$ are contiguous.
The first range starts at $L_0 = 0$ and the last ends at
\[
  R_{(c-1)(2^\kappa-1)}
  \;=\; (c-1)\cdot 2^\kappa \cdot \frac{2^{n-\kappa+1}}{c-1} - s
  \;=\; 2^{n+1} - s \;=\; N - s.
\]
Hence the union of valid starting ranges covers all of $[0, N-s]$ and
SRC-search returns level $\kappa$ with probability $1$.
For $c = 3$: $\tfrac{c-2}{c-1}\cdot 2^{n-\kappa+1} = 2^{n-\kappa}$, so the
Case~1 condition $2^{n-\kappa} < s \leq 2^{n-\kappa}$ is empty and
Case~1 does not arise.

\smallskip

\noindent\textbf{Part 4: Case~2 -- probability computation.}
When $\tfrac{c-2}{c-1}\cdot 2^{n-\kappa+1} < s \leq 2^{n-\kappa+1}$, the
gap length is positive and the $(c-1)2^\kappa-(c-2)$ valid starting ranges
and $(c-1)(2^\kappa-1)$ gap intervals alternate over $[0, N-s]$.
Since $x$ is drawn uniformly from $[0, N-s)$, the probability of each level
equals the total length of the corresponding starting positions of a query, divided by $N-s$.

\smallskip\emph{Level $\kappa$.}
There are $(c-1)2^\kappa-(c-2)$ pairwise disjoint valid starting ranges
(gap $> 0$), each of length $2^{n-\kappa+1}-s$.
Their total length is $[(c-1)2^\kappa-(c-2)]\cdot(2^{n-\kappa+1}-s)$, giving
\[
  \Prob_{\mathcal{I}_s}[\mathrm{level}_{c\text{-DAG}}(Q) = \kappa]
  \;=\; \frac{[(c-1)2^\kappa-(c-2)]\cdot(2^{n-\kappa+1}-s)}{N-s}.
\]

\smallskip\emph{Level $\kappa-1$.}
There are $(c-1)(2^\kappa-1)$ pairwise disjoint gap intervals, each of
length $s - \tfrac{c-2}{c-1}\cdot 2^{n-\kappa+1} > 0$.
Their total length is
$(c-1)(2^\kappa-1)\cdot(s-\tfrac{c-2}{c-1}\cdot 2^{n-\kappa+1})$, giving
\[
  \Prob_{\mathcal{I}_s}[\mathrm{level}_{c\text{-DAG}}(Q) = \kappa-1]
  \;=\;
  \frac{(c-1)(2^\kappa-1)\cdot(s-\tfrac{c-2}{c-1}\cdot 2^{n-\kappa+1})}
       {N-s}.
\]

\smallskip

\noindent\textbf{Part 5: Completeness.}
The two probabilities sum to $1$:
\begin{align*}
  &[(c-1)2^\kappa-(c-2)]\cdot(2^{n-\kappa+1}-s)
   + (c-1)(2^\kappa-1)\cdot\Bigl(s-\tfrac{c-2}{c-1}\cdot 2^{n-\kappa+1}\Bigr)\\
  &= (c-1)\cdot 2^{n+1} - (c-1)2^\kappa s - (c-2)\cdot 2^{n-\kappa+1} + (c-2)s \\
  &\quad
   + (c-1)2^\kappa s - (c-1)s - (c-2)\cdot 2^{n+1} + (c-2)\cdot 2^{n-\kappa+1}\\
  &= [(c-1)-(c-2)]\cdot 2^{n+1}
   + [-(c-1)+(c-1)]\cdot 2^\kappa s \\
  &\quad
   + [-(c-2)+(c-2)]\cdot 2^{n-\kappa+1}
   + [(c-2)-(c-1)]\cdot s\\
  &= 2^{n+1} - s \;=\; N - s.
\end{align*}
Dividing by $N-s$ gives $1$, confirming the distribution is complete.
\end{proof}

\begin{figure}[h]
\hspace*{-1em}
\includegraphics[width=1.05\linewidth]{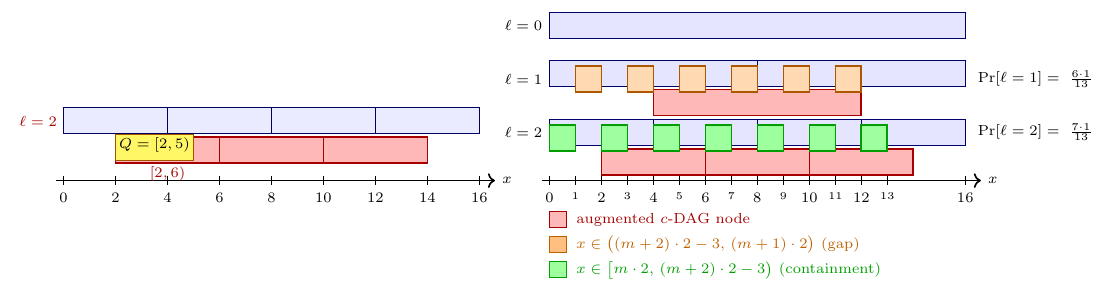}
  \caption{Illustration of Lem.~\ref{lem:cdag-level-dist} for the example
  $N=16$, $Q=[2,5)$, $s=3$, $\kappa=2$, $c=3$.
  \textbf{Left:} The query $Q=[2,5)$ starts at $x=2$, which falls in the
  containment interval $[2,3)$; the deepest canonical range of the
  $3$-DAG containing $Q$ is the augmented node $[2,6)$, so SRC-search returns level~$\kappa=2$.
  \textbf{Right:} For each $\ell\in\{0,1,2\}$, the shaded sub-intervals are the valid starting positions~$x$ for which $\mathrm{level}_{c\text{-DAG}}(Q)=\ell$: green for the containment case ($\ell=\kappa=2$) and orange for the gap case ($\ell=\kappa{-}1=1$); the $3$-DAG never returns level~$0$.}
  \label{fig:lm2-dag}
\end{figure}

\begin{theorem}[Level Overhead Distribution (LOD)]
\label{thm:lod}
Consider a $c$-DAG with $c = 2^\alpha+1$ for an integer $\alpha \geq 1$, and the 1D-Tree, both over $[M]$ with $N=2^{n+1}$.
Let $Q = [x, x+s)$ be a random query with $1<s\le N$ and $x$ drawn uniformly from $\mathcal{I}_s$, and let
$\kappa=\lfloor\log(N/s)\rfloor$. If $\kappa=0$, then $\mathrm{level}_{\mathrm{diff}}(Q)=0$ with probability $1$.
If $\kappa\ge1$ and $c-1\le 2^{\,n-\kappa+1}$, then
$\mathrm{level}_{\mathrm{diff}}(Q) \in \{0,1,\ldots,\kappa\}$~and:

\smallskip
\noindent\textbf{Case~1}
\;$2^{n-\kappa} < s \leq \tfrac{c-2}{c-1}\cdot 2^{n-\kappa+1}$:

\vspace*{-3ex}
\[
  \Prob_{\mathcal{I}_s}\!\bigl[\mathrm{level}_{\mathrm{diff}}(Q) = k\bigr]
  \;=\;
  \begin{cases}
    \dfrac{N - 2^\kappa s}{N-s} & k = 0, \\[8pt]
    \dfrac{2^{\kappa-k}\cdot s}{N-s} & k \in \{1,\ldots,\kappa\}.
  \end{cases}
\]
\noindent\textbf{Case~2}
\;$\tfrac{c-2}{c-1}\cdot 2^{n-\kappa+1} < s \leq 2^{n-\kappa+1}$:\! 

\vspace*{-3ex}
\[
  \Prob_{\mathcal{I}_s}\!\bigl[\mathrm{level}_{\mathrm{diff}}(Q) = k\bigr]\;=\;
  \begin{cases}
    \dfrac{-(c-4)\tfrac{N}{2} + (c-3)\,2^{\kappa-1}s}{N-s}
      & k = 0, \\[10pt]
    \dfrac{(c-2)\tfrac{N}{2^{k+1}} - (c-3)\,2^{\kappa-k-1}s}{N-s}
      & k \in \{1,\ldots,\kappa-1\}, \\[10pt]
    \dfrac{(c-2)\!\left(\tfrac{N}{2^\kappa}-s\right)}{N-s}
      & k = \kappa.
  \end{cases}
\]
\end{theorem}

\begin{proof} The proof proceeds in five parts.

\noindent\textbf{Step 1: Degenerate case $\kappa=0$.}
Then $s>N/2$, so every canonical range of level $\ge1$ (length $\le N/2$) is shorter than $s$ and cannot contain $Q$. Both structures return the root, hence $\mathrm{level}_{\mathrm{diff}}(Q)=0$ with probability~$1$.

\medskip
\noindent\textbf{Step 2: Range of $\mathrm{level}_{\mathrm{diff}}$.} For $\kappa\ge1$, Lem.~\ref{lem:tree-level-dist} gives
$\mathrm{level}_{\mathrm{Tree}}(Q) \in \{0,\ldots,\kappa\}$ and Lem.~\ref{lem:cdag-level-dist} gives
$\mathrm{level}_{c\text{-DAG}}(Q) \in \{\kappa-1,\kappa\}$.
Since the 1D-Tree's level~$\kappa$ intervals are also level~$\kappa$ containment
intervals of the $c$-DAG (as shown in Sub-scenario~A below), whenever the
Tree returns level~$\kappa$ the $c$-DAG also returns level~$\kappa$, ensuring $\mathrm{level}_{\mathrm{diff}}(Q) \geq 0$.
Hence $\mathrm{level}_{\mathrm{diff}}(Q) =
\mathrm{level}_{c\text{-DAG}}(Q) - \mathrm{level}_{\mathrm{Tree}}(Q) \in \{0,1,\ldots,\kappa\}$.

\medskip
\noindent\textbf{Step 3: Case~1.}
By Lem.~\ref{lem:cdag-level-dist} (Case~1), the $c$-DAG always returns
level $\kappa$, hence $\mathrm{level}_{\mathrm{diff}}(Q) = k$ if and only if
$\mathrm{level}_{\mathrm{Tree}}(Q) = \kappa - k$.
Substituting $\ell = \kappa - k$ into Lem.~\ref{lem:tree-level-dist}
gives directly:

\[
  \Prob_{\mathcal{I}_s}[\mathrm{level}_{\mathrm{diff}}(Q) = k]
  = \Prob_{\mathcal{I}_s}[\mathrm{level}_{\mathrm{Tree}}(Q) = \kappa-k]
  = 
  \begin{cases}
      \dfrac{N-2^\kappa s}{N-s} & k = 0, \\[8pt]
      \dfrac{2^{\kappa-k} s}{N-s} & k \in \{1,\ldots,\kappa\}.
    \end{cases}
\]
These are the Case~1 formulas. Completeness follows from
Lem.~\ref{lem:tree-level-dist}.

\medskip
\noindent\textbf{Step 4: Case~2.}
When $\tfrac{c-2}{c-1}\cdot 2^{n-\kappa+1} < s \leq 2^{n-\kappa+1}$, the
$c$-DAG returns level $\kappa$ or level $\kappa-1$ depending on $x$.
For each $k$, the event $\{\mathrm{level}_{\mathrm{diff}}(Q) = k\}$
decomposes into disjoint scenarios according to which level the $c$-DAG
returns.
Throughout, $x$ is uniform on $[0,N-s)$ hence all probabilities equal the
total length of the relevant starting positions divided by $N-s$.
We use the following structural facts from Lems.~\ref{lem:tree-level-dist}
and~\ref{lem:cdag-level-dist}.
\begin{itemize}
\item The 1D-Tree returns level $\kappa - j$ when:
\begin{description}
  \item[$j \in \{1,\ldots,\kappa\}$:] $x$ lies in one of $2^{\kappa-j}$ pairwise disjoint straddle intervals $(m\cdot 2^{n-\kappa+j}-s, m\cdot 2^{n-\kappa+j})$ for odd $m$, each of length $s$;
  \item[$j = 0$:] $x$ lies in one of $2^\kappa$ valid starting ranges
  $[m\cdot 2^{n-\kappa+1},\,(m+1)\cdot 2^{n-\kappa+1}-s)$, each of length $2^{n-\kappa+1}-s$.
\end{description}
\item The $c$-DAG returns level $\kappa$ when $x$ lies in a valid starting
range $[L_{m'},R_{m'}]$ of length $2^{n-\kappa+1}-s$ (from Lem.~\ref{lem:cdag-level-dist}).
\item The $c$-DAG returns level $\kappa-1$ when $x$ lies in a gap of length
$s - \tfrac{c-2}{c-1}\cdot 2^{n-\kappa+1} > 0$ between consecutive valid
starting ranges (from Lem.~\ref{lem:cdag-level-dist}).
\end{itemize}
\smallskip\noindent\textbf{Alignment claim (Case 2).}
Let $\delta=2^{n-\kappa+1}/(c-1)$, so $L_m=m\delta$, $R_m=L_m+(c-1)\delta-s$, and the gap between the valid
starting ranges $[L_m,R_m]$ and $[L_{m+1},R_{m+1}]$ is $(R_m,L_{m+1})$, of length $s-(c-2)\delta>0$. All
primitive boundaries $b$ used below are multiples of $2^{n-\kappa+1}=(c-1)\delta$ with $0<b<N$, i.e.,
$b=L_{m_0}$ with $m_0=(c-1)q$, $1\le q\le2^\kappa-1$. Then exactly $c-2$ ranges $[L_m,R_m]$ lie inside
$(b-s,b)$ (namely $m=m_0-j$, $1\le j\le c-2$) and all others are disjoint from it; exactly $c-1$ gaps
$(R_m,L_{m+1})$ lie inside $(b-s,b)$ (namely $m=m_0-j$, $1\le j\le c-1$) and all others are disjoint from it.
\emph{Proof.} $L_{m_0-j}=b-j\delta$ and $R_{m_0-j}=b+(c-1-j)\delta-s$, with $(c-2)\delta<s\le(c-1)\delta$.
For $1\le j\le c-2$: $L_{m_0-j}\ge b-(c-2)\delta>b-s$ and $R_{m_0-j}\le b+(c-2)\delta-s<b$. For $j=0$ the range
starts at $b$; for $j=c-1$ it ends at $b-s$; for $j\ge c$ it ends at $R_{m_0-j}\le b-\delta-s$; for $j<0$ it starts
at $b-j\delta\ge b+\delta$; all these miss $(b-s,b)$. The gap $(R_{m_0-j},L_{m_0-j+1})=(b+(c-1-j)\delta-s,\;
b-(j-1)\delta)$ has lower end $\ge b-s$ and upper end $\le b$ for $1\le j\le c-1$; for $j\le0$ its lower end is
$\ge b+(c-1)\delta-s\ge b$, and for $j\ge c$ its upper end is $\le b-(c-1)\delta\le b-s$. The indices are valid
because $(c-1)\delta\le b\le N-(c-1)\delta$. \hfill$\square$

\noindent\textbf{Case $k = 0$.}
$\mathrm{level}_{\mathrm{diff}}(Q) = 0$ arises in two disjoint sub-scenarios.

\textit{Sub-scenario A (both at level $\kappa$).}
By Lem.~\ref{lem:tree-level-dist}, the Tree returns level $\kappa$ when
$x \in [m\cdot 2^{n-\kappa+1},\,(m+1)\cdot 2^{n-\kappa+1}-s]$ for some $m \in \{0,\ldots,2^\kappa-1\}$ (length $2^{n-\kappa+1}-s$ each, $2^\kappa$ intervals).
The $c$-DAG, which includes the Tree's level $\kappa$ intervals as level $\kappa$
containment intervals (the binary tree intervals are exactly the first and last
children at each level of the $c$-DAG), returns level $\kappa$ for the same $x$.
Total length: $2^\kappa(2^{n-\kappa+1}-s) = N - 2^\kappa s$.

\textit{Sub-scenario B (both at level $\kappa-1$).}
By Lem.~\ref{lem:tree-level-dist}, the Tree returns level $\kappa-1$ when $x$
straddles a level-$\kappa$ primitive boundary, i.e., $x$ lies in one of
$2^{\kappa-1}$ intervals $(m\cdot 2^{n-\kappa+1}-s,\, m\cdot 2^{n-\kappa+1})$
for odd $m$, each of length $s$.
By Lem.~\ref{lem:cdag-level-dist}, the $c$-DAG returns level $\kappa-1$ when
$x$ lies in a gap interval. Around each level-$\kappa$ primitive boundary $b$, exactly $(c-1)$ gap intervals of the $c$-DAG lie entirely inside the straddle interval $(b-s,b)$ (alignment claim).
Each gap interval has length $s - \tfrac{c-2}{c-1}\cdot 2^{n-\kappa+1}$.
Total length: $2^{\kappa-1}(c-1)\bigl(s-\tfrac{c-2}{c-1}\cdot 2^{n-\kappa+1}\bigr)$.

Combining both sub-scenarios:
\begin{align*}
  \Prob[k=0]
  &= \frac{N-2^\kappa s \;+\; 2^{\kappa-1}(c-1)\bigl(s-\tfrac{c-2}{c-1}\cdot
    2^{n-\kappa+1}\bigr)}{N-s}
    \\
  & \;=\; \frac{-(c-4)\cdot 2^n + (c-3)\cdot 2^{\kappa-1}s}{N-s}
  \ .
\end{align*}

\noindent\textbf{Case $k \in \{1,\ldots,\kappa-1\}$.}
$\mathrm{level}_{\mathrm{diff}}(Q) = k$ arises in two disjoint sub-scenarios.

\textit{Sub-scenario A ($c$-DAG at $\kappa$, Tree at $\kappa-k$).}
By Lem.~\ref{lem:tree-level-dist}, the Tree returns level $\kappa-k$ when
$x$ straddles a level-$(\kappa-k+1)$ primitive boundary $b$, i.e., $x \in
(m\cdot 2^{n-\kappa+k}-s,\; m\cdot 2^{n-\kappa+k})$ for odd
$m \in \{1,\ldots,2^{\kappa-k+1}-1\}$.
There are $2^{\kappa-k}$ such straddle intervals, each of length $s$. Around each such boundary $b$, exactly $(c-2)$ valid starting ranges $[L_m,R_m]$ of the $c$-DAG lie inside $(b-s,b)$ (alignment claim), each of length $2^{n-\kappa+1}-s$.
Total length: $2^{\kappa-k}(c-2)(2^{n-\kappa+1}-s)$.

\textit{Sub-scenario B ($c$-DAG at $\kappa-1$, Tree at $\kappa-k-1$).}
The Tree returns level $\kappa-k-1$ when $x$ straddles a
level $(\kappa-k)$ primitive boundary $b'$: there are $2^{\kappa-k-1}$
such straddle intervals each of length $s$.
Around each $b'$, exactly $(c-1)$ gap intervals of the $c$-DAG lie entirely inside $(b'-s, b')$ (alignment claim).
Each has length $s-\tfrac{c-2}{c-1}\cdot 2^{n-\kappa+1}$.
Total length:
$2^{\kappa-k-1}(c-1)\bigl(s-\tfrac{c-2}{c-1}\cdot 2^{n-\kappa+1}\bigr)$.

Combining and simplifying (using $2^{n-\kappa+1} = N/2^\kappa$):
\begin{align*}
  \Prob[k]
  &= \frac{2^{\kappa-k}(c-2)(2^{n-\kappa+1}-s)
     + 2^{\kappa-k-1}(c-1)\bigl(s-\tfrac{c-2}{c-1}\cdot
     2^{n-\kappa+1}\bigr)}{N-s}\\
  &= \frac{2^{\kappa-k}(c-2)\cdot 2^{n-\kappa+1}
     - 2^{\kappa-k}(c-2)s
     + 2^{\kappa-k-1}(c-1)s
     - (c-2)\cdot 2^{\kappa-k-1}\cdot 2^{n-\kappa+1}}{N-s}\\
  &= \frac{(c-2)\cdot 2^{n-k} - (c-3)\cdot 2^{\kappa-k-1}s}{N-s}.
\end{align*}

\noindent\textbf{Case $k = \kappa$.}
This requires $\mathrm{level}_{c\text{-DAG}} = \kappa$ and
$\mathrm{level}_{\mathrm{Tree}} = 0$.
The Tree returns level $0$ (the root) when $x$ straddles the unique
level-$1$ primitive boundary at $2^n$, i.e., $x \in (2^n-s,\, 2^n)$. Exactly $(c-2)$ valid starting ranges of the $c$-DAG lie inside
$(2^n-s, 2^n)$ (alignment claim with $b=2^n=2^{\kappa-1}\cdot2^{n-\kappa+1}$), each of length $2^{n-\kappa+1}-s$.
Total length: $(c-2)(2^{n-\kappa+1}-s)$.
\[
  \Prob[k = \kappa]
  \;=\; \frac{(c-2)(2^{n-\kappa+1}-s)}{N-s}
  \;=\; \frac{(c-2)(N/2^\kappa - s)}{N-s}.
\]

\medskip
\noindent\textbf{Step 5: Completeness for Case~2.}
The probabilities for $k = 0, 1, \ldots, \kappa$ sum to $1$.
Their numerators sum~to:
\begin{align*}
  &\bigl[-(c-4)2^n + (c-3)2^{\kappa-1}s\bigr]
   + \sum_{k=1}^{\kappa-1}\bigl[(c-2)2^{n-k} - (c-3)2^{\kappa-k-1}s\bigr]
   + (c-2)(2^{n-\kappa+1}-s).
\end{align*}
Using the sums $\sum_{k=1}^{\kappa-1}2^{n-k}
= 2^{n-1}(1 - 2^{1-\kappa}) \cdot 2 = 2^n - 2^{n-\kappa+1}$
and
$\sum_{k=1}^{\kappa-1}2^{\kappa-k-1} = 2^{\kappa-1} - 1$,
this simplifies to:
\begin{align*}
  -(c-4)2^n 
&  + (c-3)2^{\kappa-1}s + (c-2)(2^n - 2^{n-\kappa+1}) - (c-3)(2^{\kappa-1}-1)s + (c-2)(2^{n-\kappa+1}-s) \\
  &= [-(c-4) + (c-2)]2^n
    + [(c-3) - (c-3)]2^{\kappa-1}s
    + [-(c-2)+(c-2)]2^{n-\kappa+1}
    \\
    & \ \ \ \
    + (c-3)s - (c-2)s \\
  &= 2\cdot 2^n - s
  \;=\; 2^{n+1} - s
  \;=\; N - s
  \ .
\end{align*}
Dividing by $N-s$ gives $1$, confirming the distribution is complete for Case~2. 
\end{proof}

Fig.~\ref{fig:ldd} illustrates the proof of
Thm.~\ref{thm:lod} for the example $N=16$, $Q=[2,5)$,
$s=3$, $\kappa=2$, $c=3$ (Case~2, since $c=3$ makes Case~1
vacuous).  Left side shows the two structures on the specific
query: the Tree returns level~$1$ via node $[0,8)$ while the
$3$-DAG returns level~$2$ via augmented node $[2,6)$, giving
$k=1$.  Right side exhibits the full partition of the valid
starting-position domain $\mathcal{I}_s=[0,13)$ by level
difference $k$: for each $k\in\{0,1,2\}$ the coloured
sub-intervals are exactly the positions identified in Steps 3
and~4 of the proof, with green marking Sub-scenario~A
($\ell_{\mathrm{DAG}}=\kappa=2$, containment) and orange
marking Sub-scenario~B ($\ell_{\mathrm{DAG}}=\kappa-1=1$, gap).

\begin{figure}[h]
\hspace*{-1em}
\includegraphics[width=1.05\linewidth]{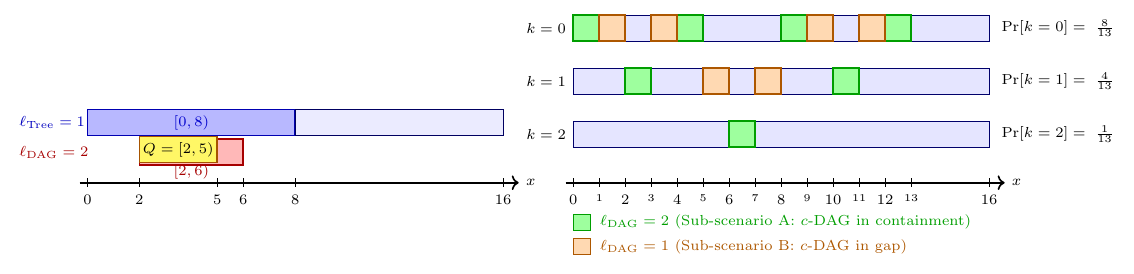}
  \caption{Illustration of Thm.~\ref{thm:lod} for the example
  $N=16$, $Q=[2,5)$, $s=3$, $\kappa=2$, $c=3$.
  \textbf{Left:} For query $Q=[2,5)$: the Tree returns level
  $\ell_{\mathrm{Tree}}=1$ (node $[0,8)$, blue) while the $3$-DAG
  returns level $\ell_{\mathrm{DAG}}=2$ (node $[2,6)$, red), giving
  $k = \ell_{\mathrm{DAG}} - \ell_{\mathrm{Tree}} = 1$.
  \textbf{Right:} For each $k\in\{0,1,2\}$, the shaded sub-intervals are the valid starting positions~$x$ for which
  $\mathrm{level}_{\mathrm{diff}}(Q)=k$: green marks Sub-scenario~A
  ($c$-DAG returns $\kappa=2$, containment) and orange marks
  Sub-scenario~B ($c$-DAG returns $\kappa{-}1=1$, gap).}
  \label{fig:ldd}
\end{figure}

\subsection{1D-Tree vs. $c$-DAG: Stream of Queries}
\label{sec:stream-results}
Sec.~\ref{sec:lod} gives the Level Overhead Distribution (LOD) for a single random
query of a fixed length $s$ (Thm.~\ref{thm:lod}). We now use it to bound the two stream-level performance measures of Sec.~\ref{sec:metrics}, the average returned level overhead $\overline{\mathrm{level}_\mathrm{diff}}$ and the average accuracy ratio
$\overline{AR}$, first for a stream of queries all of the same length, then for a stream
whose queries have different lengths.

\vspace*{-3ex}
\subsubsection{Uniform Query Sizes.}
Let $Q_1,\dots,Q_\tau$ be drawn i.i.d.\ from $\mathcal{I}_s$, all of length $s$. By linearity
of expectation, $\mathbb{E}[\overline{\mathrm{level}_\mathrm{diff}}] =
\mathbb{E}_{\mathcal{I}_s}[\mathrm{level}_\mathrm{diff}(Q)]$ and $\mathbb{E}[\overline{AR}] =
\mathbb{E}_{\mathcal{I}_s}[AR(Q)]$, thus a single-query analysis of $\mathrm{level}_\mathrm{diff}(Q)$
and $AR(Q)$ under $\mathcal{I}_s$ (via Thm.~\ref{thm:lod})~implies:

\vspace*{-1ex}
\begin{theorem}[Stream Performance, Uniform Sizes]
\label{thm:stream-uniform}
Consider a $c$-DAG with $c = 2^{\alpha}+1$ for an integer $\alpha \ge 1$, and the 1D-Tree, both over $[M]$ with $N=2^{n+1}$. Let $Q_1,\dots,Q_\tau$ be drawn i.i.d.\ from $\mathcal{I}_s$, all of the same length $s$. Let $1<s\le N$, $\kappa=\lfloor\log(N/s)\rfloor$, and, if $\kappa\ge1$, assume $c-1\le 2^{n-\kappa+1}$.

\vspace*{-2ex}
\begin{enumerate}
\item[(i)] \emph{Average returned level overhead.}
$\mathbb{E}[\overline{\mathrm{level}_\mathrm{diff}}] < 2\cdot\frac{c-2}{c-1}$, independent of $s$ and of $\tau$.
\item[(ii)] \emph{Average accuracy ratio.}
$\mathbb{E}[\overline{AR}] \ge \max\left\{1, \frac{1}{2}\left\lfloor\log\frac{N}{s}\right\rfloor\right\}$.
\end{enumerate}
\end{theorem}

\begin{proof}
\textbf{(i)} Let $\kappa = \lfloor\log(N/s)\rfloor$,
hence $2^{n-\kappa} < s \leq 2^{n-\kappa+1}$.
By Thm.~\ref{thm:lod}, $\mathrm{level}_\mathrm{diff}(Q) \in \{0,1,\ldots,\kappa\}$;
for $\kappa = 0$ this gives
$\mathbb{E}_{\mathcal{I}_s}[\mathrm{level}_\mathrm{diff}(Q)] = 0 < 2\tfrac{c-2}{c-1}$
and the bound holds trivially.
It remains to consider $\kappa \geq 1$, for which the proof
proceeds via the two cases of Thm.~\ref{thm:lod}.

\medskip\noindent\textbf{Case 1.}
$2^{n-\kappa} < s \leq \tfrac{c-2}{c-1}\cdot 2^{n-\kappa+1}$:
small-$s$ regime.

By Lem.~\ref{lem:cdag-level-dist} (Case~1), the $c$-DAG always returns
level $\kappa$.
Since the $k=0$ term contributes $0$ to the expectation, and
Thm.~\ref{thm:lod} (Case~1) gives
$\Prob_{\mathcal{I}_s}[\mathrm{level}_\mathrm{diff}(Q)=k]
= \frac{2^{\kappa-k}s}{N-s}$ for $k\in\{1,\ldots,\kappa\}$:
\[
\mathbb{E}_{\mathcal{I}_s}[\mathrm{level}_\mathrm{diff}(Q)] = \sum_{k=1}^\kappa k \cdot \frac{2^{\kappa-k} s}{N-s} = \frac{s}{N-s} \sum_{k=1}^\kappa k \cdot 2^{\kappa-k}.
\]
Applying the arithmetic-geometric formula
$\sum_{k=1}^{m}kr^{k}=r\,\frac{1-(m+1)r^m+mr^{m+1}}{(1-r)^2}$
with $r=\tfrac{1}{2}$ and $m=\kappa$:
\[
  \sum_{k=1}^{\kappa}k\!\left(\tfrac{1}{2}\right)^{\!k}
  = 2 - (\kappa+2)\,2^{-\kappa},
\]
hence $\sum_{k=1}^{\kappa}k\cdot 2^{\kappa-k} = 2^{\kappa+1}-\kappa-2$, and
$  \mathbb{E}_{\mathcal{I}_s}[\mathrm{level}_\mathrm{diff}(Q)]
  = \frac{s\,(2^{\kappa+1}-\kappa-2)}{N-s}
$.
\\
This is increasing in $s$.
Since $s \leq \tfrac{c-2}{c-1}\cdot 2^{n-\kappa+1}$, factoring
$2^{n-\kappa+1}$ from numerator and denominator
(using $N = 2^{\kappa}\cdot 2^{n-\kappa+1}$):
\[
  \mathbb{E}_{\mathcal{I}_s}[\mathrm{level}_\mathrm{diff}(Q)]
  \;\leq\;
  \frac{\dfrac{c-2}{c-1}(2^{\kappa+1}-\kappa-2)}
       {2^{\kappa}-\dfrac{c-2}{c-1}}
  \;<\; 2\cdot\frac{c-2}{c-1}\,,
\]

\medskip\noindent\textbf{Case 2.}
$\tfrac{c-2}{c-1}\cdot 2^{n-\kappa+1} < s \leq 2^{n-\kappa+1}$:
large-$s$ regime.
\\
By Thm.~\ref{thm:lod} (Case~2), for $k > 0$:
\[
\Prob_{\mathcal{I}_s}[\mathrm{level}_\mathrm{diff}(Q)=k]
=
\begin{cases}
\dfrac{(c-2)\bigl(2^{n-\kappa+1}-s\bigr)}{N-s}
& k=\kappa,\\[8pt]
\dfrac{(c-2)\,2^{n-k}-(c-3)\,2^{\kappa-k-1}s}{N-s}
& k\in\{1,\ldots,\kappa-1\}.
\end{cases}
\]
The $k=0$ term contributes $0$ to the expectation. Hence:
\begin{align*}
\mathbb{E}_{\mathcal{I}_s}[\mathrm{level}_\mathrm{diff}(Q)]
&=\sum_{k=1}^{\kappa-1}k\cdot
  \frac{(c-2)\,2^{n-k}-(c-3)\,2^{\kappa-k-1}s}{N-s}
  +\kappa\cdot\frac{(c-2)(2^{n-\kappa+1}-s)}{N-s}\\
&=\frac{1}{N-s}\!\left[
    (c-2)\sum_{k=1}^{\kappa-1}k\,2^{n-k}
    -(c-3)s\sum_{k=1}^{\kappa-1}k\,2^{\kappa-k-1}
    +\kappa(c-2)(2^{n-\kappa+1}-s)
  \right].
\end{align*}
Using $\sum_{k=1}^{m}k\,2^{n-k}=2^{n+1}-(\kappa+1)\,2^{n-\kappa+1}$
and $\sum_{k=1}^{m}k\,2^{\kappa-k-1}=2^{\kappa}-(\kappa+1)$
with $m=\kappa-1$, the numerator expands as:
\begin{align*}
&(c-2)\bigl[2^{n+1}-(\kappa+1)2^{n-\kappa+1}\bigr]
 -(c-3)s\bigl[2^{\kappa}-(\kappa+1)\bigr]
 +\kappa(c-2)2^{n-\kappa+1}-\kappa(c-2)s\\
&=(c-2)2^{n+1}-(c-2)(\kappa+1)2^{n-\kappa+1}+\kappa(c-2)2^{n-\kappa+1}\\
&\quad -(c-3)s\cdot 2^{\kappa}+(c-3)s(\kappa+1)-\kappa(c-2)s\\
&=(c-2)2^{n+1}-(c-2)2^{n-\kappa+1}-(c-3)s\cdot 2^{\kappa}+s(c-3-\kappa)\\
&=(c-2)\,2^{n-\kappa+1}(2^{\kappa}-1)+s\bigl[(3-c)(2^{\kappa}-1)-\kappa\bigr].
\end{align*}

The expectation has the form $\frac{A+Bs}{N-s}$ with $A=(c-2)2^{n-\kappa+1}(2^\kappa-1)$ and
$B=(3-c)(2^\kappa-1)-\kappa$, and its derivative in $s$ has the sign of $A+BN$. Using $N=2^\kappa2^{n-\kappa+1}$,
$A+BN=N\bigl[(1-2^{-\kappa})\bigl((3-c)2^\kappa+c-2\bigr)-\kappa\bigr]<0$: for $c=3$ the bracket equals
$1-2^{-\kappa}-\kappa<0$, and for $c\ge5$, $(3-c)2^\kappa+c-2\le4-c<0$. Hence $\mathbb{E}_{\mathcal{I}_s}[\mathrm{level}_\mathrm{diff}(Q)]$ is decreasing in $s$. Thus, for all $s$ in Case~2:
\begin{align*}
\mathbb{E}_{\mathcal{I}_s}[\mathrm{level}_\mathrm{diff}(Q)]
&<\frac{(c-2)\,2^{n-\kappa+1}(2^{\kappa}-1)
        +\dfrac{c-2}{c-1}\cdot 2^{n-\kappa+1}
         \bigl[(3-c)(2^{\kappa}-1)-\kappa\bigr]}
       {N-\dfrac{c-2}{c-1}\cdot 2^{n-\kappa+1}}\\
&=\frac{(c-2)\,2^{n-\kappa+1}(2^{\kappa}-1)
        \!\left[1+\dfrac{3-c}{c-1}\right]
        -\dfrac{c-2}{c-1}\cdot\kappa\cdot 2^{n-\kappa+1}}
       {2^{n-\kappa+1}\!\left(2^{\kappa}-\dfrac{c-2}{c-1}\right)}\\
&=\frac{\dfrac{c-2}{c-1}\cdot 2(2^{\kappa}-1)
        -\dfrac{c-2}{c-1}\cdot\kappa}
       {2^{\kappa}-\dfrac{c-2}{c-1}}
 =\frac{\dfrac{c-2}{c-1}(2^{\kappa+1}-\kappa-2)}
       {2^{\kappa}-\dfrac{c-2}{c-1}}
 \;<\; 2\cdot\frac{c-2}{c-1}.
\end{align*}
This proves part (i).

\bigskip
\noindent\textbf{(ii)} Since $AR(Q) = 2^{\mathrm{level}_\mathrm{diff}(Q)}$ (Sec.~\ref{sec:metrics}), and the
probabilities $\Prob_{\mathcal{I}_s}[\mathrm{level}_\mathrm{diff}(Q) = k]$
are given by Thm.~\ref{thm:lod}, we have:

\begin{equation}
\label{eq:AR-sum}
\mathbb{E}_{\mathcal{I}_s}[AR(Q)]
= \sum_{k=0}^{\kappa} 2^k\,
  \Prob_{\mathcal{I}_s}[\mathrm{level}_\mathrm{diff}(Q) = k]
  \ .
\end{equation}

Let $\kappa = \lfloor\log(N/s)\rfloor$, so $2^{n-\kappa} < s \leq 2^{n-\kappa+1}$.
Since $\mathrm{level}_\mathrm{diff}(Q) \geq 0$ implies
$AR(Q) \geq 1$, and for $\kappa = 0$,
Thm.~\ref{thm:lod} gives $\mathrm{level}_\mathrm{diff}(Q) = 0$ with probability $1$, hence $AR(Q) = 1$ and $\mathbb{E}_{\mathcal{I}_s}[AR(Q)] = 1 = \max\{1,0\}$ trivially,
it remains to consider $\kappa \geq 1$.
The proof again proceeds via the two cases of Thm.~\ref{thm:lod}.

\medskip\noindent\textit{Case 1.}
$2^{n-\kappa} < s \leq \tfrac{c-2}{c-1}\cdot 2^{n-\kappa+1}$:
small-$s$ regime.

By Thm~\ref{thm:lod} (Case~1):
\[
\Prob_{\mathcal{I}_s}[\mathrm{level}_\mathrm{diff}(Q) = 0]
= \frac{N - 2^\kappa s}{N-s},
\qquad
\Prob_{\mathcal{I}_s}[\mathrm{level}_\mathrm{diff}(Q) = k]
= \frac{2^{\kappa-k}\,s}{N-s},
\quad k \in \{1,\ldots,\kappa\}.
\]
For each $k \in \{1,\ldots,\kappa\}$:
$2^k\,\Prob_{\mathcal{I}_s}[\mathrm{level}_\mathrm{diff}(Q) = k]
= \frac{2^\kappa s}{N-s}$,
independent of $k$. Hence by Eq.~(\ref{eq:AR-sum}):
\[
\mathbb{E}_{\mathcal{I}_s}[AR(Q)]
= \frac{N - 2^\kappa s + \kappa\cdot 2^\kappa s}{N-s}
= \frac{N + (\kappa-1)\,2^\kappa s}{N-s}
> \frac{N + (\kappa-1)\,2^\kappa s}{N}
\ .
\]
From the definition of $\kappa$, $2^{\kappa+1} > N/s$; hence $2^\kappa s > N/2$.
Therefore:
\[
\mathbb{E}_{\mathcal{I}_s}[AR(Q)]
> 1 + \frac{(\kappa-1)\,2^\kappa s}{N}
> 1 + \frac{\kappa-1}{2}
= \frac{\kappa+1}{2}
\geq \frac{1}{2}\left\lfloor\log\frac{N}{s}\right\rfloor.
\]

\medskip\noindent\textit{Case 2.}
$\tfrac{c-2}{c-1}\cdot 2^{n-\kappa+1} < s \leq 2^{n-\kappa+1}$:
large-$s$ regime.

By Thm.~\ref{thm:lod} (Case~2):
\begin{align*}
\Prob_{\mathcal{I}_s}[\mathrm{level}_\mathrm{diff}(Q) = 0]
  &= \frac{-(c-4)\,2^n + (c-3)\,2^{\kappa-1}s}{N-s},\\
\Prob_{\mathcal{I}_s}[\mathrm{level}_\mathrm{diff}(Q) = k]
  &= \frac{(c-2)\,2^{n-k} - (c-3)\,2^{\kappa-k-1}s}{N-s},
  \quad k \in \{1,\ldots,\kappa-1\},\\
\Prob_{\mathcal{I}_s}[\mathrm{level}_\mathrm{diff}(Q) = \kappa]
  &= \frac{(c-2)\bigl(2^{n-\kappa+1}-s\bigr)}{N-s}.
\end{align*}
We compute each contribution to Eq.~(\ref{eq:AR-sum}).

\underline{For $k=0$:}
$2^0\,\Prob_{\mathcal{I}_s}[\mathrm{level}_\mathrm{diff}(Q)=0]
= \dfrac{-(c-4)\,2^n + (c-3)\,2^{\kappa-1}s}{N-s}$.

\underline{For $k \in \{1,\ldots,\kappa-1\}$:}
$2^k\,\Prob_{\mathcal{I}_s}[\mathrm{level}_\mathrm{diff}(Q) = k]
= \dfrac{(c-2)\,2^n - (c-3)\,2^{\kappa-1}s}{N-s}$,
independent of $k$ ($\kappa-1$ such terms).

\underline{For $k=\kappa$:}
$2^\kappa\,\Prob_{\mathcal{I}_s}[\mathrm{level}_\mathrm{diff}(Q) = \kappa]
= \dfrac{(c-2)\,2^{n+1} - (c-2)\,2^\kappa s}{N-s}$.

\smallskip
Summing up all the above contributions, the numerator of
$\mathbb{E}_{\mathcal{I}_s}[AR(Q)]$ equals:
\begin{align*}
    -(c-4) & \,2^{n} + (c-3)\,2^{\kappa-1}\,s
+ (\kappa-1)\bigl((c-2)\,2^{n} - (c-3)\,2^{\kappa-1}\,s\bigr)
+ (c-2)\,2^{n+1} - (c-2)\,2^{\kappa}\,s
\\
&
= \bigl((c-2)\kappa + 2\bigr)\,2^{n}
  - \bigl((c-3)\kappa + 2\bigr) 2^{\kappa-1}\,s\ .
\end{align*}
Hence:
\begin{align*}
\mathbb{E}_{\mathcal{I}_s}[AR(Q)]
&= \frac{\bigl((c-2)\kappa+2\bigr)\,2^n - \bigl((c-3)\kappa+2\bigr)\,2^{\kappa-1}s}
       {N-s}
> \frac{\bigl((c-2)\kappa+2\bigr)\,2^n - \bigl((c-3)\kappa+2\bigr)\,2^{\kappa-1}s}
       {N} \ .
\end{align*}
Since $2^\kappa \leq N/s$ (from the definition of $\kappa$), we have
$2^{\kappa-1}s \leq 2^n$. Therefore:
\[
\mathbb{E}_{\mathcal{I}_s}[AR(Q)]
> \frac{\bigl((c-2)\kappa+2\bigr)\,2^n - \bigl((c-3)\kappa+2\bigr)\,2^n}{N}
= \frac{\kappa\cdot 2^n}{N}
= \frac{\kappa}{2}
= \frac{1}{2}\left\lfloor\log\frac{N}{s}\right\rfloor.
\]
To summarize, in both cases, $\mathbb{E}_{\mathcal{I}_s}[AR(Q)] \geq
\frac{1}{2}\lfloor\log(N/s)\rfloor$, and since $AR(Q) \geq 1$ for all $Q$,
the bound $\mathbb{E}_{\mathcal{I}_s}[AR(Q)] \geq \max \{1,\frac{1}{2}\lfloor\log(N/s)\rfloor\}$ is established.

\end{proof}

\vspace*{-4ex}
\subsubsection{Non-Uniform Query Sizes.}

In practice a query stream need not use a single fixed length. We consider the mixed stream of Sec.~\ref{sec:preliminaries} with equal frequencies
$\omega_1=\dots=\omega_\sigma=1/\sigma$: each query's length is drawn uniformly from $s_1,\dots,s_\sigma$, and its position is then drawn from $\mathcal{I}_{s_i}$ for the chosen size $s_i$.

\begin{corollary}[Stream Performance, Non-Uniform Sizes]
\label{cor:stream-nonuniform}

\begin{enumerate}
\item[(i)] \emph{Average returned level overhead.} Thm.~\ref{thm:stream-uniform}(i) holds verbatim for \emph{any} sequence of query lengths $s_1,\dots,s_\tau>1$ (each satisfying the non-degeneracy condition of Thm.~\ref{thm:stream-uniform}), not necessarily equal and not necessarily drawn from a common distribution: since the bound $2\cdot\frac{c-2}{c-1}$ holds for every fixed $s>1$ individually, it bounds the average over any mix of sizes as well, with no assumption on how sizes are distributed.
\item[(ii)] \emph{Average accuracy ratio.} If every $s_i$ satisfies $1<s_i\le N$ and the non-degeneracy condition, then
\begin{align*}
    \mathbb{E}[\overline{AR}] & \;\ge\; \frac{1}{\sigma}\sum_{i=1}^{\sigma} \max\left\{1, \frac{1}{2}\left\lfloor\log\frac{N}{s_i}\right\rfloor\right\} \;\ge\; \frac{1}{2}\left(\log N \;-\; \frac{1}{\sigma}\sum_{i=1}^{\sigma}\log s_i\right) \;-\; \frac{1}{2}.
\end{align*}
\end{enumerate}
\end{corollary}

\begin{proof}
\textbf{(i)} By Thm.~\ref{thm:stream-uniform} (i), $\mathbb{E}_{\mathcal{I}_{s}}[\mathrm{level}_\mathrm{diff}(Q)] < 2\cdot\frac{c-2}{c-1}$ for every fixed $s>1$; in particular $\mathbb{E}_{\mathcal{I}_{s_i}}[\mathrm{level}_\mathrm{diff}(Q)] \le 2\cdot\frac{c-2}{c-1}$ for every $i\in\{1,\ldots,\sigma\}$. Since each query's length is drawn uniformly from $\{s_1,\ldots,s_\sigma\}$ and its position uniformly from $\mathcal{I}_{s_i}$ conditioned on the chosen length $s_i$, linearity of expectation over this two-stage sampling gives
\[
\mathbb{E}[\overline{\mathrm{level}_\mathrm{diff}}] = \frac{1}{\sigma}\sum_{i=1}^\sigma \mathbb{E}_{\mathcal{I}_{s_i}}[\mathrm{level}_\mathrm{diff}(Q)] \le \frac{1}{\sigma}\sum_{i=1}^\sigma 2\cdot\frac{c-2}{c-1} = 2\cdot\frac{c-2}{c-1},
\]
independent of the sizes $s_1,\ldots,s_\sigma$ and of how they are distributed across the stream. The same argument applies to any fixed sequence $s_1,\dots,s_\tau>1$: by linearity,
$\mathbb{E}[\overline{\mathrm{level}_\mathrm{diff}}]=\frac1\tau\sum_{i=1}^{\tau}\mathbb{E}_{\mathcal{I}_{s_i}}[\mathrm{level}_\mathrm{diff}(Q)]\le 2\cdot\frac{c-2}{c-1}$.

\medskip
\noindent\textbf{(ii)} By the same two-stage linearity argument, applying Thm.~\ref{thm:stream-uniform}(ii) termwise to each fixed length $s_i$ (with $1<s_i\le N$ by assumption) gives
\[
\mathbb{E}[\overline{AR}] = \frac{1}{\sigma}\sum_{i=1}^\sigma \mathbb{E}_{\mathcal{I}_{s_i}}[AR(Q)] \;\ge\; \frac{1}{\sigma}\sum_{i=1}^\sigma \max\left\{1, \frac12\left\lfloor\log\frac{N}{s_i}\right\rfloor\right\},
\]
which is the first inequality of part (ii). For the second inequality, note that
$\left\lfloor\log\frac{N}{s_i}\right\rfloor > \log\frac{N}{s_i} - 1$ for every $i$, so
\[
\frac12\left\lfloor\log\frac{N}{s_i}\right\rfloor \;>\; \frac12\log\frac{N}{s_i} - \frac12
\;=\; \frac12\log N - \frac12\log s_i - \frac12 \ .
\]
Since $\max\{1,z\}\ge z$ for every real $z$, this gives
\[
\max\left\{1, \frac12\left\lfloor\log\frac{N}{s_i}\right\rfloor\right\} \;\ge\; \frac12\log N - \frac12\log s_i - \frac12\ , \qquad i = 1,\ldots,\sigma\ .
\]
Averaging this bound over $i=1,\ldots,\sigma$ gives
\[
\frac{1}{\sigma}\sum_{i=1}^\sigma \max\left\{1, \frac12\left\lfloor\log\frac{N}{s_i}\right\rfloor\right\}
\;\ge\; \frac12\log N - \frac{1}{2\sigma}\sum_{i=1}^\sigma\log s_i - \frac12
\;=\; \frac12\left(\log N - \frac1\sigma\sum_{i=1}^\sigma \log s_i\right) - \frac12\ ,
\]
which, combined with the first inequality, establishes part (ii).
\end{proof}

Part (ii) reveals an asymmetry between the two measures. The level-overhead bound is
size-agnostic (Thm.~\ref{thm:stream-uniform}(i)), so mixing query sizes costs nothing.
The accuracy-ratio bound, in contrast, degrades with $\log N - \frac{1}{\sigma}\sum_i \log s_i$,
the gap between $\log N$ and the \emph{average log-size} of the queries in the stream:
larger queries reduce the guaranteed accuracy gain of the $c$-DAG over the 1D-Tree. The result extends directly to unequal frequencies $\omega_1,\dots,\omega_\sigma$ (Sec.~\ref{sec:preliminaries}), giving
$\mathbb{E}[\overline{AR}] \ge \frac{1}{2}\bigl(\log N - \sum_i \omega_i \log s_i\bigr) - \frac{1}{2}$.

\vspace*{-2ex}
\section{Extension to Arbitrary Datasets via Distribution Learning} 
\label{sec:skewed}

\vspace*{-1ex}
We present a general methodology that bridges the theoretical analysis of Sec.~\ref{sec:theo}~to~the setting of arbitrary datasets over $[M]$, in which the $N \le M$ elements of $D$~may~not~be homogeneously spread over the domain; throughout this section we assume $N=2^{n+1}$.
Because the 1D-Tree and the $c$-DAG are data-dependent (Sec.~\ref{sec:DS}), every node still splits its elements into children holding~half~of~them, up to rounding, whatever the underlying distribution: both structures retain the same~shape, the same height $\lceil \log N \rceil$, and the same per-level node counts as in Sec.~\ref{sec:theo}. What changes is that sibling canonical ranges, while carrying an almost equal number of elements, may now have very different lengths, being short where elements are dense and long where they are sparse; Fig.~\ref{fig:1d-dag-hetero} illustrates this for a heterogeneous dataset. Two consequences follow. First, a random query may descend along different paths than it would under uniform occupancy, so the returned level distributions of the two structures change. Second, and crucially, the identity $AR(Q) = 2^{\mathrm{level}_\mathrm{diff}(Q)}$ of Sec.~\ref{sec:metrics}, valid for $N=2^{n+1}$, is unaffected, since it depends only on the halving~of~element~counts~and~not~on~range~lengths.

Under heterogeneous datasets, the individual returned level distributions of the 1D-Tree and the $c$-DAG could be analytically intractable. We therefore focus on the Level Overhead Distribution (LOD): it is far more stable than the individual returned level distributions and admits a clean approximation by the uniform baseline LODs of Thm.~\ref{thm:lod}, even under significant heterogeneity.

\smallskip
\noindent
\textbf{Uniform Baseline LOD Distributions.}
The theoretical analysis of Sec.~\ref{sec:theo} assumed uniform occupancy, $N = M = 2^{n+1}$, so that all canonical ranges are dyadic. Under this model, the Returned Level Overhead $\mathrm{level}_\mathrm{diff}(Q) = \mathrm{level}_{c\text{-DAG}}(Q) - \mathrm{level}_{\mathrm{Tree}}(Q)$ for a random query of length $s$ follows a truncated quasi-dyadic distribution, specified in Thm.~\ref{thm:lod}, whose exact probability mass function is characterized by a single parameter $s^*$, the effective query length that best matches the structural scale. We refer to these distributions, parameterized by different values of $s^*$, as the \emph{(uniform) baseline LOD distributions}, and denote by $\mathcal{I}_{s^*}$ the baseline with parameter $s^*$, i.e., the distribution of $\mathrm{level}_\mathrm{diff}(Q)$ in the uniform model of Sec.~\ref{sec:theo} when $x$ is drawn from
$\mathcal{I}_{s^*}$ (we use the same symbol for this LOD), for real $1<s^*\le N$ and, if $\lfloor\log(N/s^*)\rfloor\ge1$, subject to the non-degeneracy condition of Thm.~\ref{thm:stream-uniform}.

\smallskip
\noindent
\textbf{Approximating the Empirical LOD.}
Given a dataset whose underlying distribution is unknown, we consider its Level Overhead Distribution $\mathcal{D}_s$, the distribution of
$\mathrm{level}_\mathrm{diff}(Q)$ for $Q=[x,x+s)$ with $x$ drawn from $\mathcal{I}_s$ (possibly conditioned on $D\cap Q\neq\emptyset$, as in Sec.~\ref{sec:experiments}) over the structures built on $D$; Cor.~\ref{cor:arbitrary-datasets} below holds for any such query distribution, since its proof uses only that $\mathcal{D}_s$ is a distribution on $\{0,\dots,K\}$ and that $AR(Q)=2^{\mathrm{level}_\mathrm{diff}(Q)}$ for every query; in practice it is estimated by drawing sample queries and running SRC-search on both
structures. Our goal is to approximate $\mathcal{D}_s$ using the family of uniform baseline LOD distributions. For each candidate $s^*$, we compute the $L_2$ metric between $\mathcal{D}_s$ and the baseline $\mathcal{I}_{s^*}$, and designate the minimizing value $s^*$ as the closest (best-fitting) baseline match. Because the LOD is far more stable than the individual returned level distributions it is derived from, the empirical LOD typically remains close to one of the baseline LODs even under substantial heterogeneity. This closeness transfers our theoretical guarantees to arbitrary datasets as follows.

\vspace*{-1ex}
\begin{corollary}
\label{cor:arbitrary-datasets}
Consider a dataset $D$ of $N=2^{n+1}$ elements over a domain $[M]$, with $N \le M$, and a query length
$s$. Let $\mathcal{I}_{s^*}$ be a baseline LOD with parameter $1<s^*\le N$ satisfying the non-degeneracy
condition of Thm.~\ref{thm:stream-uniform}, and let $\kappa=\lfloor\log(N/s^*)\rfloor$, the maximum level
overhead of $\mathcal{I}_{s^*}$. Let $K\ge\kappa$ be an integer such that $\mathcal{D}_s$ is supported on
$\{0,\dots,K\}$ (one may always take $K=\lceil\log N\rceil$), and suppose $\mathcal{D}_s$ is
$\varepsilon$-close to $\mathcal{I}_{s^*}$ in $L_2$ metric on $\{0,\dots,K\}$, for some $\varepsilon>0$.
Then the following two bounds hold (if $\mathcal{D}_s$ is supported on $\{0,\dots,\kappa\}$, take $K=\kappa$):

\vspace*{-4ex}
\begin{eqnarray*}
\mathbb{E}_{\mathcal{D}_s}\!\left[\mathrm{level}_\mathrm{diff}(Q)\right]
&\le&
2\cdot\frac{c-2}{c-1} \;+\; \varepsilon\sqrt{\frac{K(K+1)(2K+1)}{6}}\, ;
\\
\mathbb{E}_{\mathcal{D}_s}\!\left[AR(Q)\right]
&\ge&
\max\left\{1, \frac{1}{2}\left\lfloor\log\frac{N}{s^*}\right\rfloor\right\}
\;-\; \varepsilon\sqrt{\frac{4^{K+1}-1}{3}}
\,.
\end{eqnarray*}
\end{corollary}

\begin{proof} Both $\mathcal{D}_s$ and $\mathcal{I}_{s^*}$ are probability distributions over $\{0,1,\dots,K\}$ (for $\mathcal{I}_{s^*}$ because $\kappa\le K$). Write
\[
b_k \;=\; \Pr\nolimits_{\mathcal{D}_s}[\mathrm{level}_\mathrm{diff}(Q)=k]
\;-\; \Pr\nolimits_{\mathcal{I}_{s^*}}[\mathrm{level}_\mathrm{diff}(Q)=k],
\qquad k=0,1,\dots,K ,
\]
so that the $L_2$-closeness hypothesis reads
$\bigl(\sum_{k=0}^{K} b_k^2\bigr)^{1/2}\le\varepsilon$. For any
coefficients $a_0,\dots,a_K$, the Cauchy--Schwarz inequality
$\bigl|\sum_k a_k b_k\bigr| \le \bigl(\sum_k a_k^2\bigr)^{1/2}
\bigl(\sum_k b_k^2\bigr)^{1/2}$ therefore gives
\begin{equation}
\label{eq:cs-master}
\Bigl|\sum_{k=0}^{K} a_k b_k\Bigr|
\;\le\; \varepsilon \cdot \Bigl(\sum_{k=0}^{K} a_k^2\Bigr)^{1/2}.
\end{equation}
We apply \eqref{eq:cs-master} twice, once for each measure.

\smallskip
\noindent\textbf{(i) Returned Level Overhead.}
Taking $a_k=k$ in \eqref{eq:cs-master} and using the closed form
$\sum_{k=0}^{K}k^2=\frac{K(K+1)(2K+1)}{6}$,
\[
\mathbb{E}_{\mathcal{D}_s}[\mathrm{level}_\mathrm{diff}(Q)]
-\mathbb{E}_{\mathcal{I}_{s^*}}[\mathrm{level}_\mathrm{diff}(Q)]
=\sum_{k=0}^{K} k\, b_k
\;\le\; \varepsilon\sqrt{\frac{K(K+1)(2K+1)}{6}} .
\]
Under the non-degeneracy condition of Theorem~\ref{thm:stream-uniform},
part~(i) of that theorem bounds
$\mathbb{E}_{\mathcal{I}_{s^*}}[\mathrm{level}_\mathrm{diff}(Q)]
< 2\cdot\frac{c-2}{c-1}$. Rearranging,
\[
\mathbb{E}_{\mathcal{D}_s}[\mathrm{level}_\mathrm{diff}(Q)]
\;\le\; \mathbb{E}_{\mathcal{I}_{s^*}}[\mathrm{level}_\mathrm{diff}(Q)]
+\varepsilon \sqrt{\frac{K(K+1)(2K+1)}{6}}
\;\le\; 2\cdot\frac{c-2}{c-1}
+\varepsilon\sqrt{\frac{K(K+1)(2K+1)}{6}} .
\]

\smallskip
\noindent\textbf{(ii) Accuracy Ratio.}
By Section~\ref{sec:metrics}, $AR(Q)=2^{\mathrm{level}_\mathrm{diff}(Q)}$, hence
$\mathbb{E}_{\mathcal{D}_s}[AR(Q)]
=\sum_{k=0}^{K}2^{k}\Pr_{\mathcal{D}_s}[\mathrm{level}_\mathrm{diff}(Q)=k]$,
and similarly for $\mathcal{I}_{s^*}$. Taking $a_k=2^k$ in
\eqref{eq:cs-master} and using
$\sum_{k=0}^{K}4^{k}=\frac{4^{K+1}-1}{3}$,
\[
\mathbb{E}_{\mathcal{I}_{s^*}}[AR(Q)]-\mathbb{E}_{\mathcal{D}_s}[AR(Q)]
= -\sum_{k=0}^{K} 2^{k} b_k
\;\le\; \varepsilon\sqrt{\frac{4^{K+1}-1}{3}} .
\]
Theorem~\ref{thm:stream-uniform}(ii), instantiated at $s=s^*$, bounds
$\mathbb{E}_{\mathcal{I}_{s^*}}[AR(Q)]\ge
\max\bigl\{1,\tfrac12\lfloor\log(N/s^*)\rfloor\bigr\}$. Rearranging,
\[
\mathbb{E}_{\mathcal{D}_s}[AR(Q)]
\;\ge\; \mathbb{E}_{\mathcal{I}_{s^*}}[AR(Q)]
-\varepsilon\sqrt{\frac{4^{K+1}-1}{3}}
\;\ge\; \max\left\{1,\frac{1}{2}\left\lfloor\log\frac{N}{s^*}\right\rfloor\right\}
-\varepsilon\sqrt{\frac{4^{K+1}-1}{3}} .
\]
This establishes both bounds. 
\end{proof}

\vspace*{-3ex}
\section{Implementing SRC-Search on $c$-DAGs: Local Selection Rules}
\label{sec:src_rules}

\vspace*{-1ex}
Unlike the 1D-Tree, in which SRC-search follows a unique root-to-node path,
in the $c$-DAG the SRC-search may sometimes have a choice between several valid paths, as canonical ranges of the children overlap. This multiplicity of options creates a local selection choice opportunity at each visited node of the $c$-DAG. The foundational SRC-search algorithm of~\cite{falzon2022range} (restated as Alg.~\ref{alg:SRC}~in~Sec.~\ref{sec:DS}, for completeness) resolves the local selection choice by exhaustive exploration: it descends into \textbf{every valid child} and returns a containing node of minimum range length, which under the assumptions of Sec.~\ref{sec:theo} is a deepest containing node (Sec.~\ref{sec:DS}).

\vspace*{-1ex}
\subsection{Single-Path SRC-Search with a Local Selection Rule}

\vspace*{-0.1ex}
A \emph{single-path SRC-search} instead descends into \textbf{exactly one valid child}, chosen by a local rule~$\pi$, until no valid child remains or a leaf is reached; see pseudocode~in~Alg.~\ref{alg:src-pi}.

\begin{algorithm}[t!]
\caption{Single-Path SRC-Search on the $c$-DAG with local selection rule $\pi$}
\label{alg:src-pi}
\begin{algorithmic}[1]
\Require A range-supporting DAG $\mathcal{G}=(V,E)$, range query $Q=[x,x+s)$, current node $v\in V$.
\Ensure Node $\mathrm{cand}\in V$ with $Q\subseteq\mathrm{range}_{\mathrm{cand}}$ that is inclusion-minimal along the $\pi$-guided path, or $\mathsf{null}$ if none is reachable from $v$.
\Statex \textit{Invoked as $\mathrm{SRC}_\pi(\mathcal{G},Q,r)$ with $r$ the root of $\mathcal{G}$; by convention $|\mathrm{range}_{\mathsf{null}}|=\infty$ and $\pi(\emptyset)=\bot$.}

\State Mark $v$ as explored;
\State $\mathrm{cand}\gets\mathsf{null}$
\If{$Q\subseteq\mathrm{range}_v$}
\State $\mathrm{cand}\gets v$
\EndIf
\State $W\gets\emptyset$
\For{each $(v,w)\in E$ with $w$ not yet explored}
\If{$Q\subseteq\mathrm{range}_w$}
\State $W\gets W\cup\{w\}$
\Comment{collect valid children of $v$}
\EndIf
\EndFor
\If{$\pi(W)\neq\bot$}
\State $w^*\gets\pi(W)$
\Comment{select one child using rule $\pi$}
\State $t\gets\mathrm{SRC}_\pi(\mathcal{G},Q,w^*)$
\Comment{recurse into the selected child}
\If{$t \neq \mathsf{null}$ \textbf{and} $|\mathrm{range}_t| < |\mathrm{range}_{\mathrm{cand}}|$}
\State $\mathrm{cand}\gets t$
\EndIf
\EndIf
\State \Return $\mathrm{cand}$
\end{algorithmic}
\end{algorithm}

\smallskip
\noindent\textbf{Our Algorithm~\ref{alg:src-pi} vs. Algorithm in~\cite{falzon2022range}.}
Alg.~\ref{alg:src-pi} modifies SRC-search in~\cite{falzon2022range} by collecting all valid children into $W$, then applying $\pi$ to select one child $w^*$, and recursing only into $w^*$ instead of recursing into every child. 
It therefore returns a node that is inclusion-minimal along the $\pi$-guided
path; in the uniform model of Sec.~\ref{sec:theo} and under the non-degeneracy condition, three of our four rules always reach the exhaustive search's level exactly (Thm.~\ref{thm:pi-bounds}), but the specific \emph{node} returned at that level need not coincide with the one returned by the exhaustive search of~\cite{falzon2022range} (Sec.~\ref{sec:privacy}).

\smallskip
\noindent\textbf{Worst-case Complexity.}
The $c$-DAG has height $\lceil \log N \rceil$; at each level $\ell$ it
contains exactly $(c-1)\cdot 2^{\ell}-(c-2)$ nodes, for a total of
$\Theta(cN)$ nodes and $\Theta(c^{2}N)$ links. Alg.~\ref{alg:src-pi} evaluates at most $c$ children per level and recurses into exactly one, for at most $c\cdot(\mathrm{level}_\pi(Q)+1)$ node evaluations; the exhaustive
search of~\cite{falzon2022range} evaluates at most $c-1$ nodes per level (since a query of positive length is contained in at most $c-1$ overlapping siblings per level), each with $c$ children, for at most
$c(c-1)\cdot(\mathrm{level}_{c\text{-DAG}}(Q)+1)$ evaluations. Both levels
are $O(\log N)$, hence the single-path SRC-search reduces the node-evaluation
count by a factor of at least $\Theta\left(\left(\frac{c}{c-1}\right)^{\log N}\right)$ in the worst case, while returning its
answer along a single root-to-node descent.

\subsection{Local Selection Rules $\pi_{\cdot}$}
\label{sec:rules}
Each node~$v$ of the $c$-DAG has $c$ children: the two Tree children
$w_1$ (left) and $w_c$ (right), whose canonical ranges partition
$\mathrm{range}_v$ at the median element of $D_v$, and the $c-2$
augmented children $w_2,\ldots,w_{c-1}$, each associated with the same
number $\lceil n_v/2 \rceil$ of consecutive elements of $D_v$ but starting
at shifted ranks, so that they straddle the median and overlap both $w_1$
and $w_c$. As in Alg.~\ref{alg:src-pi}, the set of valid children of
$v$ for query $Q$ is $W= \{w:(v,w)\in E, Q\subseteq\mathrm{range}_w\}$.
Let $W_{\mathrm{Tree}} = W\cap\{w_1,w_c\}$ and $W_{\mathrm{Aug}} =
W\cap\{w_2,\ldots,w_{c-1}\}$. We consider the following local selection
rules:

\smallskip\noindent\textbf{Tree rule ($\pi_T$).}
Select the unique Tree child in $W_{\mathrm{Tree}}$; if
$W_{\mathrm{Tree}}=\emptyset$, return $\bot$ even if augmented children
are valid. Since $w_1,w_c$ partition $\mathrm{range}_v$, $|W_{\mathrm{Tree}}|\leq 1$
always. This rule exactly recovers SRC-search on the embedded 1D-Tree.

\smallskip\noindent\textbf{Tree-first rule ($\pi_{\mathrm{TF}}$).}
If $W_{\mathrm{Tree}}\neq\emptyset$, select the unique Tree child. Else if
$W_{\mathrm{Aug}}\neq\emptyset$, select uniformly at random from
$W_{\mathrm{Aug}}$. Returns $\bot$ only if $W=\emptyset$.

\smallskip\noindent\textbf{Augmented-first rule ($\pi_{\mathrm{AF}}$).}
If $W_{\mathrm{Aug}}\neq\emptyset$, select uniformly at random from
$W_{\mathrm{Aug}}$; else select the unique child of $W_{\mathrm{Tree}}$ (if
any). Returns $\bot$ only if $W=\emptyset$.

\smallskip\noindent\textbf{Random rule ($\pi_R$).}
Select uniformly at random from $W$ when $|W|\geq2$ (deterministic when
$|W|=1$). Returns $\bot$ only if $W=\emptyset$.

\smallskip
Let $\mathrm{level}_\pi(Q)$ denote the level returned by
Alg.~\ref{alg:src-pi} with rule~$\pi$, and $\mathrm{level}_{c\text{-DAG}}(Q)$
the level of the deepest $c$-DAG node containing $Q$ (Sec.~\ref{sec:DS}),
which the exhaustive search of~\cite{falzon2022range} reaches.
We prove the following:

\begin{theorem}[Tightness of Single-Path Search]
\label{thm:pi-bounds}
Assume the uniform model of Sec.~\ref{sec:theo}, and let $Q=[x,x+s)$ with $1<s\le N$ satisfy the non-degeneracy condition $c-1\le2^{n-\kappa+1}$, $\kappa=\lfloor\log(N/s)\rfloor$. Then
\[
\mathrm{level}_\pi(Q)=\mathrm{level}_{c\text{-DAG}}(Q)\quad\text{for every }\pi\in\{\pi_{\mathrm{TF}},\pi_{\mathrm{AF}},\pi_R\}.
\]
\end{theorem}

The proof rests on facts that hold for any dataset, plus one structural
lemma specific to the uniform model. Lem.~\ref{lem:tree-rule} shows that
$\pi_T$ exactly replicates SRC-search on the embedded Tree.
Lem.~\ref{lem:tf-tree} shows that $\pi_{\mathrm{TF}}$ reaches at least as
deep as the Tree, since it follows the identical path for as long as a
Tree child remains valid, by definition of the rule.
Lem.~\ref{lem:upper} shows that any single-path rule reaches a level no
deeper than the exhaustive search, since it explores only one of the
branches the exhaustive search explores, and
Lem.~\ref{lem:tree-below-cdag} shows the exhaustive search is, for the
same reason, never shallower than the Tree. These four facts hold without any assumption on the dataset. The equality of
Thm.~\ref{thm:pi-bounds} additionally requires that, at every level above the one returned
by the exhaustive search, a node containing $Q$ has a child containing $Q$; we establish this
for the uniform model in Lem.~\ref{lem:child-containment}.

\begin{lemma}[Tree Rule Equivalence]
\label{lem:tree-rule}
For any query~$Q$ and any dataset, $\mathrm{level}_{\pi_T}(Q) = \mathrm{level}_{\mathrm{Tree}}(Q)$.
\end{lemma}
\begin{proof}
By Sec.~\ref{sec:DS}, $w_1,w_c$ are the Tree children of every node of
the $c$-DAG, over any dataset. Rule $\pi_T$ selects only the unique Tree
child and returns $\bot$ when $W_{\mathrm{Tree}}=\emptyset$; hence the
$\pi_T$-guided descent coincides, node for node, with SRC-search on the
embedded Tree.
\end{proof}

\begin{lemma}[Tree-First Dominance]
\label{lem:tf-tree}
For any query~$Q$ and any dataset, $\mathrm{level}_{\mathrm{Tree}}(Q)\le\mathrm{level}_{\pi_{\mathrm{TF}}}(Q)$.
\end{lemma}
\begin{proof}
By definition, $\pi_{\mathrm{TF}}$ selects the unique Tree child whenever
$W_{\mathrm{Tree}}\neq\emptyset$, exactly as $\pi_T$ does. By induction on
the level, the $\pi_{\mathrm{TF}}$-guided descent therefore coincides,
node for node, with the $\pi_T$-guided descent -- and hence, by
Lemma~\ref{lem:tree-rule}, with the Tree's own descent -- at every level
$\ell<\mathrm{level}_{\mathrm{Tree}}(Q)$, since at each such level the
Tree child is valid and both rules select it. At level
$\mathrm{level}_{\mathrm{Tree}}(Q)$, this common node has
$W_{\mathrm{Tree}}=\emptyset$ by definition of $\mathrm{level}_{\mathrm{Tree}}(Q)$,
so $\pi_{\mathrm{TF}}$ has already reached this level, and may continue
further into an augmented child if $W_{\mathrm{Aug}}\neq\emptyset$, but
never stops earlier. Hence $\mathrm{level}_{\pi_{\mathrm{TF}}}(Q)\ge\mathrm{level}_{\mathrm{Tree}}(Q)$.
\end{proof}

\begin{lemma}[Upper Bound for Single-Path Rules]
\label{lem:upper}
For any rule~$\pi$, any query~$Q$, and any dataset, $\mathrm{level}_\pi(Q)\le\mathrm{level}_{c\text{-DAG}}(Q)$.
\end{lemma}
\begin{proof}
Alg.~\ref{alg:src-pi} recurses only into valid children, so every node on the $\pi$-guided descent, and in particular the returned node, contains $Q$.
Its level is therefore at most the largest level of a $c$-DAG node containing $Q$, which is $\mathrm{level}_{c\text{-DAG}}(Q)$.
\end{proof}

\begin{lemma}[The Tree is Never Deeper than the Exhaustive $c$-DAG]
\label{lem:tree-below-cdag}
For any query~$Q$ and any dataset, $\mathrm{level}_{c\text{-DAG}}(Q)\ge\mathrm{level}_{\mathrm{Tree}}(Q)$.
\end{lemma}
\begin{proof}
By Lem.~\ref{lem:tree-rule}, the descent that follows only Tree children
reaches level $\mathrm{level}_{\mathrm{Tree}}(Q)$. Every Tree child valid
along this descent is, by construction, also a valid child in the $c$-DAG, so the node it reaches is a $c$-DAG node of level
$\mathrm{level}_{\mathrm{Tree}}(Q)$ containing $Q$, and
$\mathrm{level}_{c\text{-DAG}}(Q)$ is the largest such level. Hence $\mathrm{level}_{c\text{-DAG}}(Q)\ge\mathrm{level}_{\mathrm{Tree}}(Q)$.
\end{proof}

We now sharpen the upper bound of Lem.~\ref{lem:upper} to an equality for
the three non-Tree rules, in the uniform model of Sec.~\ref{sec:theo}.

\begin{lemma}[Child-containment property, uniform model]
\label{lem:child-containment}
In the uniform model, let $v$ be a $c$-DAG node at level~$\ell$ with
$Q\subseteq\mathrm{range}_v$. If some level-$(\ell+1)$ $c$-DAG node
contains $Q$, then $v$ has a child $w$ with $Q\subseteq\mathrm{range}_w$.
\end{lemma}
\begin{proof}
Write $L=2^{n-\ell+1}$, $\delta=L/(2(c-1))$. Level $(\ell+1)$ nodes start
at $p\delta$ for integer $p$; level-$\ell$ nodes start at $2m\delta$ for
integer $m$, and $v=[2m\delta,2m\delta+L)$'s $c$ children occupy exactly
the consecutive indices $\{2m,\dots,2m+c-1\}$. Write $Q=[x,x+s)$. Index
$p$ is valid iff $p\in[P_{\mathrm{lo}},P_{\mathrm{hi}}]:=
[(x+s-L/2)/\delta,\ x/\delta]$. Since $Q\subseteq\mathrm{range}_v$:
\[
2m \le x/\delta = P_{\mathrm{hi}}, \qquad
2m+(c-1) \ge (x+s-L)/\delta + L/(2\delta) = (x+s-L/2)/\delta = P_{\mathrm{lo}}.
\]
By hypothesis some integer $q\in[P_{\mathrm{lo}},P_{\mathrm{hi}}]$ exists.
If $P_{\mathrm{lo}}\ge 2m$: let $p^*=\lceil P_{\mathrm{lo}}\rceil$; then
$p^*\ge 2m$, and since $q\ge P_{\mathrm{lo}}$ and $q$ is an integer,
$q\ge p^*$, while $q\le P_{\mathrm{hi}}$ gives $p^*\le P_{\mathrm{hi}}$;
also $p^*\le 2m+c-1$ since $P_{\mathrm{lo}}\le 2m+c-1$ is an integer
bound. If $P_{\mathrm{lo}}<2m$: then $2m\le P_{\mathrm{hi}}$ (shown above)
and $2m>P_{\mathrm{lo}}$, so $2m\in(P_{\mathrm{lo}},P_{\mathrm{hi}}]$, and
$p^*=2m$ works. Either way
$p^*\in\{2m,\dots,2m+c-1\}\cap[P_{\mathrm{lo}},P_{\mathrm{hi}}]$, so child
$w=w_{p^*-2m+1}$ of $v$ satisfies $Q\subseteq\mathrm{range}_w$.
\end{proof}

\begin{proof}[Proof of Thm.~\ref{thm:pi-bounds}]
Fix $\pi\in\{\pi_{\mathrm{TF}},\pi_{\mathrm{AF}},\pi_R\}$; each returns $\bot$ only when
$W=\emptyset$, and every $w\in W$ satisfies $Q\subseteq\mathrm{range}_w$. Let
$\lambda=\mathrm{level}_{c\text{-DAG}}(Q)$. Since $\kappa=\lfloor\log(N/s)\rfloor$ gives
$s>2^{n-\kappa}$, no node of level $\kappa+1$ or deeper contains $Q$, so $\lambda\le\kappa$;
by non-degeneracy, $\kappa\le n+1-\alpha$. Hence Lem.~\ref{lem:child-containment} is
applied below only at levels $\ell\le\lambda-1\le n-\alpha$, where Part~0 of the proof of
Lem.~\ref{lem:cdag-level-dist} describes levels $\ell$ and $\ell+1$. For every level
$\ell<\lambda$, the ancestor at level $\ell+1$ of a level-$\lambda$ node containing $Q$
also contains $Q$; hence, by Lem.~\ref{lem:child-containment} and induction on
$\ell=0,1,\dots,\lambda-1$, every node visited by Alg.~\ref{alg:src-pi} at a level
$\ell<\lambda$ has $W\neq\emptyset$, so the algorithm does not stop before level
$\lambda$. With Lem.~\ref{lem:upper}, equality follows.
\end{proof}

In the uniform model, $\pi_T$ recovers SRC-search on the embedded
1D-Tree exactly (see the Tree rule above), and therefore attains
exactly the 1D-Tree's Accuracy Ratio, whereas by
Thm.~\ref{thm:pi-bounds} the three non-Tree rules attain exactly the $c$-DAG's Accuracy Ratio and thus inherit the bound of
Thm.~\ref{thm:stream-uniform}(ii). Randomness among the non-Tree
rules therefore affects neither the returned level nor the Accuracy Ratio; it affects only which node is returned at the level fixed by Thm.~\ref{thm:pi-bounds} -- the question
Sec.~\ref{sec:privacy}~takes~up~next.

\vspace*{-1.2ex}
\section{Privacy Analysis}
\label{sec:privacy}

\vspace*{-1ex}
By Thm.~\ref{thm:pi-bounds}, in the uniform model the returned level is fixed by $Q$ and the rule class: $\mathrm{level}_{\mathrm{Tree}}(Q)$
under $\pi_T$, and $\mathrm{level}_{c\text{-DAG}}(Q)$ under
$\pi_{\mathrm{TF}},\pi_{\mathrm{AF}},\pi_R$. The remaining question is which specific node is returned at that level. Since $\pi_T$ makes no random choice, it is deterministic and induces no such uncertainty. For
the other three rules, randomness enters only at a node $v$ with
$|W(v)|\ge2$: two valid children of the same node, both containing $Q$. We
measure the resulting uncertainty by the Shannon entropy of the returned
node over these internal coin flips, for the fixed query $Q$:

\vspace*{-0.9ex}
\[
H_\pi(Q) \;=\; -\sum_{v} \Pr[\,\pi \text{ returns } v \mid Q\,] \log_2
\Pr[\,\pi \text{ returns } v \mid Q\,]
\ .
\]

\vspace*{-1ex}
\noindent
Positive $H_\pi(Q)$ means repeated executions of the identical query $Q$
can disclose different canonical ranges to a server observing the access
pattern, relevant to attacks exploiting query
repetition~\cite{cash2015leakage,oya2021exploiting}.

\begin{proposition}
\label{prop:c3-zero-entropy}
For $c=3$, $H_\pi(Q)=0$ for every query~$Q$ and every rule~$\pi$.
\end{proposition}

\begin{proof}
Every node has exactly one Tree child on each side ($w_1,w_3$, mutually
exclusive since they partition $\mathrm{range}_v$) and a single augmented
child $w_2$, so $|W(v)|\le2$: a tie can only be a Tree child and the augmented child, both containing $Q$. Both are children of $v$, of equal length $L/2$ and spaced $L/4$ apart, so writing $Q=[x,x+s)$, both contain $Q$ exactly when $x$ lies in an interval whose upper endpoint coincides with the containment condition of their common child (the one straddling their shared boundary); hence whenever both are valid, they share a valid
child. Consequently a tied node is never terminal: if $v$ has two valid children, the tie is resolved by both of them leading, one level further, to a common node. Applying this observation inductively -- the two children of a tie are again a Tree/augmented pair, since $c=3$ admits no
other configuration -- shows that a two-way tie can persist across levels only by continually funnelling into shared descendants, and can never itself be the level at which the search terminates (termination requires
$W=\emptyset$, which a tied node never has). The search therefore returns a single node with probability $1$, for every rule.
\end{proof}

For $c\ge5$, by contrast, positive entropy does occur. 

\emph{Example:}
for $N=32$ ($n=4$), take $v=[2,10)$ at level~2 (root $\to[0,16)\to[2,10)$).
Its 5 children at level~3 are $[2,6),[3,7),[4,8),[5,9),[6,10)$. For $Q=[5,8)$ ($s=3$, too long to fit any level-4 node, so level~3 is final):
both $[4,8)$ and $[5,9)$ contain $Q$, so $\pi_R$ (and $\pi_{\mathrm{AF}}$,
and $\pi_{\mathrm{TF}}$ whenever no Tree child is valid) return each with probability $\tfrac12$, giving $H_\pi(Q)=1>0$; only $\pi_T$ remains deterministic. Tunable branching thus offers a privacy benefit
unavailable to the $c=3$ TDAG of~\cite{demertzis2016practical,demertzis2018practical}:
for $c\ge5$, repeated executions of the same query need not disclose the
same candidate set, at no cost in returned level relative to the
exhaustive search (Thm.~\ref{thm:pi-bounds}) and no cost beyond the
Returned Level Overhead and Accuracy Ratio bounds of Sec.~\ref{sec:theo}.

\section{Experimental Evaluation}
\label{sec:experiments}

\noindent
We evaluate the 1D-Tree and the $c$-DAGs ($c=3$ and $c=5$) on two real datasets
and ask how closely the uniform analysis of Sec.~\ref{sec:theo} describes them
when its assumptions do not hold (Sec.~\ref{sec:skewed}). We report the empirical
Level Overhead Distribution (LOD), its distance to the closest uniform baseline
$\mathcal{I}_{s^*}$ of Thm.~\ref{thm:lod}, and the two stream measures of
Sec.~\ref{sec:metrics}, the average returned level overhead
$\overline{\mathrm{level}_\mathrm{diff}}$ and the average accuracy ratio
$\overline{AR}$, which Thm.~\ref{thm:stream-uniform} and
Cor.~\ref{cor:arbitrary-datasets} bound.

\subsection{Setup}

\noindent\textbf{Implementation.} All structures and the exhaustive SRC-search of
Alg.~\ref{alg:SRC} are implemented in Python and follow the definitions of
Sec.~\ref{sec:DS}: nodes are defined by ranks of the realized elements, and nodes of the same level with the same canonical range are identified (on prefixes
of both datasets we found no two nodes of the same level with the same elements and
different ranges, so this coincides with identifying nodes with the same elements). SRC-search returns the deepest node whose canonical range contains $Q$,
breaking ties in favor of the node with the smallest $n_v$. Code, per-query records
and plots are available in a
\href{https://github.com/ladan90/Data-Structures-Simulation/tree/main/1D}{public GitHub repository}.

\smallskip
\noindent\textbf{Datasets.}
\emph{Gowalla}~\cite{Cho2011Friendship} is a location-based social network whose public check-in file contains $6{,}442{,}892$ records collected between February
2009 and October 2010. We use the check-in times in seconds, measured from the
earliest check-in (4 February 2009, 05:17:38 UTC), and take as $D$ the first
$N=2^{22}=4{,}194{,}304$ of the $5{,}561{,}957$ distinct times; $D$ spans
$49{,}626{,}707$ seconds (about $574$ days, until 1 September 2010) and occupies
$N/M\approx 8.5\%$ of the positions of its domain $[M]$,
$M=\max D+1$. \emph{NYC Yellow Taxi} consists of the pickup times of the
yellow-taxi trips of January 2024 published by the NYC Taxi and Limousine Commission~\cite{nyctlc_tripdata}. Of the $2{,}964{,}624$ records, $18$ carry a timestamp outside January 2024 and are removed; times are measured from
1 January 2024, 00:00:00, and $D$ consists of the first $N=2^{20}=1{,}048{,}576$
of the $1{,}575{,}689$ distinct times. It spans $1{,}799{,}644$ seconds (about
$20.8$ days) and occupies $N/M\approx 58\%$ of its domain. The average spacing of
the elements, $\max D/(N-1)$, is $11.83$ s for Gowalla and $1.716$ s for NYC.
Both datasets are clustered in time (evenings and weekends for
Gowalla, rush hours and weekend events for NYC)~\cite{Cho2011Friendship,Ferreira2013,yu2016two},
which makes them testbeds for the framework of Sec.~\ref{sec:skewed}: the
returned level distributions of the individual structures are analytically
intractable, whereas the LOD may remain close to a uniform baseline.

\smallskip
\noindent\textbf{Query streams.}
For each query length $s\in\{60,3600,86400,604800\}$ seconds (1 minute, 1 hour,
1 day, 1 week), the stream consists of queries $Q=[x,x+s)$ whose starting
position $x$ is drawn i.i.d.\ uniformly from $\mathcal{I}_s=[0,M-s]$ (real
valued), as in Sec.~\ref{sec:preliminaries}; the same queries are answered by the
1D-Tree, the 3-DAG and the 5-DAG. We keep only non-empty queries, i.e., the stream
is drawn from $\mathcal{I}_s$ conditioned on $D\cap Q\neq\emptyset$, because the
accuracy of a returned node is defined only when $D\cap Q\neq\emptyset$
(Sec.~\ref{sec:metrics}); queries are drawn i.i.d., $|D\cap Q|$ is computed exactly, and empty queries are discarded.
Cor.~\ref{cor:arbitrary-datasets} is applied to the LOD under this conditional
distribution. For the four lengths in increasing order, the discarded fraction of
the draws is $0.449$, $0.249$, $0.050$ and $0.015$ on Gowalla and $0.002$ ($202$ of
$133{,}702$ draws), $0$, $0$ and $0$ on NYC; the large fractions on Gowalla at
short lengths reflect gaps between check-ins. Keeping the empty queries, as an earlier version of this work did, changes the Gowalla results at short lengths: for
$s=60$ and $s=3600$ the mean level of the 1D-Tree is then $19.41$ and $15.33$
(against $17.32$ and $13.13$) and the average level overhead of the 3-DAG is $0.56$
and $0.75$ (against $1.00$ and $1.00$); NYC is practically unaffected.
We issue an initial batch of $5{,}000$ queries and then batches of $500$, and compare the empirical LODs, cumulative since the first query, of the last two
snapshots. The stream is declared stable when their $L_2$ distance is below
$\varepsilon=0.001$ for both $c$-DAGs at the same snapshot, and we then issue
$\tau=120{,}000$ further queries (at most $200{,}000$ per length, never reached).
All measures are computed on the final snapshot, and the boxes of
Figs.~\ref{fig:level-diff-whisker} and~\ref{fig:ar-whisker} cover all snapshots
from stabilization on. The random seed is fixed. In all $1{,}069{,}000$
processed queries the $c$-DAG returned a node at least as deep as the 1D-Tree ($\mathrm{level}_\mathrm{diff}(Q)\ge 0$), and, since $N$ is a power of
two, the measured accuracy ratio $n_{\mathrm{Tree}}(Q)/n_{c\text{-DAG}}(Q)$ equals
$2^{\mathrm{level}_\mathrm{diff}(Q)}$ for every query; both properties are
checked by the code.

\smallskip
\noindent\textbf{Uniform baseline and measures.}
For each length and each $c$ we compare the empirical LOD with the baseline LOD
$\mathcal{I}_{s^*}$ of Thm.~\ref{thm:lod} for $N=2^{22}$ (Gowalla) and $N=2^{20}$ (NYC).
The nominal parameter is $s^*=\mathrm{round}(s/\mathrm{step})$ with $\mathrm{step}=\max D/(N-1)$, i.e.\ $5$, $304$, $7302$, $51116$ for Gowalla and $35$, $2098$, $50342$, $352391$ for NYC. We use the integer $s^*$ within $\pm 500$ of the nominal value (step $1$) that minimizes the $L_2$ distance to the empirical
LOD, among those satisfying $1<s^*\le N$ and the non-degeneracy condition of Thm.~\ref{thm:lod}. The reference values are the expectations of the fitted
baseline, the bound $2\frac{c-2}{c-1}$ of Thm.~\ref{thm:stream-uniform}(i) ($1$ for $c=3$, $1.5$ for $c=5$), and the lower bound of Thm.~\ref{thm:stream-uniform}(ii)
at the fitted $s^*$. The $L_2$ distance is very flat in $s^*$: for the 3-DAG, its
values at the nominal and at the fitted $s^*$ differ by less than $0.001$ in all
eight configurations, and in $9$ of the $16$ configurations the minimizer lies at an end of the search window. We therefore use the $L_2$ distance and the baseline expectations, not the fitted $s^*$ itself, as results.

\subsection{Results}

\noindent\textbf{Returned levels.}
Fig.~\ref{fig:level-dist} shows the empirical returned level distributions for
$s=3600$ and $s=86400$; the remaining lengths are in Tables~\ref{tab:gowalla}
and~\ref{tab:nyc}. Since $\mathrm{level}_\mathrm{diff}(Q)\ge0$ for every query, the
CDF of each $c$-DAG never lies above that of the 1D-Tree. The mean returned level
of the 1D-Tree, the 3-DAG and the 5-DAG, for $s=60,3600,86400,604800$, is
$17.32/18.32/18.53$, $13.13/14.13/14.35$, $10.06/11.03/11.24$ and
$7.68/8.52/8.76$ on Gowalla, and $13.25/14.26/14.57$, $7.34/8.33/8.63$,
$2.69/3.47/3.85$ and $0.50/0.74/1.00$ on NYC. The 1D-Tree returns the root for
$0.00\%$, $0.01\%$, $0.19\%$ and $1.25\%$ of the Gowalla queries and for
$0.00\%$, $0.20\%$, $5.05\%$ and $50.35\%$ of the NYC queries; the 3-DAG returns the
root only for NYC at one week ($26.04\%$), and the 5-DAG never does. On Gowalla, the $c$-DAGs thus return levels in a more restricted range (for example, at least level $3$ at one week for the 3-DAG, $4$ for the 5-DAG, and at least level $6$ at one day), which limits false positives.

\smallskip
\noindent\textbf{Level Overhead Distribution.}
Figs.~\ref{fig:lod-gowalla} and~\ref{fig:lod-nyc} compare the empirical LOD with the best-fit baseline. The baseline reproduces the empirical LOD closely: the $L_2$ distance is at most $0.01$ in $12$ of the $16$ configurations (at most
$0.003$ in $7$). It is larger in four: Gowalla at one week ($0.039$ for the 3-DAG, $0.132$ for the 5-DAG) and NYC at one day ($0.058$ and $0.062$); Keeping the empty queries would have worsened the fit on Gowalla at short lengths: in a run without the non-empty filter the distances were
$0.258$ and $0.190$ for $s=60$, and $0.143$ and $0.048$ for $s=3600$ (3-DAG and 5-DAG), against at most $0.003$ for the non-empty queries.

\smallskip
\noindent\textbf{Average returned level overhead.}
Fig.~\ref{fig:level-diff-whisker} and Tables~\ref{tab:gowalla}
and~\ref{tab:nyc} give $\overline{\mathrm{level}_\mathrm{diff}}$. It is below the
bound of Thm.~\ref{thm:stream-uniform}(i) in $13$ of the $16$
configurations; in the other three (3-DAG, Gowalla at $s=60$ and $s=3600$, NYC at
$s=60$) it is $1.0002$, $1.0008$ and $1.0072$, within $1.9$ standard errors (about $0.004$) of the baseline expectation, which is below the bound; the theorem
bounds an expectation, whereas we report sample means. Over all configurations it
ranges between $0.85$ and $1.00$ (3-DAG) and between $1.08$ and $1.22$ (5-DAG) on
Gowalla, and between $0.24$ and $1.01$ (3-DAG) and between $0.50$ and $1.32$ (5-DAG)
on NYC. The empirical value is within $0.013$ of the baseline expectation in $13$
of the $16$ configurations; the others are Gowalla at one week ($0.074$ for the
3-DAG, $0.263$ for the 5-DAG) and NYC at one day for the 5-DAG ($0.065$), all in
configurations with a large $L_2$ distance. Applying the first bound of
Cor.~\ref{cor:arbitrary-datasets} with $\varepsilon$ set to the fitted $L_2$
distance (an estimate of $\varepsilon$, up to sampling error) gives upper bounds between $1.005$
and $2.762$, which the empirical values respect in all configurations. The boxes in
Fig.~\ref{fig:level-diff-whisker} are narrow: the snapshot averages from
stabilization on span at most $0.025$. Thus the overhead in traversal depth
remains a small constant on real data, at most $0.008$ above the proven constant in
the three cases above and below it otherwise.

\smallskip
\noindent\textbf{Average accuracy ratio.}
Fig.~\ref{fig:ar-whisker} and the tables give $\overline{AR}$, which ranges from $2.6$ to $9.7$ on Gowalla and from $1.2$ to $9.5$ on NYC and decreases with the query length. The mean is dominated by rare queries with large level difference:
the median of $AR$ is $1$ for the 3-DAG on Gowalla for all four lengths and for NYC at one
week ($5$ configurations), and $2$ in the other $11$ configurations, and the
geometric mean lies between $1.18$ and $2.50$. Excluding the two 3-DAG
configurations whose fitted $s^*$ is $2$ (Gowalla $s=3600$, NYC $s=60$), the
empirical $\overline{AR}$ lies between $0.71$ and $1.21$ times the expectation of the
fitted baseline. For these two configurations the fit is flat in $s^*$ and the
baseline expectation ($10.5$ and $9.5$) is not informative; they lie below the lower
bound of Thm.~\ref{thm:stream-uniform}(ii) ($6.3$ against $10.5$, $7.3$ against
$9.5$). The empirical $\overline{AR}$ is at least that lower bound in $11$ of
the $16$ configurations; of the other three, it is $0.99$, $0.95$ and $0.87$
times the bound (NYC $s=3600$, Gowalla $s=86400$ and $s=604800$, all for the 3-DAG).
The second bound of Cor.~\ref{cor:arbitrary-datasets}, which subtracts $\varepsilon\sqrt{(4^{K+1}-1)/3}$, is informative (above $1$) in only one configuration (NYC, $s=3600$, 3-DAG: $1.69$ against an empirical $4.46$), so we compare $\overline{AR}$ only with the baseline expectation. At the shortest length the snapshot averages of $\overline{AR}$ are less stable than those of the overhead (Gowalla,
$s=60$: $6.24$ to $8.54$ for the 3-DAG and $7.35$ to $10.23$ for the 5-DAG); for
$s\ge3600$ the boxes span less than $0.45$.

\smallskip
\noindent\textbf{Stabilization.}
Fig.~\ref{fig:stability} shows that the $L_2$ distance between consecutive LODs
decreases with the number of queries and falls below $\varepsilon$ after $12{,}500$,
$14{,}000$, $16{,}000$ and $16{,}500$ queries on Gowalla and $13{,}500$, $14{,}000$,
$14{,}000$ and $8{,}500$ on NYC (lengths in increasing order). It is below
$3.3\cdot10^{-4}$ at the end of every stream, so the additional $120{,}000$
queries are well past stabilization.

\smallskip
\noindent In summary, on both datasets the $c$-DAGs keep the level overhead at or
below the proven constant up to sampling error, reduce the number of returned elements
by a factor $2^k$ that is modest for the typical query and larger on average because
of rare queries with a large level difference, and the LOD of the real data stays
close to a uniform baseline in most configurations.

\begin{figure}[H]
\centering
\includegraphics[width=0.48\textwidth]{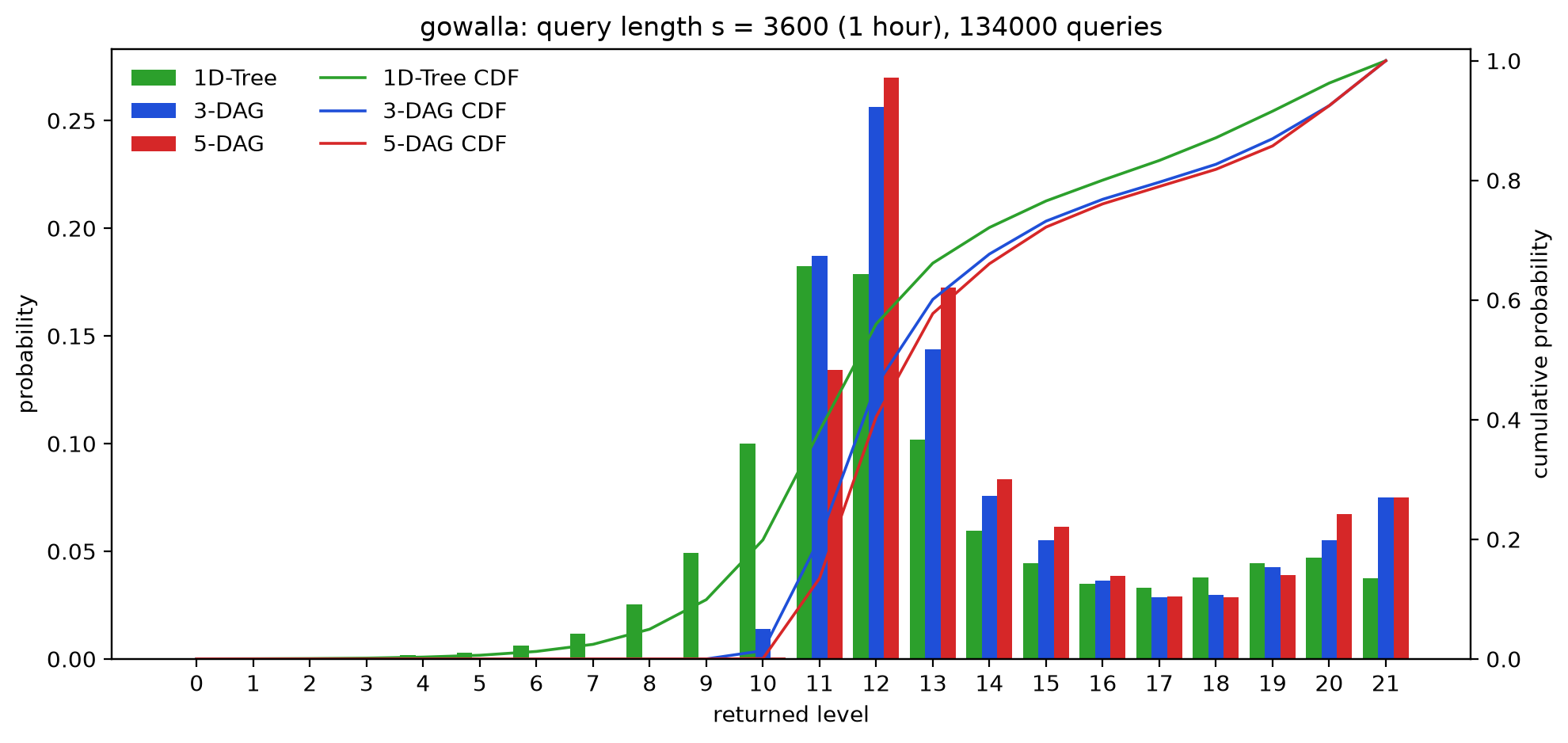}\hfill
\includegraphics[width=0.48\textwidth]{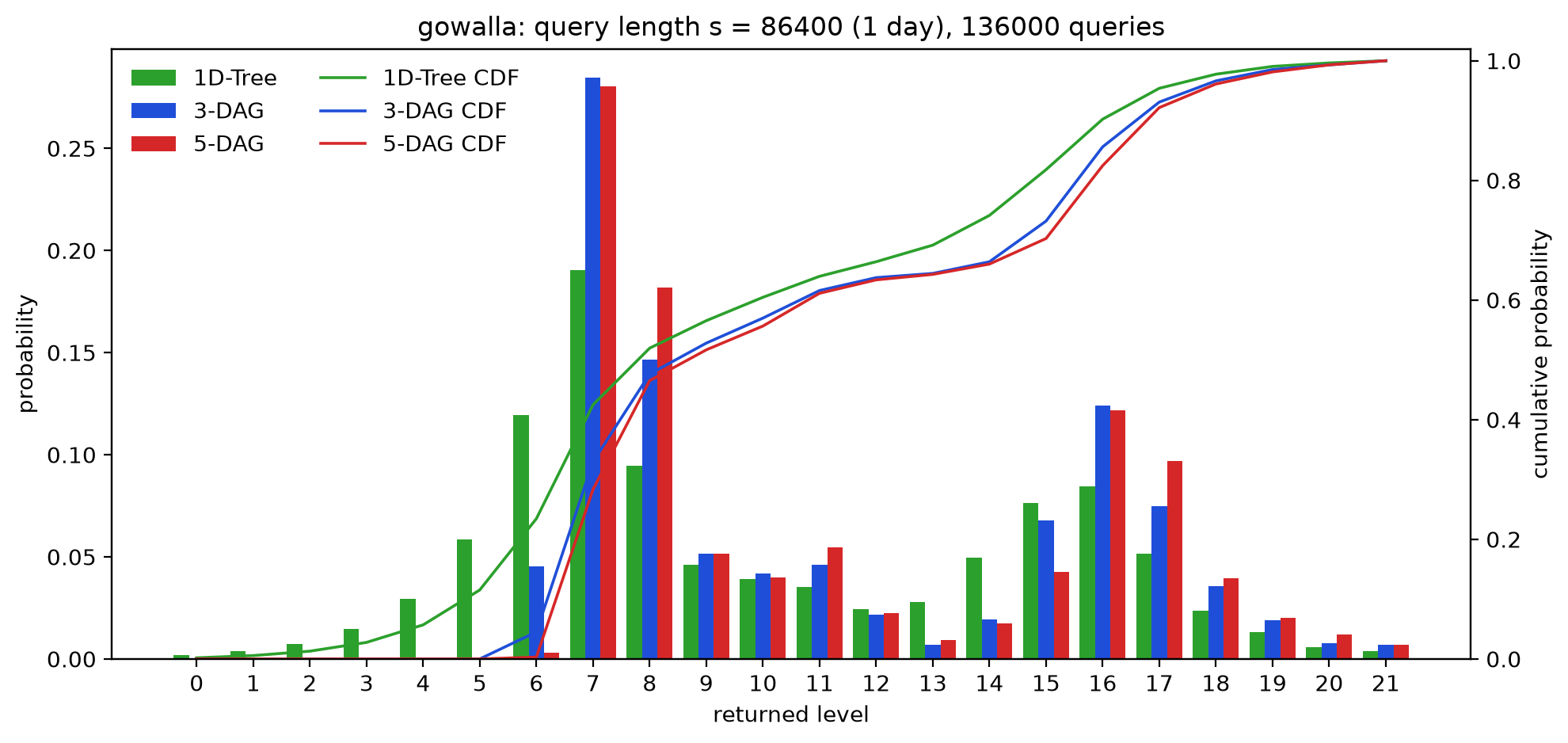}\\[1mm]
\includegraphics[width=0.48\textwidth]{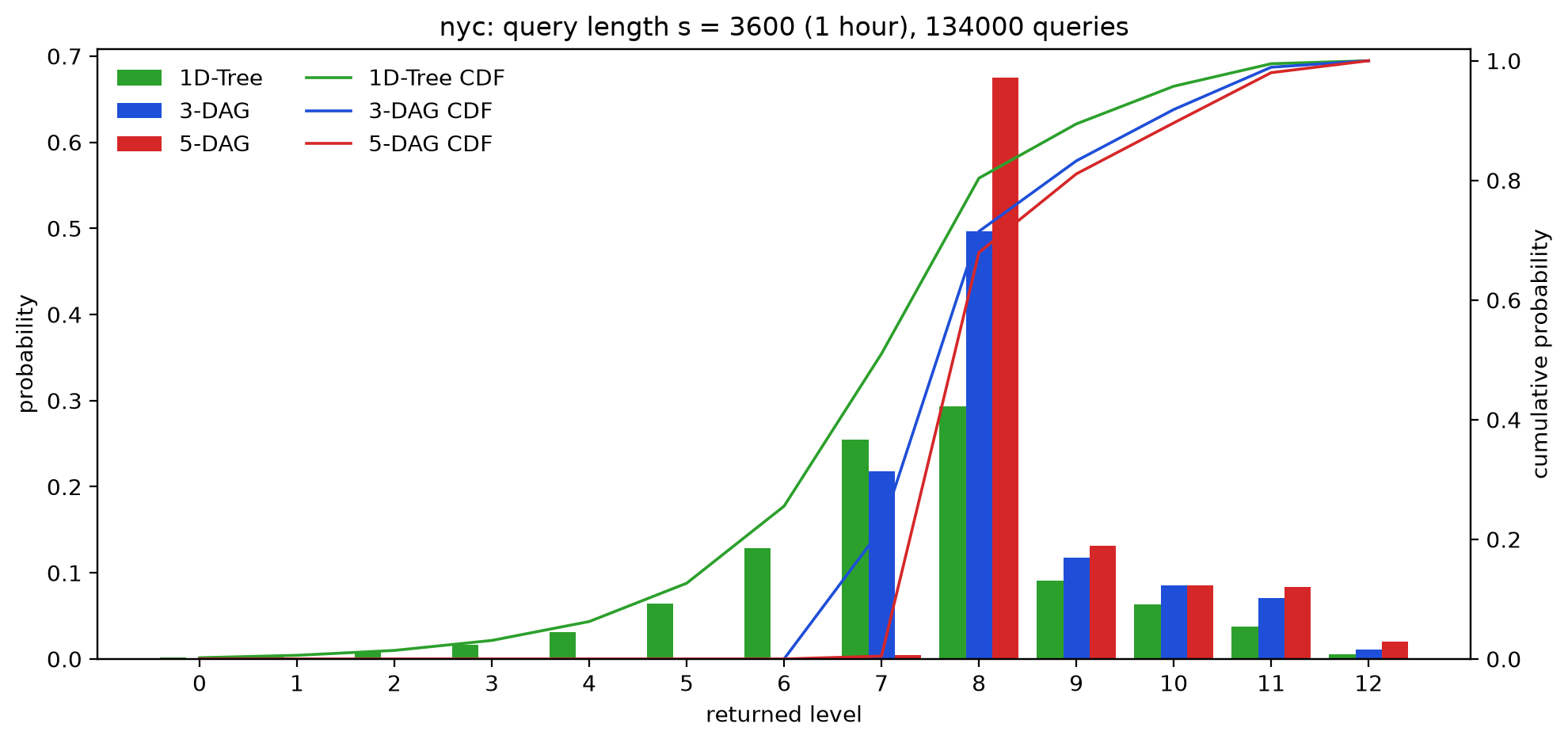}\hfill
\includegraphics[width=0.48\textwidth]{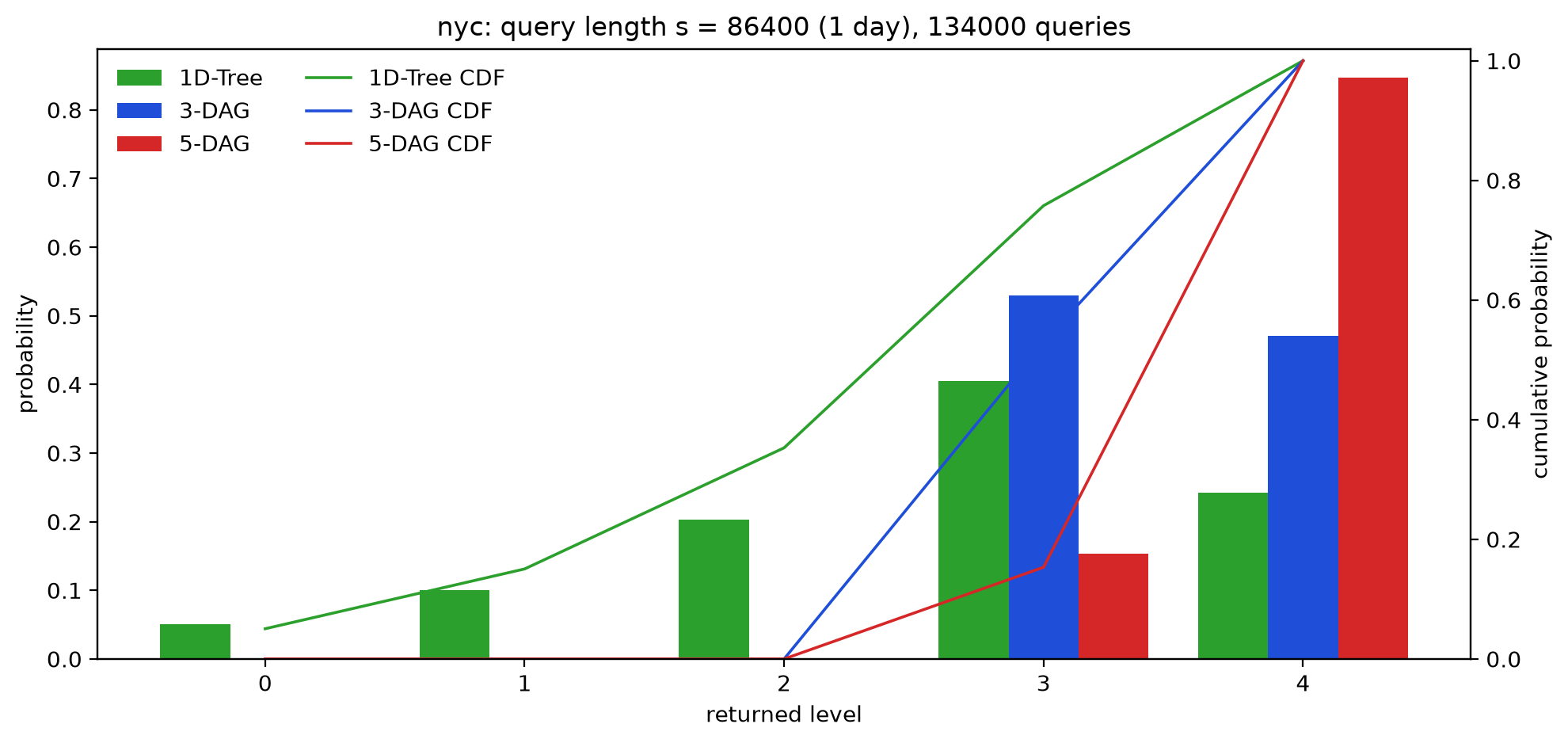}
\caption{Returned level distributions on Gowalla (top) and NYC Yellow Taxi (bottom), for $s=3600$ (left) and $s=86400$ (right). Green, blue and red bars: empirical probabilities for the 1D-Tree, 3-DAG and 5-DAG (left axis); lines: their
cumulative distribution functions (right axis).}
\label{fig:level-dist}
\end{figure}

\begin{figure}[H]
\centering
\includegraphics[width=0.48\textwidth]{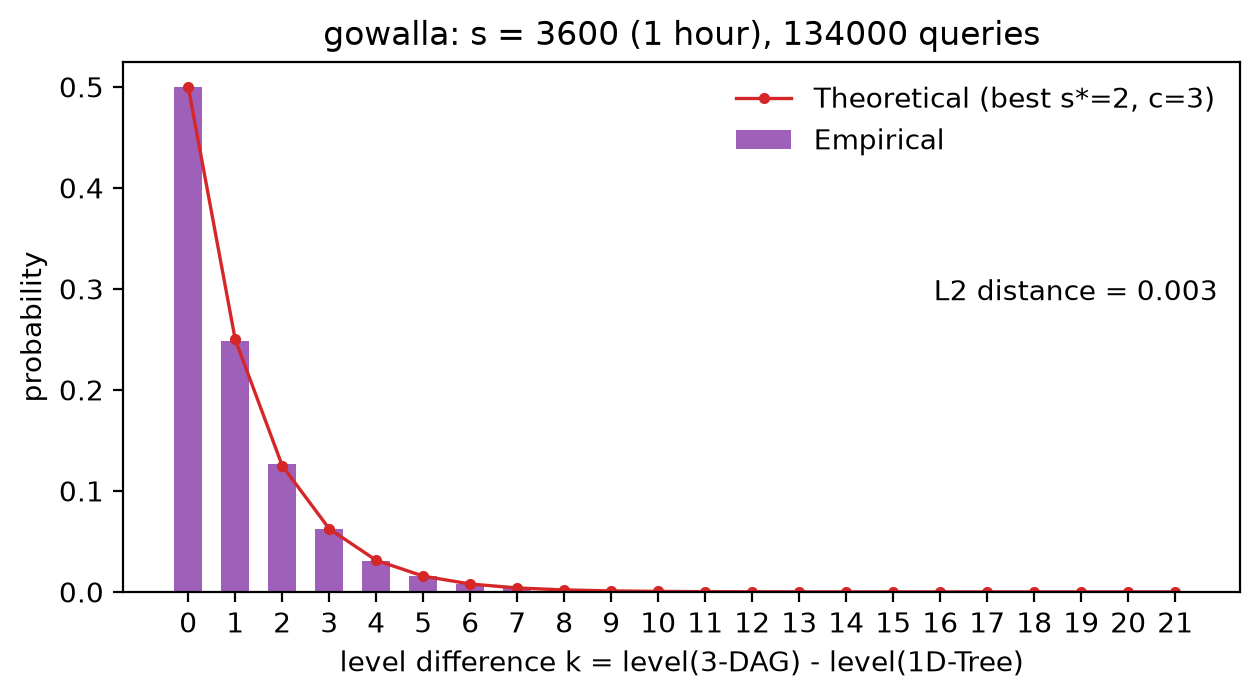}\hfill
\includegraphics[width=0.48\textwidth]{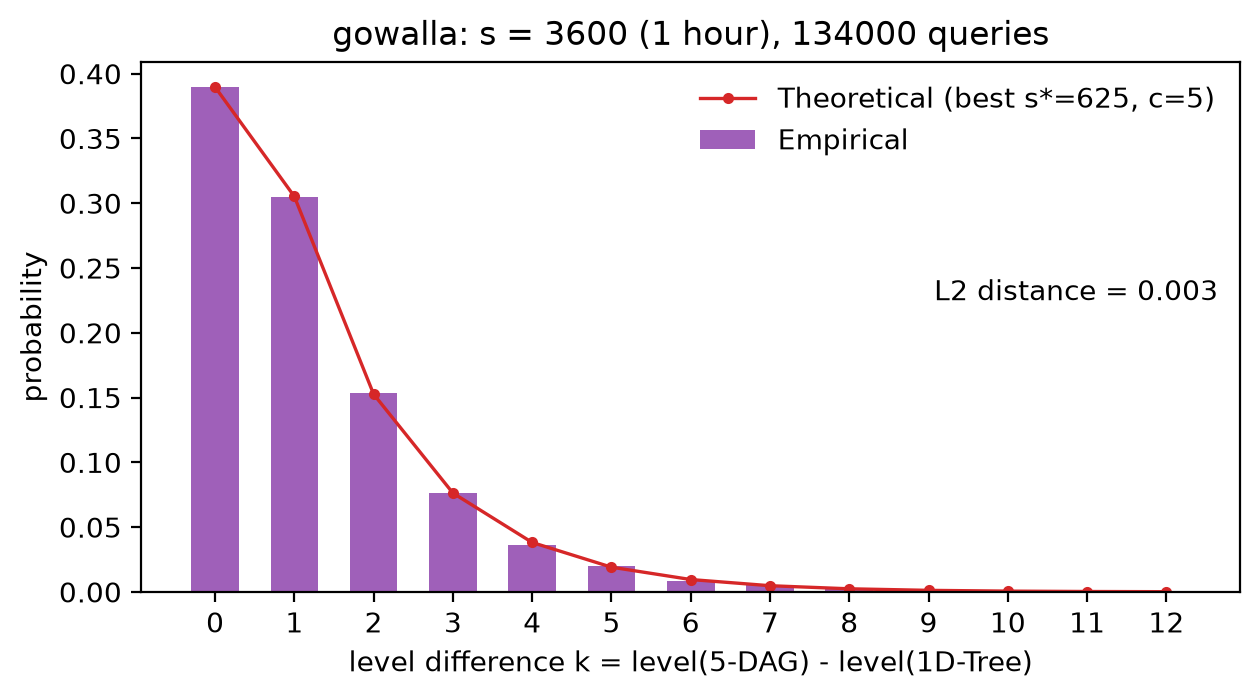}\\[1mm]
\includegraphics[width=0.48\textwidth]{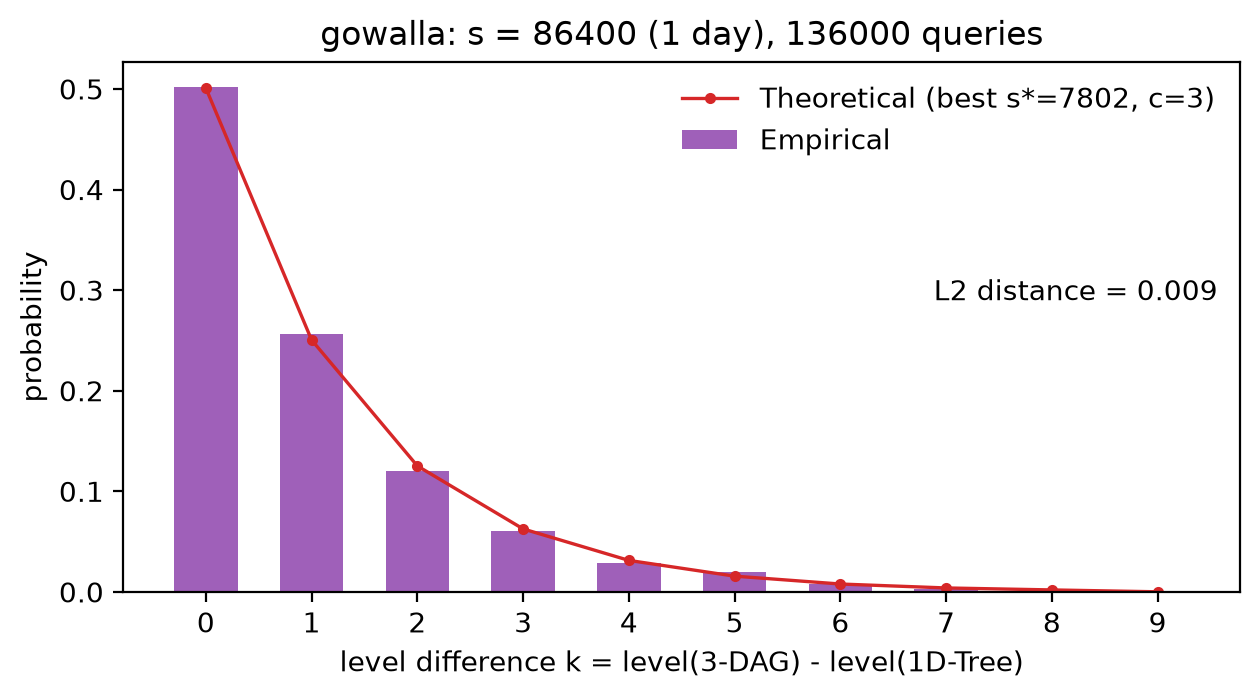}\hfill
\includegraphics[width=0.48\textwidth]{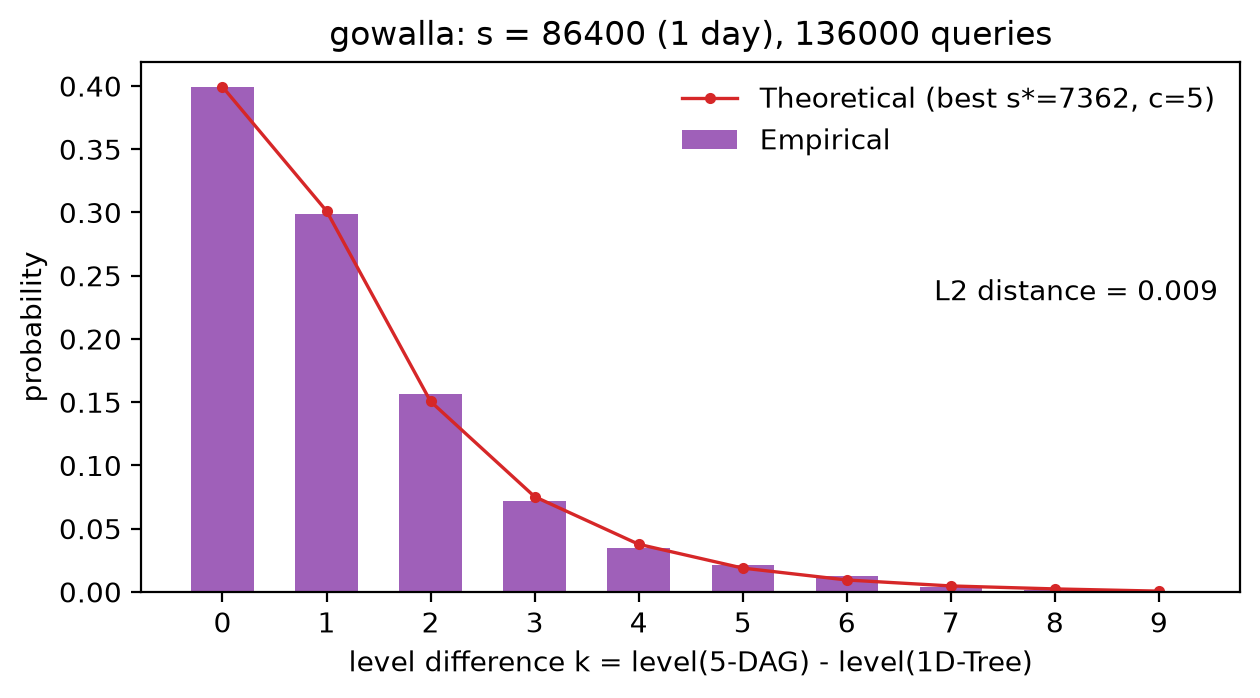}
\caption{Level Overhead Distributions on Gowalla. Top: $s=3600$; bottom: $s=86400$;
left: 3-DAG; right: 5-DAG, each against the 1D-Tree. Purple bars: empirical LOD over
non-empty queries. Red line: LOD of the best-fit baseline $\mathcal{I}_{s^*}$ of
Thm.~\ref{thm:lod} (best $s^*$ and $L_2$ distance given in each panel).}
\label{fig:lod-gowalla}
\end{figure}

\begin{figure}[H]
\centering
\includegraphics[width=0.48\textwidth]{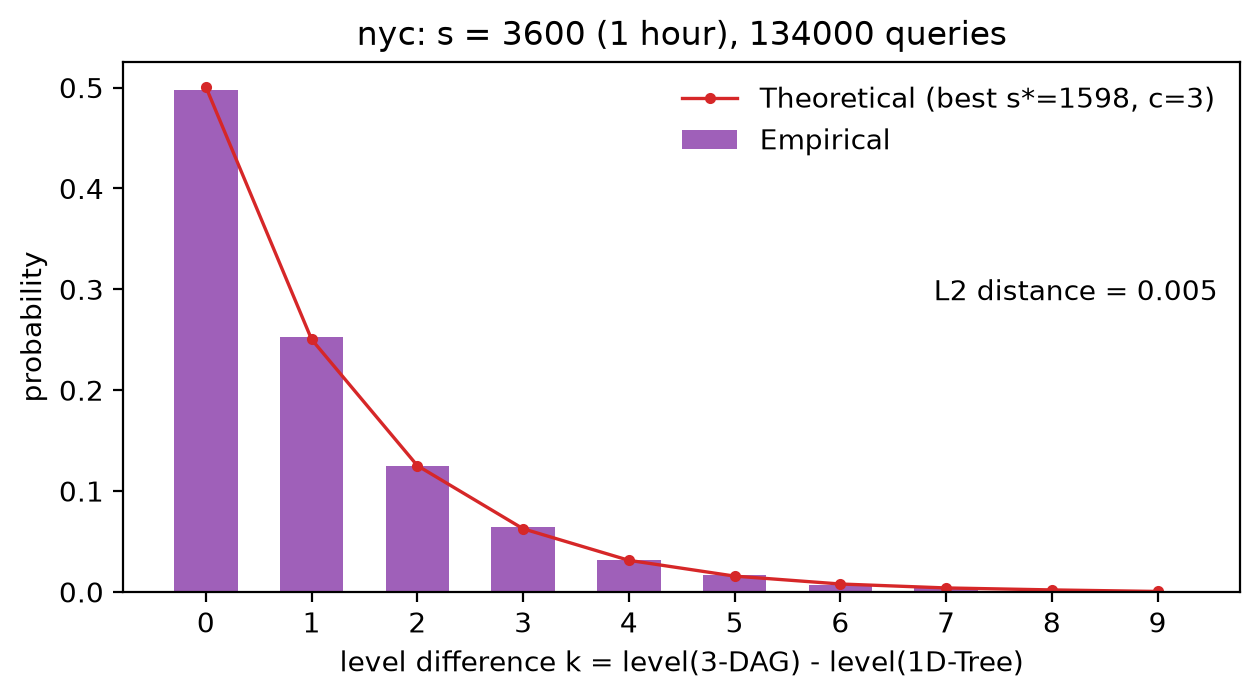}\hfill
\includegraphics[width=0.48\textwidth]{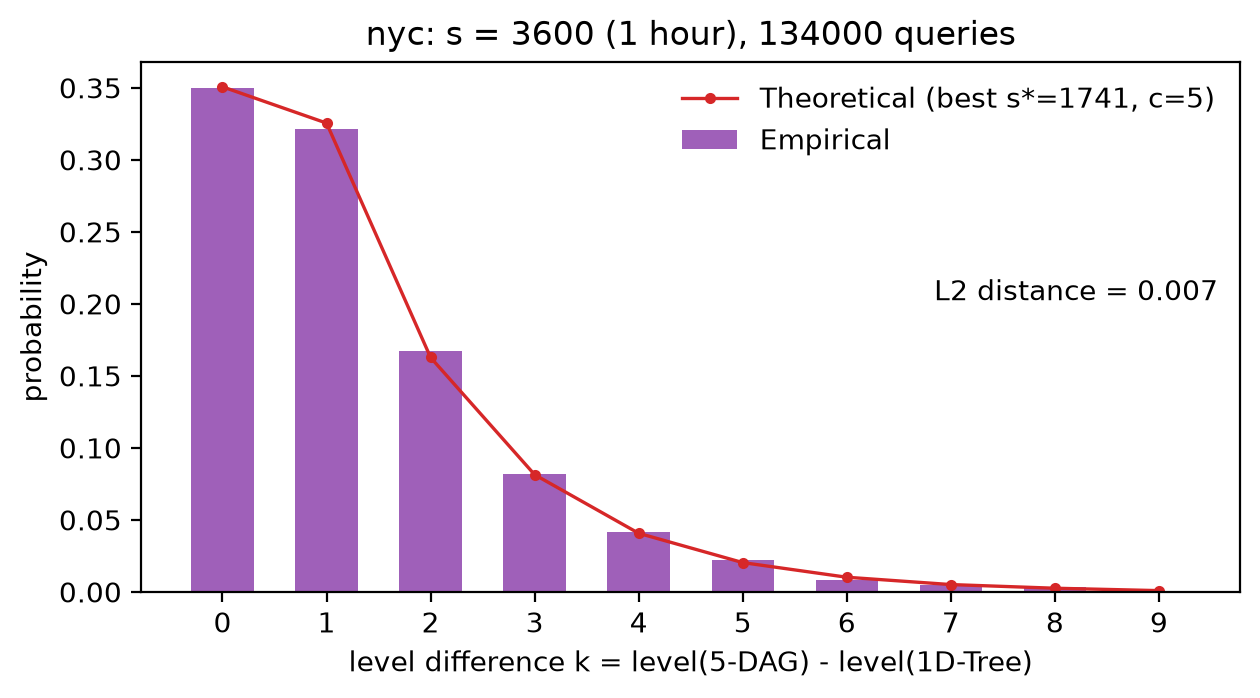}\\[1mm]
\includegraphics[width=0.48\textwidth]{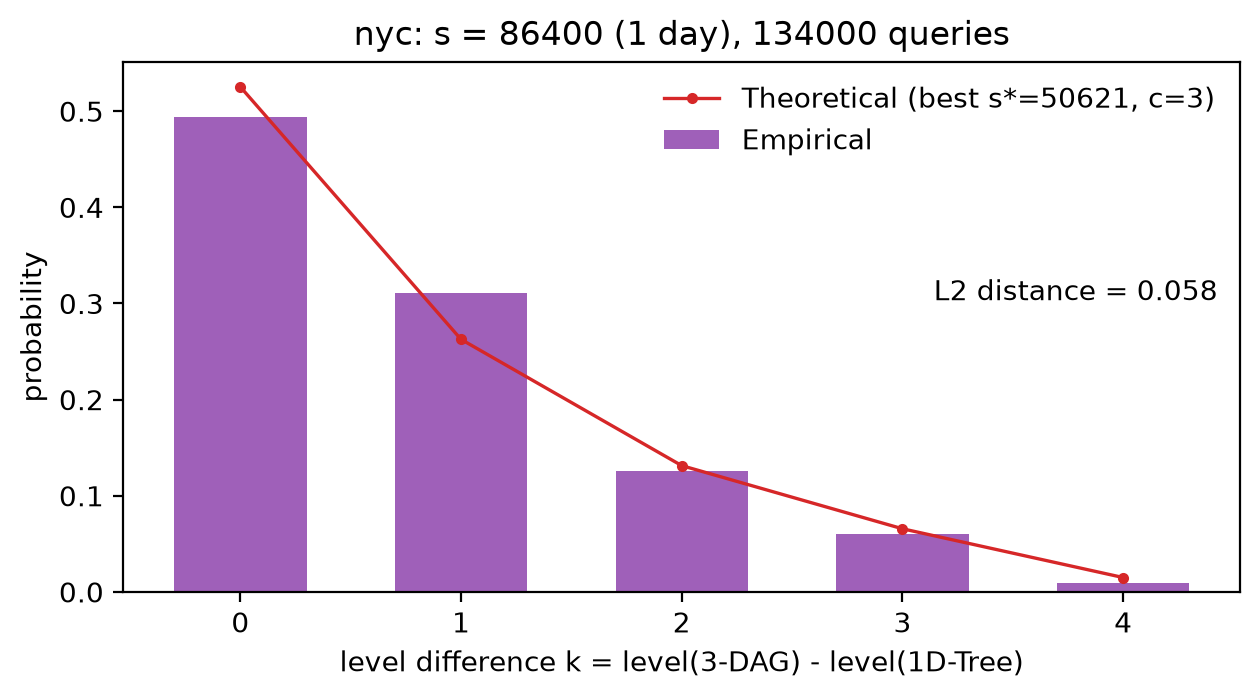}\hfill
\includegraphics[width=0.48\textwidth]{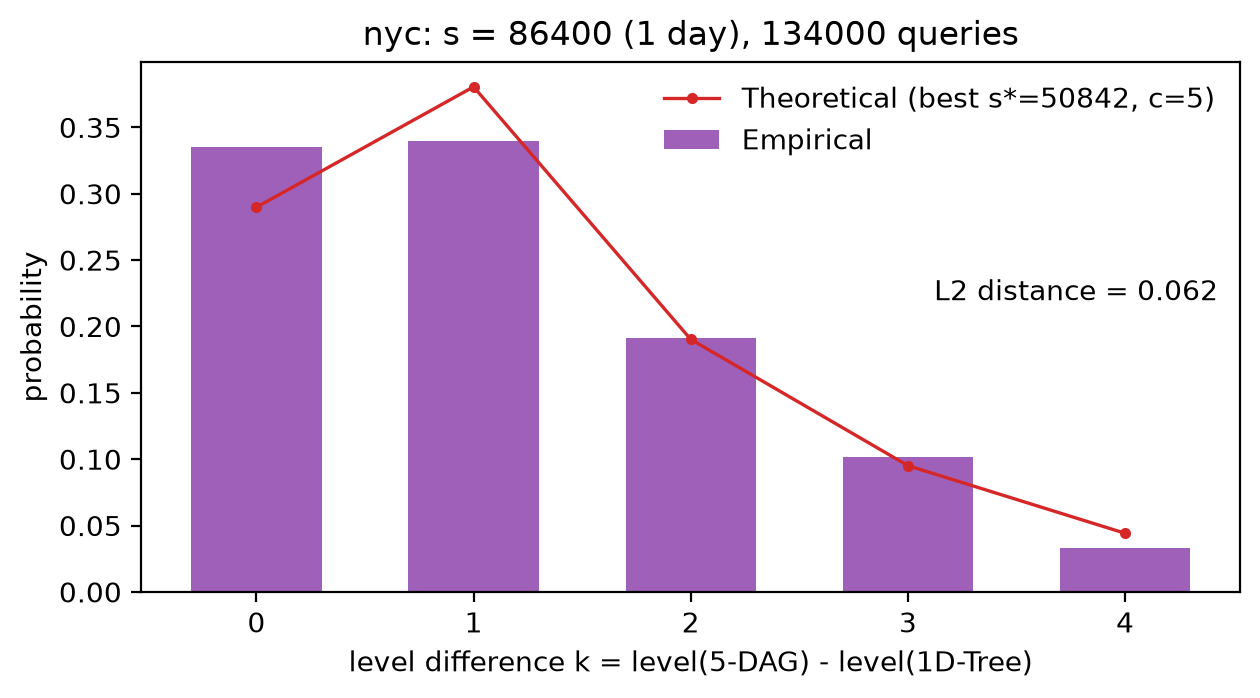}
\caption{Level Overhead Distributions on NYC Yellow Taxi; layout and notation as in
Fig.~\ref{fig:lod-gowalla}.}
\label{fig:lod-nyc}
\end{figure}

\begin{figure}[H]
\centering
\includegraphics[width=0.48\textwidth]{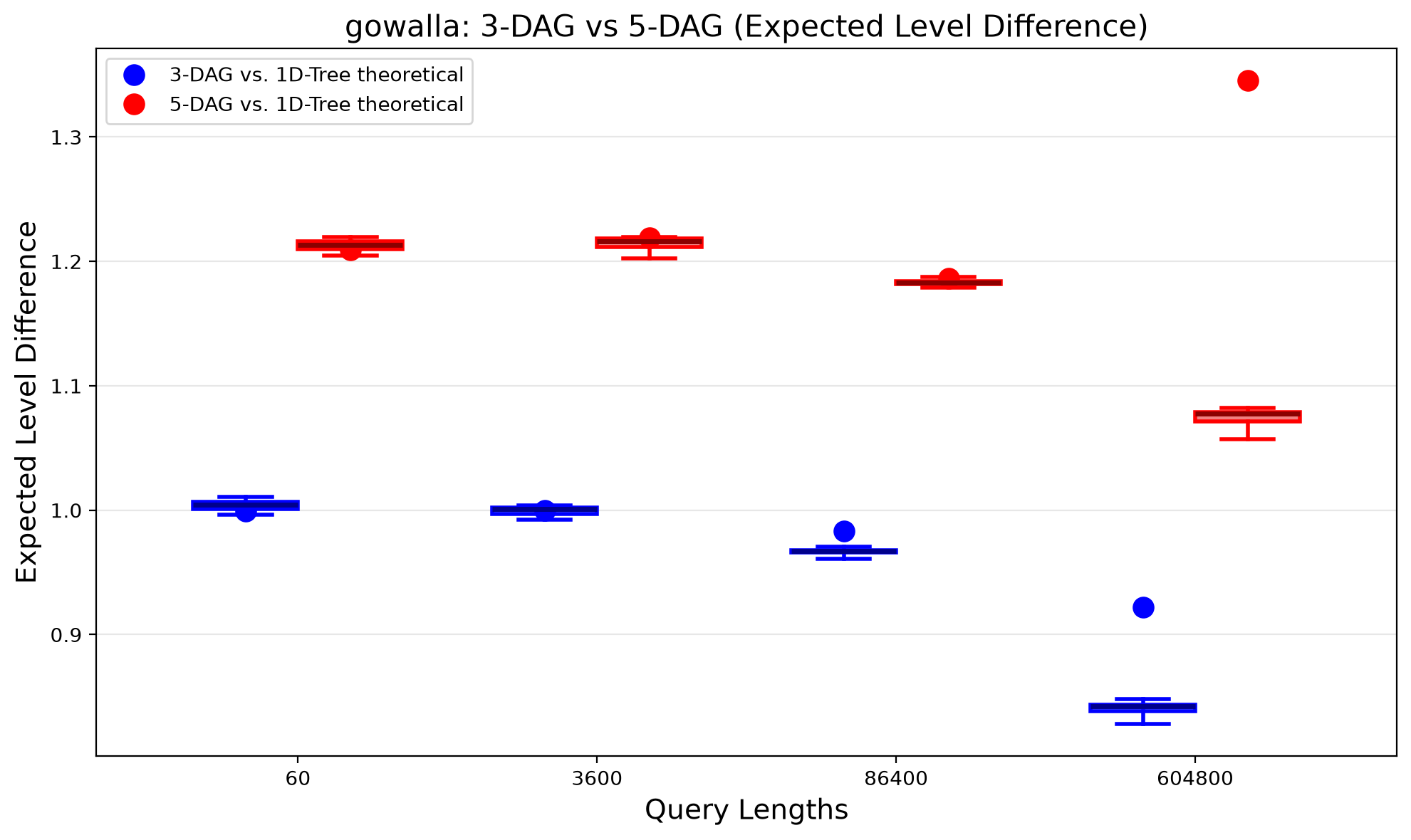}\hfill
\includegraphics[width=0.48\textwidth]{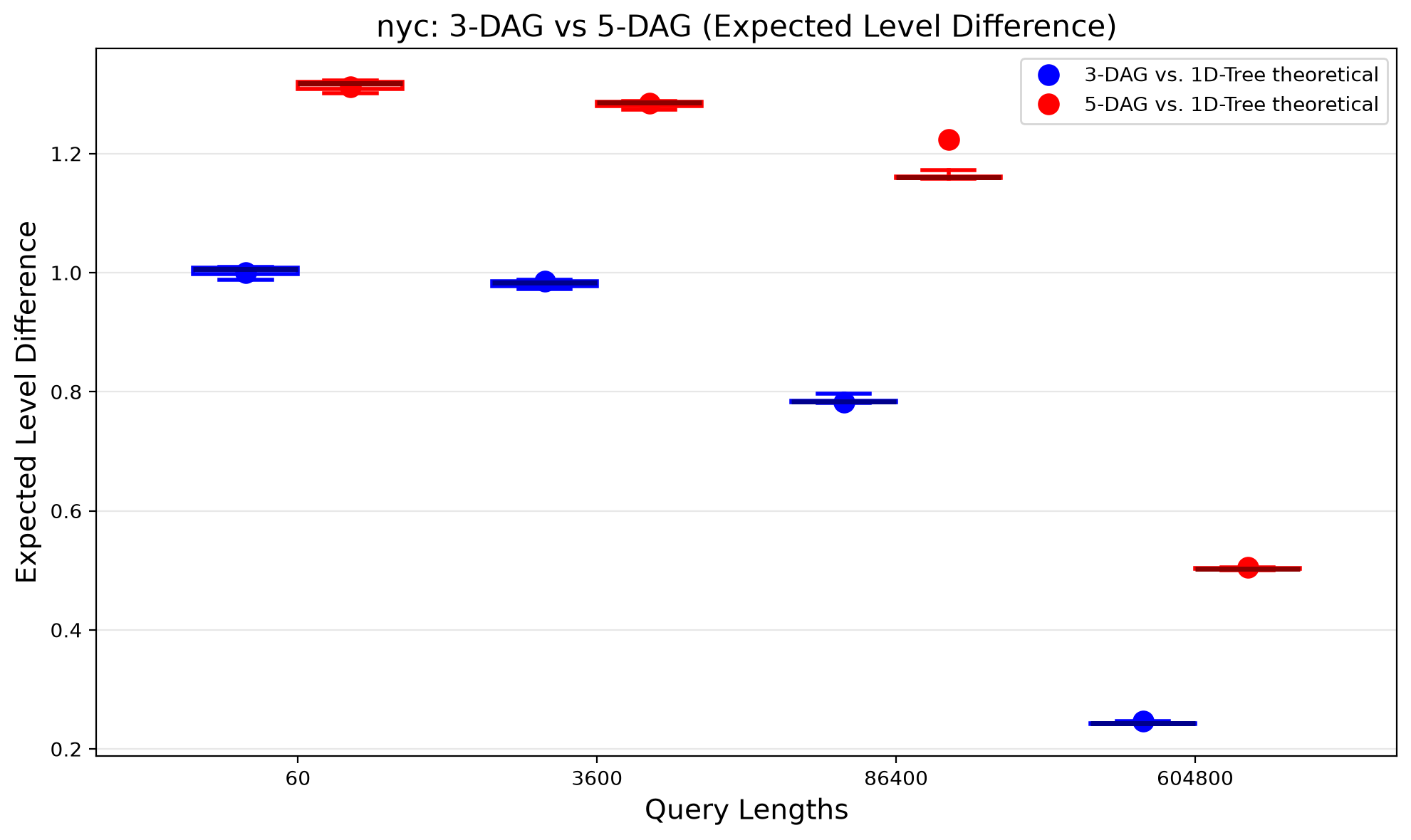}
\caption{Average returned level overhead $\overline{\mathrm{level}_\mathrm{diff}}$ of
each $c$-DAG against the 1D-Tree on Gowalla (left) and NYC (right), for query lengths
of 1 minute, 1 hour, 1 day and 1 week. Boxes: snapshots from stabilization on. Dots:
expectation of the best-fit baseline. Blue: 3-DAG; red: 5-DAG.}
\label{fig:level-diff-whisker}
\end{figure}

\begin{figure}[H]
\centering
\includegraphics[width=0.48\textwidth]{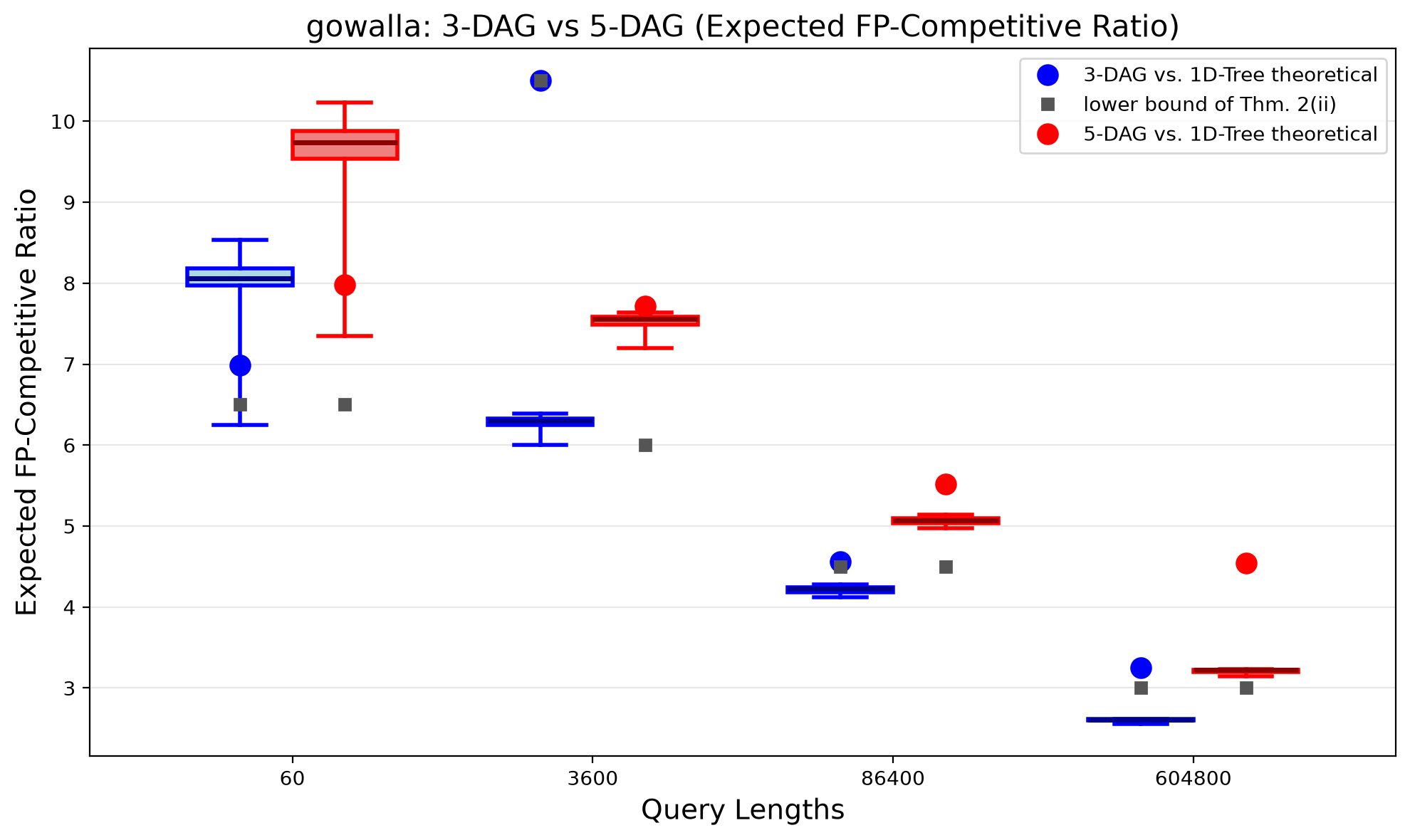}\hfill
\includegraphics[width=0.48\textwidth]{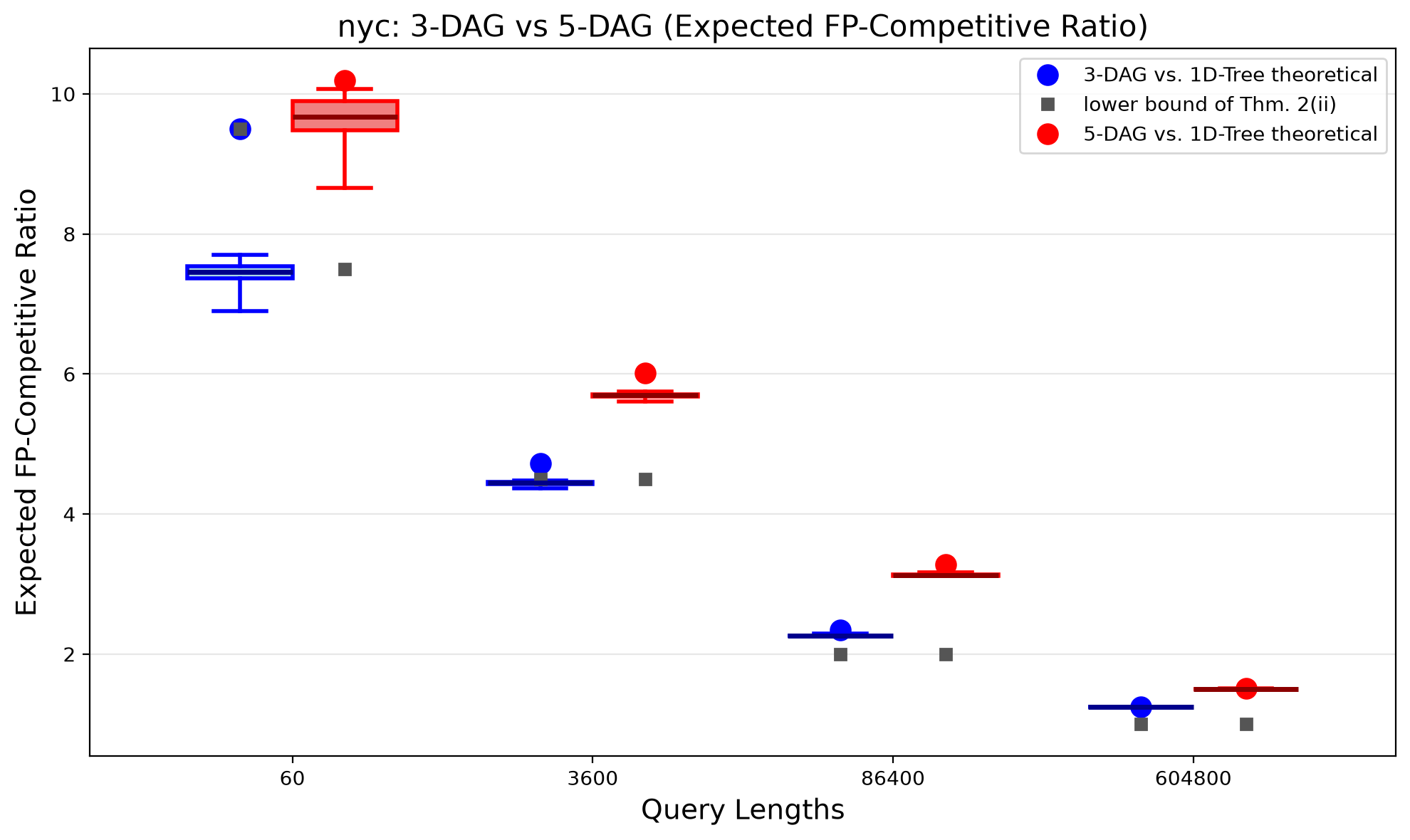}
\caption{Average accuracy ratio $\overline{AR}$ of each $c$-DAG against the 1D-Tree
on Gowalla (left) and NYC (right). Boxes: snapshots from stabilization on. Dots:
$E[2^k]$ under the best-fit baseline. Grey squares: lower bound of Thm.~\ref{thm:stream-uniform}(ii) at the fitted $s^*$. Blue: 3-DAG; red: 5-DAG.}
\label{fig:ar-whisker}
\end{figure}

\begin{figure}[H]
\centering
\includegraphics[width=0.48\textwidth]{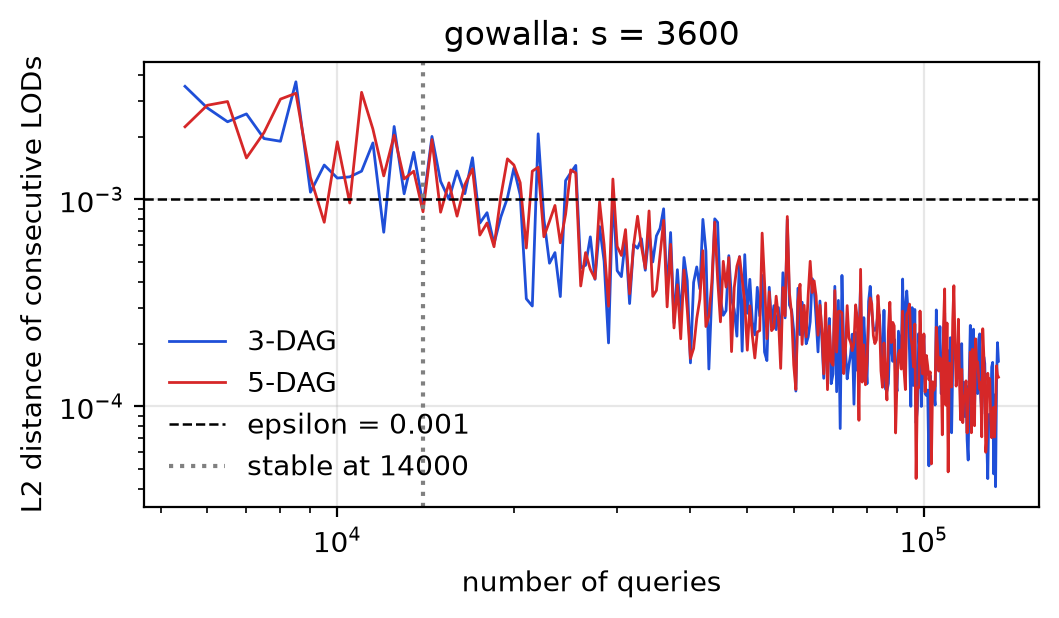}\hfill
\includegraphics[width=0.48\textwidth]{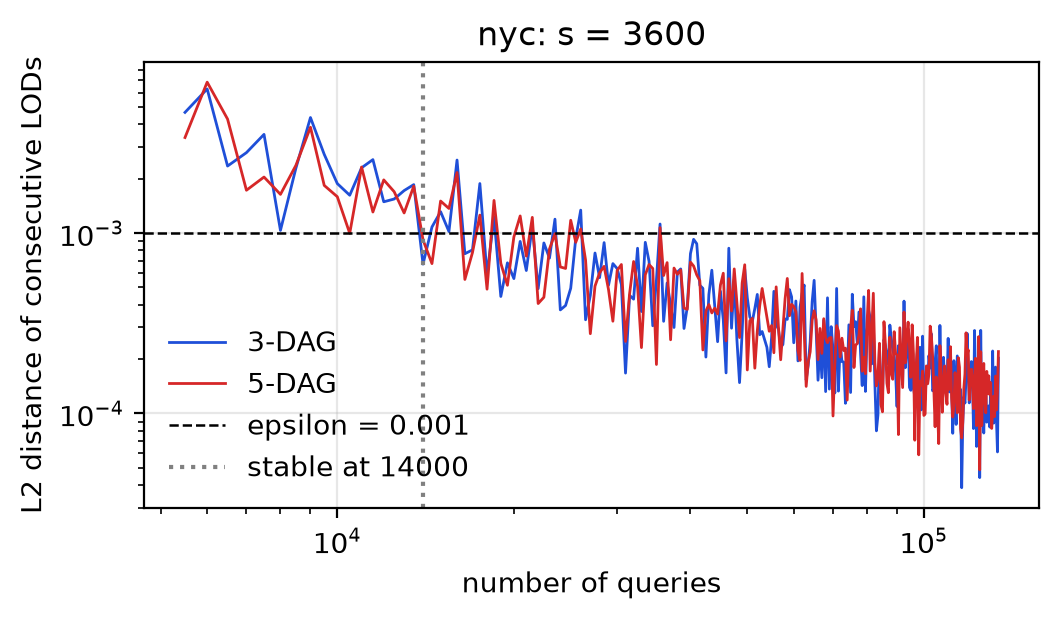}
\caption{$L_2$ distance between the LODs of consecutive snapshots against the number of queries for $s=3600$ on Gowalla (left) and NYC (right). Dashed line:
$\varepsilon=0.001$; dotted line: stabilization point.}
\label{fig:stability}
\end{figure}

\begin{table}[H]
\centering
\footnotesize
\setlength{\tabcolsep}{3pt}
\caption{Gowalla ($N=2^{22}$), non-empty query streams. Columns: query length $s$
(seconds); $c$; number of queries; query count at stabilization; $\overline{\mathrm{level}_\mathrm{diff}}$:
empirical and expectation of the best-fit baseline (the bound of
Thm.~\ref{thm:stream-uniform}(i) is $1$ for $c=3$, $1.5$ for $c=5$); $\overline{AR}$: empirical,
baseline expectation and lower bound of Thm.~\ref{thm:stream-uniform}(ii) at the fitted $s^*$;
median of $AR$; $L_2$ distance of the best-fit baseline LOD. Values are rounded; the
3-DAG overhead is $1.0002$ and $1.0008$ for $s=60$ and $s=3600$.}
\label{tab:gowalla}
\begin{tabular}{rrrrrrrrrrr}
\hline
$s$ & $c$ & queries & stable & \multicolumn{2}{c}{level overhead} & \multicolumn{3}{c}{$\overline{AR}$} & med.\ $AR$ & $L_2$ \\
 &  &  & at & emp. & base. & emp. & base. & l.b. &  &  \\
\hline
60     & 3 & 132{,}500 & 12{,}500 & 1.00 & 1.00 & 7.9 & 7.0  & 6.5  & 1 & 0.001 \\
60     & 5 & 132{,}500 & 12{,}500 & 1.21 & 1.21 & 9.7 & 8.0  & 6.5  & 2 & 0.001 \\
3600   & 3 & 134{,}000 & 14{,}000 & 1.00 & 1.00 & 6.3 & 10.5 & 10.5 & 1 & 0.003 \\
3600   & 5 & 134{,}000 & 14{,}000 & 1.22 & 1.22 & 7.6 & 7.7  & 6.0  & 2 & 0.003 \\
86400  & 3 & 136{,}000 & 16{,}000 & 0.97 & 0.98 & 4.3 & 4.6  & 4.5  & 1 & 0.009 \\
86400  & 5 & 136{,}000 & 16{,}000 & 1.19 & 1.19 & 5.1 & 5.5  & 4.5  & 2 & 0.009 \\
604800 & 3 & 136{,}500 & 16{,}500 & 0.85 & 0.92 & 2.6 & 3.3  & 3.0  & 1 & 0.039 \\
604800 & 5 & 136{,}500 & 16{,}500 & 1.08 & 1.35 & 3.2 & 4.5  & 3.0  & 2 & 0.132 \\
\hline
\end{tabular}
\end{table}

\begin{table}[H]
\centering
\footnotesize
\setlength{\tabcolsep}{3pt}
\caption{NYC Yellow Taxi ($N=2^{20}$), non-empty query streams; columns as in
Table~\ref{tab:gowalla}. The 3-DAG overhead for $s=60$ is $1.0072$.}
\label{tab:nyc}
\begin{tabular}{rrrrrrrrrrr}
\hline
$s$ & $c$ & queries & stable & \multicolumn{2}{c}{level overhead} & \multicolumn{3}{c}{$\overline{AR}$} & med.\ $AR$ & $L_2$ \\
 &  &  & at & emp. & base. & emp. & base. & l.b. &  &  \\
\hline
60     & 3 & 133{,}500 & 13{,}500 & 1.01 & 1.00 & 7.3 & 9.5  & 9.5 & 2 & 0.003 \\
60     & 5 & 133{,}500 & 13{,}500 & 1.32 & 1.31 & 9.5 & 10.2 & 7.5 & 2 & 0.003 \\
3600   & 3 & 134{,}000 & 14{,}000 & 0.99 & 0.99 & 4.5 & 4.7  & 4.5 & 2 & 0.005 \\
3600   & 5 & 134{,}000 & 14{,}000 & 1.29 & 1.28 & 5.7 & 6.0  & 4.5 & 2 & 0.007 \\
86400  & 3 & 134{,}000 & 14{,}000 & 0.78 & 0.78 & 2.3 & 2.3  & 2.0 & 2 & 0.058 \\
86400  & 5 & 134{,}000 & 14{,}000 & 1.16 & 1.22 & 3.1 & 3.3  & 2.0 & 2 & 0.062 \\
604800 & 3 & 128{,}500 & 8{,}500  & 0.24 & 0.25 & 1.2 & 1.2  & 1.0 & 1 & 0.005 \\
604800 & 5 & 128{,}500 & 8{,}500  & 0.50 & 0.51 & 1.5 & 1.5  & 1.0 & 2 & 0.002 \\
\hline
\end{tabular}
\end{table}

\section{Conclusion}
\label{sec:conclusion}
We introduced a competitive framework to formally compare data structures supporting SRC-search with range queries, and used it to compare the data-dependent $c$-DAG family with the 1D-Tree. Under uniform occupancy, for every query length $1<s\le N$ satisfying the non-degeneracy condition, the $c$-DAG returns on average a node at most $2\cdot\frac{c-2}{c-1}$ levels deeper than the 1D-Tree, independently of $N$ and $s$, while
improving the expected accuracy ratio by a factor $\Omega(\log(N/s))$
(Thm.~\ref{thm:stream-uniform}); both bounds extend to mixed query lengths (Cor.~\ref{cor:stream-nonuniform}) and, with an explicit error term, to arbitrary datasets
whose LOD is close to a uniform baseline (Cor.~\ref{cor:arbitrary-datasets}). The price is
this small additive level overhead, a larger structure (for $N$ a power of two, $O(cN)$
nodes and $O(c^2N)$ links), and a higher per-level cost of the exhaustive SRC-search.
We also introduced single-path SRC-search with local selection rules (Sec.~\ref{sec:src_rules}), discussed their privacy implications (Sec.~\ref{sec:privacy}),
and observed on the Gowalla and NYC Yellow Taxi datasets that the level overhead stays at or
below the proven constant, up to sampling error, and that the LOD stays close to a uniform
baseline in most configurations (Sec.~\ref{sec:experiments}).

It would be interesting to apply our methodology to other data structures and search principles, for example returning a constant number of covering nodes instead of one.
Our main technical tool, the level overhead distribution (LOD), determines the average returned level overhead and, for $N$ a power of two, also the average accuracy ratio, since
then $AR(Q)=2^{\mathrm{level}_\mathrm{diff}(Q)}$; how it depends on the configuration of the
data points in non-uniform datasets requires more study, for example by approximating the
empirical LOD with mixtures of baseline LODs instead of a single best-fit baseline.
Other performance and security measures could also be explored within our framework.
Finally, are there other DAG families that are even more efficient than the $c$-DAG in time,
space and accuracy? If yes, what is the trade-off between time complexity, accuracy and
other measures?

\bibliographystyle{plain}
\bibliography{Reference}

\end{document}